\documentclass[a4paper,twocolumn,11pt,accepted=2026-02-17]{quantumarticle}
\pdfoutput=1
\usepackage[utf8]{inputenc}
\usepackage[english]{babel}
\usepackage[T1]{fontenc}
\usepackage{amsmath}

\usepackage[numbers,sort&compress]{natbib}

\usepackage{qcircuit}
\usepackage{graphicx}
\usepackage{siunitx}
\usepackage{amsmath,amssymb,amsthm,mathrsfs,amsfonts,dsfont, pifont}
\usepackage{algorithm}
\usepackage{subfigure, epsfig}
\usepackage{braket}
\usepackage{mleftright}\mleftright
\usepackage{bm}
\usepackage{enumerate}
\usepackage[dvipsnames]{xcolor}
\usepackage{color}
\usepackage{fancybox, graphicx}
\usepackage[skins]{tcolorbox}
\usepackage{multirow}
\usepackage{mathtools}
\usepackage{blkarray}
\usepackage{makecell}
\usepackage{algpseudocode}
\usepackage{adjustbox}

\usepackage{tikz}
\usetikzlibrary{shapes.misc}

\usepackage[colorlinks]{hyperref}
\hypersetup{
	colorlinks=true,
	citecolor = {blue},
    linkcolor=blue,
    filecolor=magenta,      
    urlcolor=cyan,
    final=true
}
\usepackage{cleveref}
\Crefname{equation}{Eq.}{Eqs.}

\usepackage{ifdraft}
\ifdraft{
	\newcommand{\authnote}[3]{{\color{#3} {\bf  \text{#1: }}#2}}
}{
	\newcommand{\authnote}[3]{}
}

\newtheorem{theorem}{Theorem}
\newtheorem{definition}{Definition}
\newtheorem{corollary}{Corollary}[theorem]
\newtheorem{lemma}{Lemma}

\newcommand{\thm}[1]{\hyperref[thm:#1]{Theorem~\ref*{thm:#1}}}
\newcommand{\defn}[1]{\hyperref[defn:#1]{Definition~\ref*{defn:#1}}}
\newcommand{\lem}[1]{\hyperref[lem:#1]{Lemma~\ref*{lem:#1}}}
\newcommand{\prop}[1]{\hyperref[prop:#1]{Proposition~\ref*{prop:#1}}}
\newcommand{\fig}[1]{\hyperref[Fig:#1]{Figure~\ref*{Fig:#1}}}
\newcommand{\tab}[1]{\hyperref[tab:#1]{Table~\ref*{tab:#1}}}
\renewcommand{\sec}[1]{\hyperref[Sec:#1]{Section~\ref*{Sec:#1}}}
\newcommand{\append}[1]{\hyperref[App:#1]{Appendix~\ref*{App:#1}}}
\newcommand{\cor}[1]{\hyperref[cor:#1]{Corollary~\ref*{cor:#1}}}
\newcommand{\obs}[1]{\hyperref[obs:#1]{Observation~\ref*{obs:#1}}}

\newcommand{\cmark}{\ding{51}}%
\newcommand{\xmark}{\ding{55}}%

\definecolor{amethyst}{rgb}{0.6, 0.4, 0.8}

\newcommand{\norm}[1]{\left\lVert#1\right\rVert}
\newcommand{\vertiii}[1]{{\left\vert\kern-0.25ex\left\vert\kern-0.25ex\left\vert #1
		\right\vert\kern-0.25ex\right\vert\kern-0.25ex\right\vert}}

\newcommand{\bigO}[1]{\mathcal{O}\left( #1 \right)}
\newcommand{\bigOt}[1]{\widetilde{\mathcal{O}}\left( #1 \right)}	
\newcommand{\nrm}[1]{\left\| #1 \right\|}
\newcommand{\eps}{\varepsilon}
\newcommand{\Tr}{\mbox{\rm Tr}}
\newcommand{\tr}[1]{\Tr\left(#1\right)}
\DeclarePairedDelimiterX\ketbra[2]{| }{|}{#1 \delimsize\rangle\!\delimsize\langle #2}
\makeatletter
\let\pra\@undefined                  
\let\prb\@undefined                  
\let\pre\@undefined                  
\let\prl\@undefined                  
\let\rmp\@undefined                  
\let\jmp\@undefined                  
\let\jmo\@undefined                  
\makeatother

\allowdisplaybreaks

\usepackage{titlesec}
\begin{document}

\title{A streamlined quantum algorithm for topological data analysis with exponentially fewer qubits}

\author{Sam McArdle}
\affiliation{AWS Center for Quantum Computing, Pasadena, CA 91125, USA}

\author{Andr\'{a}s Gily\'{e}n}
\affiliation{Alfr\'{e}d R\'{e}nyi Institute of Mathematics, Budapest, Hungary}

\author{Mario Berta}
\affiliation{Institute for Quantum Information, RWTH Aachen University, Aachen, Germany}
\affiliation{Department of Computing, Imperial College London, London, UK}

\date{30/3/26}

\begin{abstract}
Topological invariants of a dataset, such as the number of holes that survive from one length scale to another (persistent Betti numbers) can be used to analyze and classify data in machine learning applications. We present an improved quantum algorithm for computing persistent Betti numbers, and provide an end-to-end complexity analysis. Our approach provides large polynomial time improvements, and an exponential space saving, over existing quantum algorithms. Subject to gap dependencies, our algorithm obtains an almost quintic speedup in the number of datapoints over previously known rigorous classical algorithms for computing the persistent Betti numbers to constant additive error -- the salient task for applications. However, we also introduce a quantum-inspired classical power method with provable scaling only quadratically worse than the quantum algorithm. This gives a provable classical algorithm with scaling comparable to existing classical heuristics, subject to assumptions on the gap scaling. We discuss whether quantum algorithms can achieve an exponential speedup for tasks of practical interest, as claimed previously. We conclude that there is currently no evidence for this being the case.
\end{abstract}

\maketitle

\section{Introduction}\label{Sec:Intro}
Topological data analysis (TDA) is a method for extracting insights from large, high-dimensional, and noisy datasets~\cite{carlsson2020topologicalreview}. TDA extends methods from algebraic topology (using group theory to classify topological objects) to datapoints sampled from an underlying topological manifold. Topological features (such as the number of connected components, or holes of a given dimension) can be good identifiers for a dataset, as they are often robust to noise and perturbations accrued when collecting the data, unlike geometric features. Topological features may be used directly to interpret the dataset -- for example, to find holes in the coverage of sensors in a network~\cite{de2007coverage}, identify isolated areas in voter preference data~\cite{feng2021persistent}, or classify structure in the distribution of matter in the universe~\cite{pranav2017topology}. They can also be used as general identifiers to compare different datasets in situations where the topology does not necessarily have an obvious interpretation -- such as identifying participant groups in functional magnetic resonance imaging brain scans~\cite{rieck2020uncovering}, or detecting patterns in time series data~\cite{perea2015sliding} to predict financial crashes~\cite{gidea2018financialcrash}. Using topological features as input to machine learning models is growing increasingly popular within physics~\cite{leykam2022TDAphysicsreview} (such as the unsupervised detection of phase transitions in lattice models~\cite{sale2022latticemodelTDA,tirelli2021hubbardmodelTDA}) and other areas of science~\cite{hensel2021survey}. 

Classical algorithms for TDA proceed
through linear algebra. An approximation of the manifold is constructed from a set of simplices (vertices, line segments, triangles, and higher dimensional generalizations) built by connecting together points within a certain length scale of each other. We are interested in computing the persistent Betti number $\beta_k^{i,j}$, the number of $k$-dimensional holes that survive from scale $i$ to scale $j$, for a range of dimensions and length scales. The number of $k$-dimensional holes at a single scale is known as the Betti number $\beta_k^i$, and is a less informative quantity~\cite{neumann2019limitations}. A dataset consisting of $N$ points can contain $\binom{N}{k+1}$ $k$-dimensional simplices, which we represent by orthonormal basis vectors. Classical algorithms for $\beta_k^{i,j}$ (and $\beta_k^i$) scale polynomially in time and space with the size of the simplicial complex, which can make them too costly to apply to systems with many datapoints, or for high $k$. However, large $k$ values are likely also less prevalent in applications due to their reduced interpretability~\cite{leykam2022TDAphysicsreview}. 

\begin{table*}[!t]
    \centering
    \renewcommand{\arraystretch}{2.5}
    \begin{adjustbox}{max width=\textwidth}
    \begin{tabular}{|c|c|c|c|c|c|c|}
    \hline
         & Algorithm & $\beta_k^i$ & $\beta_k^{i,j}$ & \makecell{Operator encoding \\ the topology} & Gate count & \makecell{Space complexity} \\ \hline
        \multirow{6}{*}{Quantum} & LGZ~\cite{lloyd2016quantum} & \cmark  & \cmark\kern-2mm\xmark
        \footnote{It is suggested in Ref.~\cite{web:LloydTalk} that this approach can be used for computing persistent Betti numbers by exploiting the quantum Zeno effect. We have been unable to reproduce this line of reasoning, as we discuss in Appendix~\ref{AppSubSub:Zeno}.}
        & \makecell{Hermitian embedding \\ of boundary operators \\ $\partial_k^i$, $\partial_{k+1}^i$} & $\mathcal{O}\left( \frac{LN^3 \beta_k^i \binom{N}{k+1}}{\Delta^2   \mathrm{Min}\left( \Lambda_{\partial_k^i}, \Lambda_{\partial_{k+1}^i} \right)}  \right)$ & $\mathcal{O}(N^2)$ \\ \cline{2-7}
        & GK~\cite{gunn2019review} & \cmark  & \xmark & \makecell{$\bigoplus_k \partial_k + \partial_k^\dag $} & $\tilde{\mathcal{O}}\left( \frac{LN^2 k \sqrt{\beta_k^i \binom{N}{k+1}}}{\Delta \Lambda}  \right)$ & $\mathcal{O}(N)$ \\ \cline{2-7}
        & UAS+~\cite{ubaru2021quantum} & \cmark  & \xmark & \makecell{Combinatorial \\ Laplacian} & $\mathcal{O}\left( \frac{LN (\beta_k^i)^{1.5} \binom{N}{k+1}^{1.5}}{\Delta^3 \Lambda} \right)$\footnote{One gets this scaling via substituting $|S_k^i|$ by the upper bound $\binom{N}{k+1}$. } 
        & $\mathcal{O}(N)$ \\ \cline{2-7}
        & Hay~\cite{hayakawa2021quantum} & \cmark & \cmark & \makecell{Persistent \\ combinatorial \\ Laplacian} & $\tilde{\mathcal{O}}\left( \frac{L N^8 k^4  \beta_k^{i,j} \binom{N}{k+1}}{\Delta^2 \Lambda_1^2 \Lambda_2}  \right)$ & $\mathcal{O}(N)$ \\ \cline{2-7}
        & AMS~\cite{ameneyro2022quantum} & \cmark & \cmark\kern-2mm\xmark\footnote{We have been unable to verify the correctness of this approach, due to nuances associated with changing the basis from simplices in the complex at scale $i$ to those at scale $j$. We discuss this in more detail in Appendix~\ref{AppSubSub:RestrictedChainGroup}.} & \makecell{Hermitian embedding \\ of boundary operators} & Unspecified & $\mathcal{O}(N^2)$ \\ \cline{2-7}
        & BSG+~\cite{berry2022quantifying} & \cmark & \xmark & \makecell{Hermitian embedding \\ of boundary operators}  &   $\tilde{\mathcal{O}}\left( \frac{LN^{2} \sqrt{\beta_k^{i} \binom{N}{k+1}}}{\Delta   \Lambda}  \right)$ & $\mathcal{O}(N\log(N))$ \\ \cline{2-7}
        & This work & \cmark & \cmark & $\partial_k^i, \partial_{k+1}^j$ &   $\tilde{\mathcal{O}}\left( \frac{L(Nk)^{3/2} \sqrt{ \beta_k^{i,j} \binom{N}{k+1}}}{\Delta   \Lambda_{\Pi \Pi}^{0.5} \mathrm{Min}\left( \Lambda_{\partial_k^i}, \Lambda_{\partial_{k+1}^j} \right)}  \right)$ \footnote{We have converted the gate depth obtained for our method to a gate count by multiplying by the spatial complexity $k\log(N)$.} & $\mathcal{O}(k\log(N))$ \\ \hline \hline
        \multirow{4}{*}{Classical} & Textbook~\cite{Edelsbrunner2002} & \cmark  & \cmark & \makecell{Boundary \\ operator} & $\mathcal{O}\left( |S_{k+1}^j|^3\right)$ & $\mathcal{O}(|S_{k+1}^j|^2)$ \\ \cline{2-7}
        & \makecell{Optimised\\ textbook~\cite{milosavljevic2011zigzag,milosavljevic2010:inriaReport}} & \cmark  & \cmark & \makecell{Boundary \\ operator} & $\mathcal{O}\left( |S_{k+1}^j|^\omega \right)$ & $\mathcal{O}(|S_{k+1}^j|^2)$ \\ \cline{2-7}
        & \makecell{Heuristic \\ sparsification~\cite{mischaikow2013morse}} & \cmark  & \cmark & \makecell{Boundary \\ operator} & $\mathcal{O}\left(|S_{k+1}^j| + |{\bar{S}_{k+1}^j}|^\omega \right)$ & $\mathcal{O}\left(\mathrm{Max}(|S_{k+1}^j|, |{\bar{S}_{k+1}^j}|^2)\right)$ \\ \cline{2-7}
        & \makecell{Power \\ method~\cite{friedman1998computing}} & \cmark  & \xmark & \makecell{Combinatorial \\ Laplacian} & $\tilde{\mathcal{O}}\left(\frac{L|S_k^i| (k^2 \beta_k^i + k (\beta_k^i)^2)}{\Lambda}  \right)$ & $\mathcal{O}\big{(}|S_k^i|(k + \beta_k^i) \big{)}$ \\ \cline{2-7}
        & \makecell{Quantum-inspired \\ power method \\ (This work)} & \cmark  & \cmark & \makecell{Boundary \\ Operator} & $\mathcal{O} \left( L\beta_k^{i,j} N S \frac{\log^2\left(\frac{1}{\epsilon}\right)}{\mathrm{Min}\left(\Lambda_{\partial_k^i}, \Lambda_{\partial_{k+1}^j}\right) \Lambda_{\Pi \Pi}^{0.5}} + S (\beta_k^{i,j})^2 + (\beta_k^{i,j})^3 \right)$ & $\mathcal{O}(S(N+\beta_k^{i,j}))$ \\ \cline{1-7}
    \end{tabular}
    \end{adjustbox}
    \caption{A comparison of quantum and classical algorithms for topological data analysis. We present the complexities to estimate the quantities $\beta_k^i, \beta_k^{i,j}$ to additive error $\Delta$, at $L$ different (pairs of) length scales. $|S_k^i|$ denotes the number of $k$-simplices in the complex at scale $i$. For classical algorithms, we do not include the cost $N \sum_{j=0}^k |S_j^i|$ of finding the $k$-simplices in the complex at scale $i$, as this is typically subleading for small $k$, and is usually ignored in existing classical analyses. We also do not explicitly track logarithmic factors in the classical complexities, associated with updating simplices, etc. The complexities in Refs.~\cite{lloyd2016quantum,hayakawa2021quantum,ameneyro2022quantum} were not explicitly stated in terms of the parameters that we use in this paper, hence the above bounds are our best estimates based on the subroutines used in those works. Initial quantum algorithms were only able to compute regular Betti numbers, not persistent Betti numbers. For the quantum algorithms we present the complexity for the regime which we find practically most relevant, where $k \ll N$, and $|S_k^i|$ is approximately $\binom{N}{k+1}$. The quantity $\Lambda$ denotes the smallest non-zero singular value of the matrix used for encoding the topology ($\partial_k^i, \partial_{k+1}^i$, or the combinatorial Laplacian). The gap $\Lambda_{\Pi\Pi}$ appearing in our work is defined in Sec.~\ref{Subsec:SubspaceProj}, and equals $1$ when computing the non-persistent $\beta_k^i$ values.
   	The quantity $\bar{S}_k^i$ denotes the number of $k$-simplices in a sparsified complex, and $\omega \in[2,2.372)$ is the matrix multiplication exponent. We define the parameter $S := |S_{k-1}^i| + |S_{k}^j| + |S_{k+1}^j|$
	} \label{tab:PriorAlgorithmComp}
	\renewcommand{\arraystretch}{1}
\end{table*}

A quantum algorithm for computing Betti numbers has been proposed~\cite{lloyd2016quantum}, experimentally demonstrated as a proof-of-principle~\cite{huang2018demonstration, akhalwaya2022TowardsNisqTDA}, refined~\cite{gunn2019review, gyurik2020towards, ubaru2021quantum}, and recently extended to the more informative persistent Betti numbers~\cite{hayakawa2021quantum} -- with claims of an exponential speedup over classical algorithms. The quantum algorithm exploits the ability of $N$ qubits to represent and perform linear algebra in a vector space of dimension $2^N$, in some cases with cost polynomial in $N$. However, the quantum algorithm actually returns the (persistent) Betti number, normalized by the number of $k$-simplices in the complex. While tasks closely related to normalized Betti number estimation have been shown to be DQC1-hard~\cite{gyurik2020towards,cade2021complexity} (note that the task of estimating normalized (persistent) Betti numbers of typically considered clique complexes has not yet been shown to be DQC1-hard~\cite{cade2021complexity}), there is currently no evidence for a regime with practical use cases. When using quantum algorithms to compute the salient quantity $\beta_k^{i,j}$, the exponential advantage is lost (a similar issue afflicts quantum algorithms for approximating the Jones polynomial~\cite{aharonov2006polynomial, shor2007estimating} and other \#P problems~\cite{bordewich2009approximate}). Therefore, quantum algorithms for TDA should be optimised, if they are to achieve significant polynomial advantage over classical approaches for applications of interest.

\section{Overview of results}\label{Subsec:ResultsOverview}

We present and analyze a streamlined quantum algorithm for estimating persistent Betti numbers. 
We reduce the problem to estimating the normalized rank of a projector that encodes the relevant topological features of the data, and show how this problem can be efficiently solved using quantum singular value transformation (QSVT)~\cite{gilyen2019quantum}. QSVT provides efficient means for applying a polynomial function to the singular values of a matrix given by a unitary block encoding. 

We show how to construct block encodings of operators that represent the topology of a simplicial complex, and how to use QSVT to transform these into projectors onto the relevant subspaces.\footnote{We were motivated by Ref.~\cite{hayakawa2021quantum}, which also used QSVT for computing persistent Betti numbers. Our approach is similar at a high level, but we work with different and more elementary operators to encode the topology. We provide an explicit comparison between our approach and the algorithm of Ref.~\cite{hayakawa2021quantum} in Sec.~\ref{Subsec:QuantumComplexity}.} We consider two possible ways of mapping simplices to qubits; a direct approach (taken by previous quantum algorithms) that uses $N$ qubits for $N$ datapoints, and a compact approach, introduced herein, that uses $(k+1)\log(N)$ qubits to store a $k$-simplex state. The compact mapping provides an exponential space saving over classical and quantum algorithms for $k = \bigO{\mathrm{polylog}(N)}$.

For a dense simplicial complex constructed from $N$ points in $\mathbb{R}^d$, such that the number of $k$-simplices scales roughly as $\binom{N}{k+1}$, we have the following main result:\footnote{As noted in Appendix~\ref{AppSec:OverallComplexity} the following complexity is achieved when $k(\log(d)\log(b) + b) < N$, where $b$ is the number of bits required to describe the coordinates.} (see Appendix~\ref{AppSec:OverallComplexity} for details)

\begin{theorem}\label{theorem:overallcomplexity}
(Persistent Betti number estimation) For the above parameters, we can estimate $\beta_k^{i,j}$ to additive error $\Delta$, with success probability greater than $(1 - \eta)$ using $\bigO{\log\left(\frac{1}{\eta} \right)}$ repetitions of a quantum circuit that acts on $\bigO{\min\{k\log(N), N\}}$ qubits, and has depth
\begin{equation}
    \tilde{\mathcal{O}}\left( \frac{N^{3/2} \sqrt{k \beta_k^{i,j} \binom{N}{k+1}}}{\Delta   \Lambda_{\Pi\Pi}^{0.5} \mathrm{Min}\left(\Lambda_{\partial_k^i}, \Lambda_{\partial_{k+1}^j} \right)} \right),
\end{equation}
(if using the circuit acting on $N$ qubits, then the depth is reduced by a factor of $\sqrt{k}$) where $\Lambda_{\partial_k^i}, \Lambda_{\partial_{k+1}^j}$ denote the size of the smallest non-zero singular value of the specified boundary operators, and $\Lambda_{\Pi\Pi}$ is another gap parameter defined in Sec.~\ref{Subsec:SubspaceProj}, which has value $1$ if $i=j$.
\end{theorem}

We compare our approach to existing quantum algorithms in Table~\ref{tab:PriorAlgorithmComp}. Only the more recent approaches have been able to compute both regular and persistent Betti numbers -- the latter being the key quantity of interest for topological data analysis~\cite{neumann2019limitations}. Prior works took more circuitous routes to computing (persistent) Betti numbers -- either by implementing the relevant projections in less efficient ways (e.g., quantum phase estimation)~\cite{lloyd2016quantum, gunn2019review, ameneyro2022quantum}, encoding the topology in more complex operators~\cite{lloyd2016quantum, gunn2019review, hayakawa2021quantum, ameneyro2022quantum}, or using incoherent, asymptotically less efficient approaches to estimate subspace dimensions~\cite{lloyd2016quantum, ubaru2021quantum, hayakawa2021quantum, ameneyro2022quantum, akhalwaya2022TowardsNisqTDA}. In particular, compared to the existing quantum algorithm~\cite{hayakawa2021quantum} for computing $\beta_k^{i,j}$, our approach provides an exponential reduction in the number of qubits required for $k = \mathcal{O} \left( \mathrm{polylog}(N) \right)$, as well as a time complexity improvement of $\mathcal{O} \left( N^{6.5} \sqrt{\binom{N}{k+1}}\Delta^{-1} \right)$ (assuming similar overheads from the respective gaps $\Lambda$ of the different operators used for encoding the topology).

For dense complexes, our quantum algorithm achieves a polynomial speedup over classical algorithms for the practically relevant task of computing the actual $\beta_k^{i,j}$ values, subject to the dependence of the gap parameters on $N$. The speedup could be bigger for larger values of $k$, but is less prominent when compared to heuristic methods that sparsify the complex. We additionally present a quantum-inspired classical algorithm based on the power method, which scales only quadratically worse in the number of simplices than our quantum algorithm, with a similar gap dependence. Hence, we expect no better than a quadratic speedup for our quantum algorithm, for the task of computing the persistent Betti number to additive error. Our quantum algorithm achieves an exponential space saving compared to classical approaches for the practically relevant case of dense complexes in the small $k$ regime. An extended comparison between our quantum algorithm and classical approaches can be found in Sec.~\ref{Subsec:QuantumSpeedups}.

The rest of this manuscript is laid out as follows. In Sec.~\ref{Sec:MathsBackground} we provide a self-contained introduction to the aspects of topological data analysis required to understand our quantum algorithm (we provide a more extended and pedagogical introduction in Appendix~\ref{AppSec:MathsBackgroundPedagogy}). We then introduce the quantum algorithm used for computing persistent Betti numbers in Sec.~\ref{Sec:QuantumAlgorithm}. In Sec.~\ref{Sec:ResourceEstimates} we outline the costs of the individual building blocks of our quantum algorithm, with detailed discussion in the Appendices. In Sec.~\ref{Sec:Classical} we review existing classical algorithms for computing persistent Betti numbers, and introduce our quantum-inspired classical algorithm that scales only quadratically worse than the quantum algorithm. In Sec.~\ref{Sec:Discussion} we discuss the complexity of our quantum algorithm, and the prospects for quantum speedups for practical problems of interest. We conclude in Sec.~\ref{Sec:Conclusion}.

\subsection{Recent developments}
Shortly after the first version of this work was released on arXiv, a number of related results were released on arXiv: 
\begin{itemize}
    \item Refs.~\cite{crichigno2022clique,schmidhuber2022complexity,king2023promise} investigated TDA from a complexity theoretic standpoint. Ref.~\cite{schmidhuber2022complexity} showed that determining if the Betti number of a (clique-dense) clique complex is non-zero is NP-hard in general. Refs.~\cite{crichigno2022clique,king2023promise} proved a stronger result, showing that this problem is QMA$_1$-hard. These works rigorously show that quantum algorithms should not be expected to provide exponential speedups for (persistent) Betti number estimation.
    \item Ref.~\cite{berry2022quantifying} introduced optimizations of the quantum algorithm for (non-persistent) Betti numbers, including fault-tolerant overheads. This work also presented a family of graphs with (analytically computable) large Betti numbers, that result in a superpolynomial speedup of the quantum algorithm over textbook approaches for relative error estimates.
    \item Ref.~\cite{apers2022simple} presented a simple randomized classical algorithm for estimating normalized Betti numbers to additive error. The algorithm scales as
    \begin{equation*}
    \left( \frac{N}{\lambda_{\mathrm{max}}} \right)^{\bigO{\frac{1}{\sqrt{\Lambda}}\log\left(\frac{1}{\Delta}\right)}}\cdot \mathrm{poly}(N)
    \end{equation*}
    where $\lambda_{\mathrm{max}}$ is the largest eigenvalue, and assuming that we can efficiently sample and check $k$-simplices. When $k = \Omega(N)$, the algorithm runs in polynomial time for clique complexes with constant gap $\Lambda$ and error $\Delta = \Omega(1/\mathrm{poly}(N))$ (or $\Delta$ constant and $\Lambda = \Omega(1/\log^2(N))$). These are more restrictive conditions than quantum algorithms (which can simultaneously have both $\Lambda, \Delta = \Omega(1/poly(N))$). A related randomized classical algorithm was proposed independently in Ref.~\cite{berry2022quantifying}, the complexity of which depends on the mixing time of a Markov chain that enters due to the use of importance sampling, which might be hard to bound in practical use cases.
\end{itemize}

\section{Background on topological data analysis}\label{Sec:MathsBackground}

\begin{figure}[!ht]
    \centering
    \includegraphics[draft=false]{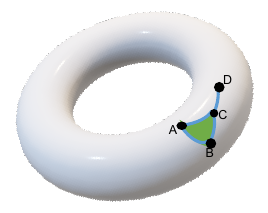}
    \caption{An example subset of datapoints arising from an underlying toroidal manifold. The datapoints ($0$-simplices $A, B, C, D$) are ordered alphabetically. After defining a length scale within which to connect datapoints, we add the relevant simplices. We show $1$-simplices $AB, BC, AC, CD$ (blue) and a $2$-simplex $ABC$ (green) that form elements of the resulting simplicial complex for this dataset. Simplices inherit the ordering of the vertices. We ascribe an orientation to the simplices based on the ordering of the vertices. An odd permutation of the vertices changes the orientation of the simplex; i.e. $AC = -CA$.}
    \label{fig:Torus}
\end{figure}

We seek to compute $\beta_k^{i,j}$ for a sufficient number of dimensions $k$ and scales $i,j$ so that we can capture the main topological features of the given dataset. As illustrated in Fig.~\ref{fig:Torus}, to understand the topology at scale $i$ we connect datapoints within distance $\mu_i$, generating a graph from the dataset. We use the so-called \emph{clique complex} of the graph, where simplices correspond to the cliques in the graph: a $(k+1)$-clique defines a $k$-simplex. 
This way every $k+1$ mutually connected vertices induce a $k$-simplex, which we can think of as the convex hull of the corresponding data points. The set of simplices and all of their constituent faces (e.g., the $2$-simplex $ABC$ has $1$-simplices $AB, BC, AC$ and $0$-simplices $A,B,C$ as faces) at scale $i$ is referred to as the simplicial complex $S^i$, while $S_k^i$ denotes the $k$-simplices in $S^i$. Accordingly, the number of simplices in the complex at scale $i$ is denoted by $|S^i|$, and the number of $k$-simplices by $|S_k^i|$.

This simplicial complex is known as a Vietoris-Rips complex, and is one possible choice~\cite{otter2017roadmap} for approximating the underlying topological space. The maximum number of $k$-simplices $|S_k^i|$ is $\binom{N}{k+1}$. Mapping the simplices to orthonormal basis vectors enables the use of linear algebraic tools for understanding the topology. 

The goal is to identify holes in the complex; a hole is a region of empty space, demarcated by its boundary. Studying these boundaries is crucial for finding holes, motivating the definition of a boundary operator. The $k$-th boundary operator maps $k$-simplices to their oriented boundaries. The action of the boundary operator $\partial$ on a $k$-simplex is
\begin{align}\label{Eq:BoundaryOp}
    \partial [v_0, ..., v_k] = \sum_{l=0}^k (-1)^l [v_0, ... \hat{v}_l, ... v_k],
\end{align}
where $v_0 ... v_k$ are the ordered vertices in the $k$-simplex, and $\hat{v}_l$ means the vertex is excluded from the simplex. We denote by $\partial_k^i\colon \langle S_k^i \rangle \rightarrow \langle S_{k-1}^i\rangle$ the boundary operator on the complex at length scale $i$ restricted to the subspaces $\langle S_k^i \rangle$ and $\langle S_{k-1}^i\rangle$ spanned by the basis vectors corresponding to $k$ and $k-1$ simplices present at scale $i$. For example, $\partial_1^i [CD] = D - C$, and $\partial_2^i [ABC] = BC - AC + AB$. 

We can form `chains' of $k$-simplices by taking linear combinations of their corresponding vectors. The boundary operator extends linearly onto chains of $k$-simplices. 
When applying the boundary operator to a closed chain of $k$-simplices (known as a $k$-cycle), each $(k-1)$-face appears twice, each time with opposite signs. For example,
\begin{equation}
    \partial_1^i [AB + BC - AC] = B - A + C - B - C + A = 0.
\end{equation}
More formally, all $k$-cycles are in the kernel of $\partial_k^i$. A $k$-cycle can surround either empty space, or will correspond to the boundary of a $(k+1)$-chain in the complex (the example above corresponds to this latter case, as this $k$-cycle is given by $\partial_2^i [ABC]$). The $k$-cycles corresponding to the boundary of a $(k+1)$-chain in the complex are thus in the image of $\partial_{k+1}^i$. 

To count holes in the complex at scale $i$, one can consider all $k$-cycles, and remove those that are the boundaries of $(k+1)$-chains in the complex. Accordingly, the $k$-th Betti number at scale $i$ is defined~\cite{hatcher2005algebraic} as 
\begin{align}\label{eq:PlainBetti}
\beta_k^i = \mathrm{dim}\left(\mathrm{Ker}(\partial_k^i)\right) - \mathrm{dim}\left(\mathrm{Im}(\partial_{k+1}^i)\right),
\end{align} 
where the first term counts the number of `hole-like' objects, and the latter removes those that are boundaries of $(k+1)$-simplices, leaving the number of true holes in the complex. Note that since $\partial_\ell^i$ is restricted to the subspaces $\langle S_\ell^i \rangle$ and $\langle S_{\ell-1}^i\rangle$ we have $\mathrm{Im}(\partial_{k+1}^i) \subseteq \mathrm{Ker}(\partial_k^i)\subseteq\langle S_k^i \rangle$.

Computing the Betti numbers at different scales is not sufficient to determine the persistence of topological features, as the Betti number does not uniquely identify the holes~\cite{neumann2019limitations}. In order to determine how many holes at scale $i$ are still present at scale $j$, one can consider the simplicial complexes formed at both scales ($S^i, S^j$), and modify Eq.~\eqref{eq:PlainBetti} to give the persistent Betti numbers~\cite{Edelsbrunner2002} as
\begin{align}\label{eq:PersistentBetti}
    \beta_k^{i,j} = \mathrm{dim}\left( \mathrm{Ker}(\partial_k^i) \right) - \mathrm{dim}\left( \mathrm{Ker}(\partial_k^i) \cap \mathrm{Im}(\partial_{k+1}^j) \right).
\end{align}
The first term in this expression again counts all hole-like objects at scale $i$. The second term removes all holes that were present at scale $i$, but have been `filled-in' by linear combinations of $(k+1)$-simplices at scale $j$ (which automatically includes any $(k+1)$-simplices at scale $i$ that fill-in holes at scale $i$). As above $\mathrm{Ker}(\partial_k^i)\subseteq\langle S_k^i \rangle$ and $\mathrm{Im}(\partial_{k+1}^j)\subseteq\langle S_k^j \rangle$, while $\langle S_k^i \rangle \subseteq \langle S_k^j \rangle$.

Once the persistent Betti numbers have been computed for all pairs of scales, they are typically summarized in a persistence diagram, that shows when topological features are created and destroyed. Different datasets can be compared by computing the similarities between these diagrams, either using  stable distance measures, or vectorization of the diagrams~\cite{leykam2022TDAphysicsreview,carlsson2020topologicalreview,otter2017roadmap}.

\section{Quantum algorithm for persistent Betti numbers}\label{Sec:QuantumAlgorithm}

\begin{figure*}
    \centering
    \includegraphics[width=0.9\linewidth]{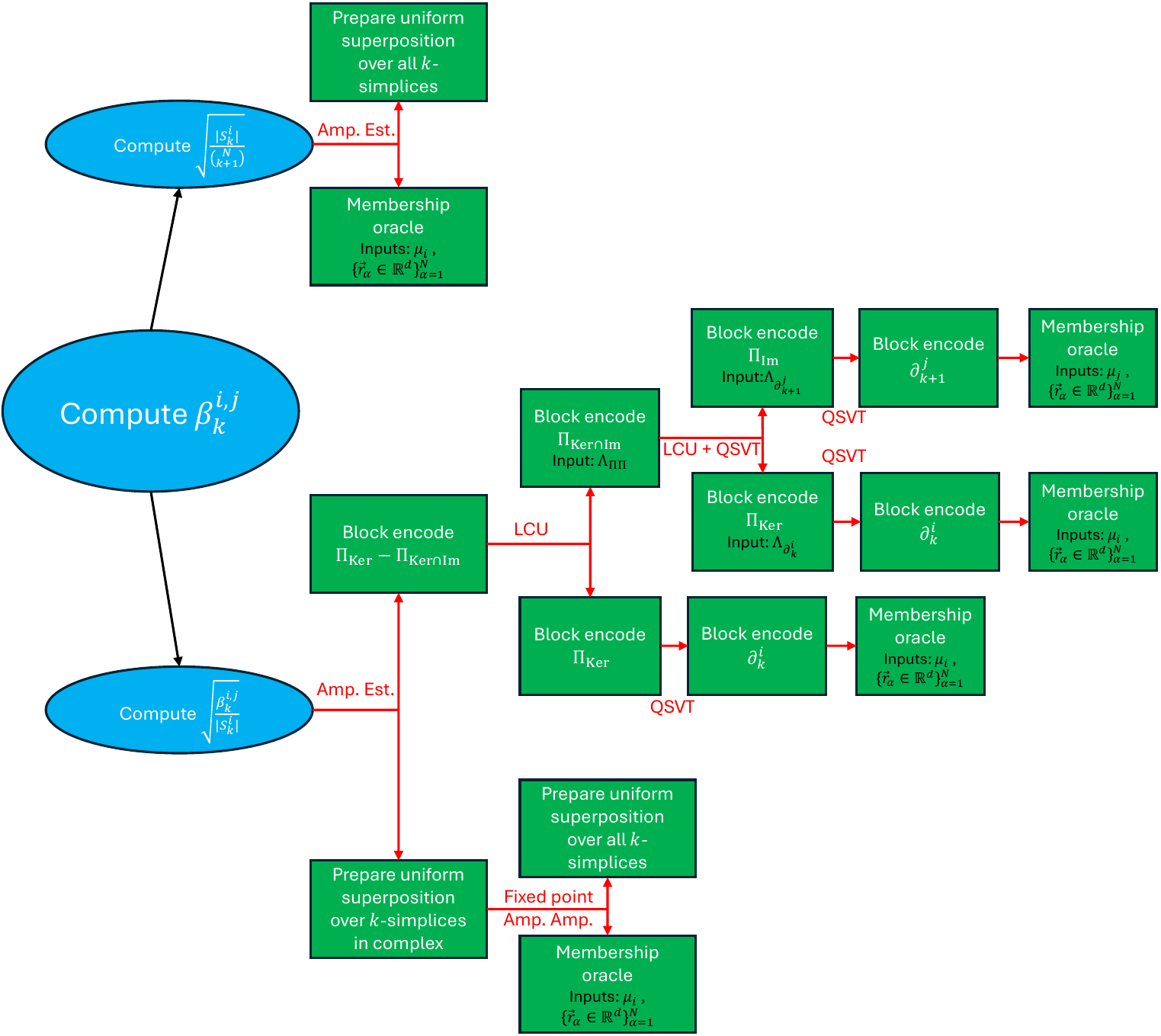}
    \caption{A flowchart illustrating the calls to made by the algorithm to various quantum and classical subroutines. Blue ovals denote classical outputs, and green boxes denote quantum subroutines. Black arrows denote classical postprocessing, while red arrows indicate coherent calls to a quantum subroutine. Some of the key quantum subroutines are named in red (Amp. Est./Amp. denotes amplitude estimation/amplification, and LCU denotes linear combination of unitaries). Quantum subroutines have their classical data inputs specified in black, with the exception of ubiquitous dependencies on the values of $N, k, \Delta$.}
    \label{fig:FlowDiagram}
\end{figure*}

In contrast to the approach taken herein, prior quantum algorithms did not directly use the formula in Eq.~(\ref{eq:PersistentBetti}) for computing (persistent) Betti numbers, for a number of reasons. One such reason is that the boundary operator is not Hermitian, so one cannot directly use traditional approaches such as quantum phase estimation. Therefore, earlier works mostly computed (persistent) Betti numbers by utilizing the (persistent) combinatorial Laplacian or related operators (see Appendix~\ref{AppSec:MathsBackgroundPedagogy}), that either made it difficult to extend methods from Betti numbers to persistent Betti numbers~\cite{lloyd2016quantum,gunn2019review,gyurik2020towards, ubaru2021quantum} or resulted in higher computational complexity~\cite{hayakawa2021quantum}.\\

We reduce the problem of estimating persistent Betti numbers to that of estimating the ratio of the ranks of two orthogonal projectors. We present a quantum native approach in the QSVT framework for solving this problem. Consider the orthogonal projectors $\Pi \preceq \widetilde{\Pi}$\footnote{where $X \preceq Y$ denotes that $Y-X$ is positive semidefinite.}. If we treat $\Pi$ as an observable and $\rho:=\widetilde{\Pi}/\mathrm{rank}(\widetilde{\Pi})$ as a quantum state, then the ratio of the ranks $\frac{\mathrm{rank}(\Pi)}{\mathrm{rank}(\widetilde{\Pi})}$ equals $\Pi$'s expectation value when evaluated on the mixed state $\rho$. We can speed up the estimation of this expectation value using amplitude estimation if we replace the mixed state $\rho$ by its purification $\ket{\psi}$. For example, when $\widetilde{\Pi}=I$, we choose $\ket{\psi}$ to be the maximally entangled state. We require a state preparation unitary $V_{\psi}$ for $\ket{\psi}$ and a block encoding~\cite{gilyen2019quantum} $V_{\Pi}$ of $\Pi$. This leads to the following result (we provide a full version in Appendix~\ref{AppSub:DimensionEstProof}).

\begin{theorem}\label{Theorem:SubspaceDimensionEst}
(Normalized projector rank estimation) Given the above ingredients, we can estimate $\sqrt{\frac{\mathrm{rank}(\Pi)}{\mathrm{rank}(\widetilde{\Pi})}}$ to additive error $\delta$ with success probability $\geq 1- \eta$ using $\bigO{\log(\eta^{-1})}$ incoherent repetitions of a quantum circuit, which makes $\bigO{\frac{1}{\delta}}$ calls to $V_{\psi}$ and $V_{\Pi}$ (and their inverses), and uses $\bigO{\frac{a}{\delta}}$ additional single- and two-qubit gates, where $a$ is an upper bound on the number of qubits on which $V_{\psi}$ and $V_{\Pi}$ act.
\end{theorem}

\begin{proof}
We employ amplitude estimation with precision $\delta$ to the following process: prepare $\ket{\psi}$ and measure $\Pi$. We repeat this process $\mathcal{O}\left( \log(\eta^{-1})\right)$ times and take the median of the estimates, which boosts the success probability exponentially due to the Chernoff bound.
\end{proof}

We compute (persistent) Betti numbers by the above subroutine via estimating the normalized rank of
\begin{equation}
    \Pi_\beta := \Pi_{\mathrm{Ker}(\partial_k^i)} - \Pi_{\mathrm{Ker}(\partial_k^i) \cap \mathrm{Im}(\partial_{k+1}^j)} \preceq \Pi_{\langle S_k^i \rangle}
\end{equation}
to sufficient precision (see Appendix~\ref{AppSub:CoroPersistentProof}).

To obtain unitary block-encodings $V_{\Pi_{\mathrm{Ker}}}$ (of $\Pi_{\mathrm{Ker}(\partial_k^i)}$) and  $V_{\Pi_{\mathrm{Im}}}$ (of $\Pi_{\mathrm{Im}(\partial_{k+1}^j)}$) we apply approximate threshold functions to $\partial_k^i$ and $(\partial_{k+1}^j)^\dagger$ via QSVT, which introduces a reciprocal dependence on the smallest non-zero singular values $\Lambda_{\partial_k^i}$, $\Lambda_{\partial_{k+1}^j}$ of $\partial_k^i$ and $\partial_{k+1}^j$ respectively. 

In general $\mathrm{Ker}(\partial_k^i) \cap \mathrm{Im}(\partial_{k+1}^j)$ is spanned by those singular vectors of $\Pi_{\mathrm{Ker}(\partial_k^i)}\cdot \Pi_{\mathrm{Im}(\partial_{k+1}^j)}$ that have singular value $1$. We can (approximately) implement a block encoding of $\Pi_{\mathrm{Ker}(\partial_k^i) \cap \mathrm{Im}(\partial_{k+1}^j)}$ by once again applying a threshold function via QSVT to $\Pi_{\mathrm{Ker}(\partial_k^i)}\cdot \Pi_{\mathrm{Im}(\partial_{k+1}^j)}$ (but now with a threshold of approximately $1$). The required sharpness of the threshold is determined by the quantity $\Lambda_{\Pi \Pi}$, the deviation of the largest $\neq 1$ singular value of $\Pi_{\mathrm{Ker}(\partial_k^i)}\cdot \Pi_{\mathrm{Im}(\partial_{k+1}^j)}$ from 1. This results in an additional factor of $\Lambda_{\Pi\Pi}^{-0.5}$ (see \cite[Lemma 35]{gilyen2018QSingValTransfArXiv} for an explanation of the quadratic improvement) in the complexity of the algorithm. For ordinary Betti numbers (i.e., $i=j$) we have $\Lambda_{\Pi\Pi} = 1$ since
$\mathrm{Im}(\partial_{k+1}^i)\subseteq \mathrm{Ker}(\partial_k^i)$ and thus $\Pi_{\mathrm{Ker}(\partial_k^i)}\cdot \Pi_{\mathrm{Im}(\partial_{k+1}^{j=i})}=\Pi_{\mathrm{Im}(\partial_{k+1}^i)}$.

In Fig.~\ref{fig:FlowDiagram} we present a flow diagram for our quantum algorithm, showing calls to the main subroutines of the algorithm. This is complemented by the minimal pseudocode in Algorithm~\ref{Alg:MembershipOracle} and Algorithm~\ref{Alg:PersistentBetti}. In order to determine the overall resource costs for estimating persistent Betti numbers using our quantum algorithm, we need to determine the cost of these subroutines. In the following section, we discuss how to implement the block encodings listed above. These require a number of ingredients, whose costs we establish for both the direct and compact mappings:
\begin{itemize}
    \item Implement the `membership oracle' (Algorithm~\ref{Alg:MembershipOracle}) that determines if a given $k$-simplex is present in the complex at scale $i$, using a sequence of elementary operations, including those used for querying a (quantum) lookup-table that stores the coordinates of the data points.
    \item Obtain block encodings of $\partial_k^i$, $\partial_{k+1}^j$ using elementary operations and the membership oracle.
    \item Obtain block encodings of $\Pi_{\mathrm{Ker}}$, $\Pi_{\mathrm{Im}}$, $\Pi_{\mathrm{Ker}(\partial_k^i) \cap \mathrm{Im}(\partial_{k+1}^j)}$, and $\Pi_\beta$ using QSVT applied to block encodings of $\partial_k^i$, $\partial_{k+1}^j$ and block encoding manipulations (linear combinations of block encodings and products of block encodings~\cite{gilyen2018QSingValTransfArXiv}).
\end{itemize}

\section{Resource costs of building blocks}\label{Sec:ResourceEstimates}

A fundamental quantum subroutine required by the algorithm is the membership oracle, for which we provide pseudocode in Algorithm~\ref{Alg:MembershipOracle}.

\begin{algorithm}[H]
\caption{Membership Oracle $\hat{O}_{m_k^i}$}\label{Alg:MembershipOracle}
\begin{algorithmic}[1]
\Require{$k$-simplex $s_k$, length-scale $\mu_i$}
\Ensure{$|s_k\rangle |a\rangle\mapsto |s_k \rangle |a \oplus (s_k \in S_k^i)\rangle$}

\State Load point coordinates $\vec{r}_j$ for all vertices $j\in s_k$
\State Compute all $\binom{k+1}{2}$ pairwise distances
\State If all distances $\le \mu_i$, then flip membership ancilla
\State Uncompute distances and unload point coordinates
\State \Return updated state
\end{algorithmic}
\end{algorithm}

Below, we provide minimal pseudocode for the full quantum algorithm, in Algorithm~\ref{Alg:PersistentBetti}.

\begin{algorithm}[H]
\caption{Estimate Persistent Betti Number $\hat{\beta}_k^{i,j}$}\label{Alg:PersistentBetti}
\begin{algorithmic}[1]
\Require{$N$ datapoints in $\mathbb{R}^d$; feature dimension $k$; scales $i,j$ with lengths $\mu_i,\mu_j$; target error $\Delta$}
\Ensure{$\hat{\beta}_k^{i,j} = \beta_k^{i,j} \pm \Delta$}
\Statex
\State \textbf{Assumptions:} Gap parameters $\Lambda_{\partial_k^i}$, $\Lambda_{\partial_{k+1}^j}$, $\Lambda_{\Pi\Pi}$

\Statex
\State \textbf{Step 1: Estimate} 
\[
\hat{X} = \sqrt{\frac{|S_k^i|}{\binom{N}{k+1}}} \pm \delta_1,
\quad
\delta_1 = 
\frac{\Delta}{4\beta_k^{i,j}}
\sqrt{\frac{|S_k^i|}{\binom{N}{k+1}}}
\]
\State $\hat{X} \gets \textsc{AmpEst}\left(\Pi_{S_k^i},\,|\psi_{S_k}\rangle,\,O(\delta_1)\right)$
\State \hspace{1em} \textit{// Block-encode $\Pi_{S_k^i}$ using the membership oracle $\hat{O}_{m_k^i}$}
\State \hspace{1em} \textit{// Prepare $|\psi_{S_k}\rangle$ as a uniform superposition state over all $k$-simplices}

\Statex
\State \textbf{Step 2: Estimate}
\[
\hat{Y} = \sqrt{\frac{\beta_k^{i,j}}{|S_k^i|}} \pm \delta_2,
\quad
\delta_2 = \frac{\Delta}{4\sqrt{|S_k^i|\,\beta_k^{i,j}}}
\]
\State $\hat{Y} \gets \textsc{AmpEst}\left(\Pi_{\mathrm{Ker}} - \Pi_{\mathrm{Ker}\cap\mathrm{Im}},\,|\psi_{S_k^i}\rangle,\,O(\delta_2)\right)$
\State \hspace{1em} \textit{// Use QSVT + LCU to block-encode $\Pi_{\mathrm{Ker}} - \Pi_{\mathrm{Ker}\cap\mathrm{Im}}$, with $\tilde{\mathcal{O}}\left(\frac{\sqrt{Nk}}{\Lambda_{\Pi \Pi}^{0.5} \min(\Lambda_{\partial_k^i}, \Lambda_{\partial_{k+1}^j})} \right)$ calls to $\hat{O}_{m_k^i}$}
\State \hspace{1em} \textit{// Prepare $|\psi_{S_k^i}\rangle$ using fixed-point amplitude amplification, which makes $\tilde{\mathcal{O}}\left(\sqrt{\binom{N}{k+1}/ |S_k^i|} \right)$ calls $\hat{O}_{m_k^i}$}

\Statex
\State \textbf{Step 3: Output estimate}
\[
\hat{\beta}_k^{i,j}
  = \binom{N}{k+1}\,\hat{X}^2\,\hat{Y}^2
\]
\State \Return $\hat{\beta}_k^{i,j}$
\end{algorithmic}
\end{algorithm}

\subsection{Mapping simplices to qubits}
We consider two approaches to map simplices to quantum states: the direct mapping using $\mathcal{O}(N)$ qubits for a $k$-simplex on $N$ datapoints, and the compact mapping, introduced herein, using $\mathcal{O}(k\log(N))$ qubits.

The direct mapping, considered in previous quantum algorithms~\cite{lloyd2016quantum}, encodes $k$-simplices as Hamming weight-$(k+1)$ computational basis states of $N$ qubits. Each qubit corresponds to a datapoint, and its value is $\ket{1}$ if and only if the vertex corresponding to the datapoint is incident to the given simplex.\footnote{Here we mean that the vertex is part of the defining $(k+1)$-clique in the corresponding graph, rather than the datapoint being in the convex hull of the points geometrically corresponding to the simplex.}

The compact mapping encodes a $k$-simplex built from an $N$-point dataset using $(k+1)\lceil \log(N+1) \rceil$ qubits. Each of the $(k+1)$ registers represents a vertex present in the simplex. The vertices are enumerated in binary, starting from $1$. The all-$0$ state $\ket{\bar{0}}$ denotes the absence of a vertex, and is required for implementing the boundary operator. The ordering of the vertices in the state is the same as the ordering chosen for the datapoints.

\begin{table}[!ht]
    \centering
    \begin{tabular}{|c|c|c|c|} \hline
        $k$ & Simplex & Direct & Compact  \\ \hline
        0 & $A$ & $\ket{1000000}$ & $\ket{001}$ \\
        1 & $BD$ & $\ket{0101000}$ & $\ket{010}\ket{100}$ \\
        2 & $CEG$ & $\ket{0010101}$ & $\ket{011}\ket{101}\ket{111}$ \\ \hline
    \end{tabular}
    \caption{Example mappings from simplices to qubits for simplices built from a vertex set $\{A, ... , G\}$.}
    \label{tab:SimplexMappings}
\end{table}

For $N$ datapoints, there are $\binom{N}{k+1}$ possible $k$-simplices. Classically storing generic linear combinations of these requires $\Omega \binom{N}{k+1}$ memory, while a quantum register can be put into a superposition of all possible $k$-simplex states. The compact mapping can thus provide an exponential reduction in spatial resources over classical approaches, for dense complexes. The compact mapping is also an exponential improvement over the direct mapping when $k = \mathcal{O} \left( \mathrm{polylog}(N) \right)$.

\subsection{Membership oracle cost}\label{Subsec:MembershipOracle}
In order to prepare $\ket{\psi_{S_k^i}}$ (a purification of $\Pi_{\langle S_k^i \rangle}/|S_k^i|$) and to restrict the boundary operators to act only on simplices present in the complex at a given length scale, the algorithm makes calls to a membership oracle $O_{m_k^i}$. The membership oracle determines if a given $k$-simplex is present in the complex at scale $i$ (or not) based on the positions of the vertices $\{\vec{r}_\alpha\}$, and the length scale $\mu_i$ considered. The membership oracle acts as
\begin{equation}
O_{m_k^i} \ket{s_k} \ket{a} = \begin{cases}
			\ket{s_k} \ket{a \oplus 1} & \text{if $s_k \in S_k^{i}$}  \\
            \ket{s_k} \ket{a \oplus 0} & \text{if $s_k \notin S_k^{i}$} 
\end{cases}
\end{equation}

We explicitly show how to implement the membership oracle for both the direct and compact mappings in Appendix~\ref{AppSub:MembershipOracles}. It is often claimed that the quantum algorithm for TDA does not require a quantum lookup-table. While it is possible to implement the algorithm without it~\cite{ubaru2021quantum}, typical formulations do in fact use quantum table lookups~\cite{lloyd2016quantum}. However, in contrast to other quantum machine learning algorithms, the quantum lookup-table used in TDA is not exponentially larger than the number of qubits used for storing the main register. We only require the ability to read from the quantum lookup-table, not to write to it.

\begin{table}[!ht]
    \centering
    \begin{tabular}{|c|c|c|} \hline
          Load & Space & Depth \\ \hline
        Direct & $N$ & $N\log(N)$ \\ \hline
         \makecell{Compact \\(Slow)} & \makecell{$k\log(N)$ \\ $+ kdb_d^2$} & \makecell{$N \log(db_d)$ \\ $+ k(\log(d)\log(b_d) + b_d) $} \\ \hline
          \makecell{Compact\\(Fast)} & \makecell{$N k d b_d$ \\ $+ k d b_d^2$} & \makecell{$\log(N)$ \\ $+ k(\log(d)\log(b_d) + b_d) $} \\ \hline
    \end{tabular}
    \caption{A comparison of the asymptotic time and space complexity of implementing the membership oracle in the two encodings, using the optimal memory models for each. The datapoints are in $\mathbb{R}^d$ and each coordinate is represented with $b$ bits, chosen sufficiently large to suppress under/overflow errors in the distance calculation. These complexities are derived in Appendix~\ref{AppSub:MembershipOracles}.}
    \label{tab:MembershipOracleComp}
\end{table}

\begin{table*}[t]
    \centering
\begin{tabular}{|c|c|c|} \hline
            & $(\alpha, m, \epsilon)$ & Gate depth  \\ \hline
    Compact & $(\sqrt{(N+1)(k+1)}, \log(N+1) + \log(k+1) + 1, 0)$ & $\mathcal{O}(k\log\log(N+1))$ \& $1\times O_{m_k^i}$ \\ \hline
    Direct & $(\sqrt{N}, 2, 0)$ & $\mathcal{O}(\log(N))$ \& $1\times O_{m_k^i}, O_{m_{k-1}^i}$ \\ \hline
\end{tabular}
    \caption{Parameters for the block-encoding $V_{\partial_k^i}$ of $\partial_k^i$ in the compact and direct mappings. The parameters for $V_{\partial_{k+1}^j}$ follow from replacing $k \rightarrow k+1$ and $i \rightarrow j$. The implementation of the membership oracles differ for the compact and direct encodings, as discussed in Sec.~\ref{Subsec:MembershipOracle}.}
    \label{tab:BoundaryOperatorCosts}
\end{table*}

Our membership oracle for direct mapped simplices follows the approach to clique finding in Ref.~\cite{metwalli2020cliquefinding}, and uses a slow-load quantum lookup-table that can be implemented with a purely classical memory by iterating a (classical) list of present edges, and applying a Toffoli gate on the corresponding vertices to increment a counter. The counter is used for verifying that all edges of the simplex are present.

Our membership oracle for compact mapped simplices also makes use of quantum lookup-tables~\cite{hann2021qram}; either performing fast loads using many ancilla qubits or, as with the direct mapping, slow loads with few ancilla qubits and a classical memory. We coherently load the coordinates of the vertices defining the simplex, and calculate the distance between each pair of them. Comparing the distances against the length scale $\mu_i$ flags edges that are not present and yields a membership oracle implementation. The compact membership oracle also verifies that the vertices are stored in increasing order for consistency.

TDA is widely applied to high dimensional data. In order to reduce the cost of the membership oracle in the compact mapping, we can implement a classical pre-processing step. We can use the Johnson–Lindenstrauss lemma to embed the high-dimensional vertices into a lower dimensional space, while approximately maintaining the distances between vertices. The dimension $d$ can be reduced to $\bigO{\log(N)/\epsilon_{JL}^2}$ to ensure the distances between all the points stay accurate to a factor of $(1\pm \epsilon_{JL})$.

\subsection{Boundary operator cost}\label{Subsec:BoundaryOpCost}
In this section we discuss the resource costs to construct block encodings of the boundary operator $\partial_k^i$ and $\partial_{k+1}^j$ according to \Cref{Eq:BoundaryOp}.

\begin{definition}\label{def:blockDef}
	We say that a unitary matrix $U$ is an $(\alpha,(a,b),\eps)$-block-encoding of the matrix $A$ if
	\begin{equation}\label{eq:blockDef}
		\nrm{A-\alpha (\bra{0}^{\otimes a}\otimes I)U(\ket{0}^{\otimes b}\otimes I)}\leq \eps.
	\end{equation}
	We also use this notion in cases where $(\bra{0}^{\otimes a}\otimes I)U(\ket{0}^{\otimes b}\otimes I)$ acts on larger subspaces than $A$ itself, in which in case, in \Cref{eq:blockDef}, by $A$ we mean its trivial embedding\footnote{By the trivial embedding we mean extending $A$ with $0$ matrix elements, so that the non-zero singular values and the corresponding singular-vector pairs are unchanged.} into these larger subspaces. Setting $m=\max\{a,b\}$ we might call it an $(\alpha,m,\eps)$-block-encoding or $m$-qubit block-encoding for brevity. 
\end{definition}

We construct a unitary $V_{\partial_k^i}$ such that $(\bra{0}^{\otimes a}\otimes I) V_{\partial_k^i} (\ket{0}^{\otimes b}\otimes I)= \frac{1}{\alpha} \partial_k^i $ (and analogously for $\partial_{k+1}^j$).\footnote{Literally implementing the restriction of $\partial_k^i$ to the subspaces $\langle S_k^i \rangle$ and $\langle S_{k+1}^i \rangle$ is technically difficult. Instead we implement $\bar{\partial}_k^i$ which is an extension of $\partial_k^i$ onto the larger space defined by all $(k+1)\log(N)$ qubit states of the compact encoding (respectively $N$ qubit states of the direct encoding), extending $\partial_k^i$ trivially to the additional dimension so that $\mathrm{Im}(\partial_k^i)=\mathrm{Im}(\bar{\partial}_k^i)\subseteq \langle S_k^i \rangle$ and $\mathrm{Im}((\partial_k^i)^\dagger)=\mathrm{Im}((\bar{\partial}_k^i)^\dagger)\subseteq \langle S_{k+1}^i \rangle$, which is sufficient for our purposes.} The costs are summarized in Table~\ref{tab:BoundaryOperatorCosts}, and are derived in Appendix~\ref{AppSub:BoundaryOp}.

Our compact mapped approach proceeds by coherently swapping the vertex to be deleted into the final position of the register, and then setting it to $\ket{\bar{0}}$ to represent the absence of a vertex. 

Our direct mapped approach relies on a previously observed correspondence between the boundary operator and second quantized fermionic operators: $\partial + \partial^\dag = \sum_i a_i + a_i^\dag $~\cite{cade2021complexity, akhalwaya2022efficient, kerenidis2022quantum, ubaru2021quantum}. We considered existing approaches for implementing block-encodings of such operators~\cite{kerenidis2022quantum, wan2021exponentially}. The chosen approach~\cite{kerenidis2022quantum} for the direct mapping is a $(\sqrt{N}, 2, 0)$ block-encoding of $\partial_k^i$ that can be implemented with $\mathcal{O}(\log(N))$ gate depth and two calls to the membership oracle. This can be compared with the approach of Ref.~\cite{hayakawa2021quantum} that implemented a $(Nk, \mathcal{O}(\log(N), 0)$ block-encoding of $\partial_k^i$ requiring $\Omega(N)$ depth and one call to the membership oracle.

\subsection{Subspace projector cost}\label{Subsec:SubspaceProj}
As discussed in Sec.~\ref{Sec:QuantumAlgorithm} we use unitary block encodings $V_{\Pi_{\mathrm{Ker}}}$ (of $\Pi_{\mathrm{Ker}(\partial_k^i)}$) and $V_{\Pi_{\mathrm{Im}}}$ (of $\Pi_{\mathrm{Im}(\partial_{k+1}^j)}$), which are in turn used for building $V_{\Pi_{\mathrm{Ker}(\partial_k^i) \cap \mathrm{Im}(\partial_{k+1}^j)}}$ (of $ \Pi_{\mathrm{Ker} \cap \mathrm{Im} }$). In this section we present an outline of the methods used for building these block encodings, and state their cost in Table~\ref{tab:ProjectorEncodingsKer} and Table~\ref{tab:ProjectorEncodingsKerIm}. We refer to Appendix~\ref{AppSub:SubspaceProjectors} for details.

\begin{table}[!ht]
    \centering
    \begin{tabular}{|c|c|c|} \hline
          & $(\alpha, m, \epsilon)$ & Costs ($\mathcal{O}$) \\ \hline
         C & \makecell{$(1,$ \\ $ \log\left((N\!+\!1)(k\!+\!1)\right) \!+\! 2, $ \\ $ \epsilon_k) $} & \makecell{$\frac{\sqrt{(N+1)(k+1)}}{\Lambda_{\partial_k^i}} \log\left(\frac{1}{\epsilon_k}\right) $ \\ $\times(V_{\partial_k^i} \!,\! V_{\partial_k^i}^\dag)$} \\ \hline
        D & $(1, 3, \epsilon_k) $ & \makecell{$\frac{\sqrt{N}}{\Lambda_{\partial_k^i}} \log\left(\frac{1}{\epsilon_k}\right) $ \\ $\times(V_{\partial_k^i} \!,\! V_{\partial_k^i}^\dag)$} \\ \hline
    \end{tabular}
    \caption{A summary of the costs to implement the block encoding $V_{\Pi_{\mathrm{Ker}}}$ for the compact (C) and direct (D) mappings.}
    \label{tab:ProjectorEncodingsKer}
\end{table}

\begin{table*}[t]
    \centering
    \begin{tabular}{|c|c|c|} \hline
          & $(\alpha, m, \epsilon)$ & Costs \\ \hline
         Compact & \makecell{$\left(2, \log\left((N+1)^2(k+1)(k+2) \right) + 6, \epsilon_{p} + \epsilon_{k} + \frac{8}{\Lambda_{\Pi\Pi}^{0.5}}\log\left( \epsilon_{p}^{-1} \right) \sqrt{\epsilon_{k} + \epsilon_{i}} \right)$} &
         \multirow{2}{*}{\makecell{$\mathcal{O}\left( \frac{1}{\Lambda_{\Pi\Pi}^{0.5}} \log\left(\frac{1}{\epsilon_{p}}\right) \right) \times$ \\ $V_{\Pi_{\mathrm{Ker}}}$, $V_{\Pi_{\mathrm{Ker}}}^\dag$, $V_{\Pi_{\mathrm{Im}}}$, $V_{\Pi_{\mathrm{Im}}}^\dag$}}
         \\ \cline{1-2}
        Direct & $\left(2, 8, \epsilon_{p} + \epsilon_{k} + \frac{8}{\Lambda_{\Pi\Pi}^{0.5}}\log\left( \epsilon_{p}^{-1} \right) \sqrt{\epsilon_{k} + \epsilon_{i}} \right) $ &  \\ \hline
    \end{tabular}
    \caption{A summary of the costs to implement the block encoding $V_{\Pi_\beta}$ of $ \Pi_\beta := \Pi_{\mathrm{Ker}(\partial_k^i)} - \Pi_{\mathrm{Ker}(\partial_k^i) \cap \mathrm{Im}(\partial_{k+1}^j)}$.
    }
    \label{tab:ProjectorEncodingsKerIm}
\end{table*}

The QSVT circuit for implementing $V_{\Pi_{\mathrm{Ker}}}$ uses singular value threshold projectors~\cite{gilyen2019quantum}, i.e., it (approximately) transforms $\partial_{k}^i$ into the projector $\Pi_{\mathrm{Ker}}$ by mapping $0$ singular values to $1$ and non-zero singular values (that are at least $\Lambda$) to approximately $0$. A technical detail is that we need to use an even-degree polynomial~\cite{gilyen2019quantum} ensuring that the left and right singular vectors coincide, so that we end up with an (approximate) orthonormal projector. Singular value threshold projectors are a generalization of eigenvalue threshold projectors that are used for example in Ref.~\cite{lin2020near}, for applying a filter function that projects out eigenvectors above a given threshold. In Ref.~\cite{lin2020near}, to prepare the ground state, the threshold is set lower than the gap $\Lambda$ between the ground and first excited state. In our application we can use the polynomial approximations $P(x)$ of the threshold function used in either of \cite{gilyen2019quantum,lin2020near}, but we need to carefully track the propagation of errors stemming from the approximation error~$\epsilon_k$. 

We can implement $V_{\Pi_{\mathrm{Im}}}$ similarly by observing that $\mathrm{Im}(\partial_{k+1}^j)$ is the orthogonal complement of $\mathrm{Ker}((\partial_{k+1}^j)^\dagger)$. We can thus use the above construction to implement $V_{\Pi_{\mathrm{Im}}}$ by applying QSVT to $(\partial_{k+1}^j)^\dagger$ with the polynomial $1-P(x)$ (as opposed to $P(x)$ above).

Given the above methods for implementing $V_{\Pi_{\mathrm{Ker}}}$ and $V_{\Pi_{\mathrm{Im}}}$, we can implement $V_{\Pi_{\mathrm{Ker}(\partial_k^i) \cap \mathrm{Im}(\partial_{k+1}^j)}}$. We cannot simply take a product of the two projectors, as in general they do not commute with each other (because of new simplices that enter the complex at scale $j$). 
However, we can implement $\Pi_{\mathrm{Ker}(\partial_k^i) \cap \mathrm{Im}(\partial_{k+1}^j)}$ by first taking the product of the block-encoded matrices $\Pi_{\mathrm{Ker}}$ and $\Pi_{\mathrm{Im}}$~\cite{gilyen2019quantum}, and then applying QSVT to the resulting block encoding of $\Pi_{\mathrm{Ker}}\cdot \Pi_{\mathrm{Im}}$. 
We apply a function that sends all $\neq 1$ singular values to $0$. In reality, we can only send singular values below a threshold $(1 - \Lambda_{\Pi\Pi})$ to zero, and we choose $\Lambda_{\Pi\Pi}$ so that all $\neq 1$ singular values of $\Pi_{\mathrm{Ker}}\cdot \Pi_{\mathrm{Im}}$ are below $(1 - \Lambda_{\Pi\Pi})$. Such a thresholding costs $\mathcal{O}\left( \Lambda_{\Pi\Pi}^{-0.5} \log\left(\epsilon_{p}^{-1} \right) \right)$ (see \cite[Lemma 35]{gilyen2018QSingValTransfArXiv} for an explanation of the quadratic improvement) calls to $V_{\Pi_{\mathrm{Ker}}}$ and $V_{\Pi_{\mathrm{Im}}}$ to implement $V_{\Pi_{\mathrm{Ker}(\partial_k^i) \cap \mathrm{Im}(\partial_{k+1}^j)}}$ with precision $\epsilon_{p}$.

\subsection{State preparation cost}\label{Subsec:StatePrep}
In this section we describe how to implement $V_{\psi}$ preparing $\ket{\psi_{S_k^i}}$ from the all-$0$ state, where $\ket{\psi_{S_k^i}}$ is a purified maximally mixed state over $k$-simplices in the complex at scale $i$. We discuss our approach in more detail in Appendix~\ref{AppSub:StatePrep}.

We use (fixed-point) amplitude amplification to construct an approximate version of $\widetilde{V}_{\psi}$ that prepares an $\epsilon_\psi$-approximation of $\ket{\psi_{S_k^i}}$. 
For both compact and direct encodings the algorithm uses $\bigOt{\sqrt{\frac{\binom{N}{k+1}}{|S_{k}^i|}}\log\left(\frac{1}{\epsilon_\psi}\right)}$ calls to the membership oracle $O_{m_{k}^i}$, and to an operator $U_{\mathrm{uni}}$ (and its inverse) that prepares a purification of the maximally mixed state over all possible $k$-simplices.

In the direct mapping, $k$-simplices are Hamming weight $(k+1)$ computational basis states on $N$ qubits. We can implement $U_{\mathrm{uni}}$ using the approaches in Ref.~\cite{bartschi2019Dicke,bartschi2022DickeImproved} for preparing Dicke states, the most efficient of which has $\mathcal{O}(k\log(N))$ depth. 

For the compact mapping, we can implement $U_{\mathrm{uni}}$ (to construct a uniform superposition of $k$-simplices, with correctly ordered vertices) through a multi-step process. We first place all vertex registers into an equal superposition. We then use exact amplitude amplification to convert this state to a superposition with no repeated vertices (but ignoring the ordering of vertices). We can use exact amplitude amplification, because we can classically exactly compute the probability $\binom{N}{k+1}\cdot (k+1)!/N^{k+1}$ that we do not get repeated vertices. Finally, we can use a reversible quantum sorting network~\cite{gilyen2014MScThesis}, to obtain the desired equal superposition over all $k$-simplices, with their vertices correctly ordered.\footnote{This process can be made even more efficient by performing reversible sorting only after fixed-point amplitude amplification has been performed to filter out simplices that are not in the complex at the chosen length scale.} In the practically relevant regime when $k\leq \sqrt{N}$ this requires
\begin{align}
    \bigO{\log(N) + k\log(k) \log\log(N)}
\end{align}
circuit depth and $\mathcal{O}\left( k\log(N)  \right)$ ancilla qubits.

\section{Quantum-inspired classical algorithm for TDA}\label{Sec:Classical}

In this section we review existing classical algorithms for topological data analysis, and present a classical algorithm based on the power method, inspired by our quantum algorithm, for computing persistent Betti numbers.
There are a number of classical algorithms for persistent Betti numbers. The most widely used approach~\cite{Edelsbrunner2002} (referred to as the `textbook', `standard', or `column' algorithm) proceeds by `reducing' an $\left(|S_{k}^j|+|S_{k+1}^j|\right) \times \left(|S_{k-1}^j|+|S_{k}^j|\right)$ boundary matrix to identify pairs of simplices that create and (later) destroy a $k$-dimensional topological feature. The algorithm requires only a single repetition to compute the persistent Betti numbers at all length scales. Define a shorthand notation $|S_{k, k+1}^j| := |S_{k}^j|+|S_{k+1}^j|$. This algorithm has worst-case runtime of $\mathcal{O}(|S_{k, k+1}^j|^3)$~\cite{morozov2005persistence}, although the runtime can approach linear in practice, due to sparsity in the complex~\cite{milosavljevic2011zigzag}, and its constant factors have been heavily optimised~\cite{chen2011persistent, de2011dualities}. The complexity can be improved to $\mathcal{O}(|S_{k, k+1}^j|^\omega)$~\cite{milosavljevic2011zigzag, milosavljevic2010:inriaReport} (where $\omega \in[2,2.372)$ is the exponent for matrix multiplication).

Classical techniques have also been introduced that prune simplices from the complex. While it is NP-hard to find the maximally pruned complex~\cite{joswig2006computing}, a number of heuristic approaches have been developed~\cite{mrozek2008homology, zomorodian2010tidy, barmak2012strong, mischaikow2013morse, dlotko2014simplification, boissonnat2018strong} (and see Ref.~\cite{otter2017roadmap} Secs. 3.1 and 5.2.6). If the number of simplices in the reduced complex is $|\bar{S}|$, the algorithm of Ref.~\cite{mischaikow2013morse} reduces the complexity of computing persistent Betti numbers from $\mathcal{O}(|S_{k,k+1}^j|^\omega)$ to $\mathcal{O}(|S_{k,k+1}^j| + |\bar{S}|^\omega)$. When $|\bar{S}| \ll |S_{k,k+1}^j|$ the algorithm runs approximately linearly in the number of simplices in the original complex.

An alternative approach, which has parallels with the quantum algorithm for TDA, is to use the power method to find the dimension of the kernel of the combinatorial Laplacian matrix~\cite{friedman1998computing}. This approach was previously only applicable to computing Betti numbers. Like quantum algorithms for TDA, this approach also has a runtime that depends on the gap between the ground and first excited state of the operator used to encode the topology of the complex. This approach has a time complexity of $\tilde{\mathcal{O}}\left( \frac{|S_k^i|(k^2 \beta_k^i + k (\beta_k^i)^2)}{\Lambda}\right)$ and a spatial complexity of $\mathcal{O}\left( |S_k^i|(k+\beta_k^i) \right)$~\cite{friedman1998computing}. This complexity provides an apples-to-apples comparison for using quantum algorithms to compute regular Betti numbers. Even in the dense complex regime of $|S_k^i| \sim \binom{N}{k+1}$, the quantum algorithm only achieves an approximately quadratic speedup in $N$. With the recent introduction of the persistent combinatorial Laplacian~\cite{wang2020persistent, memoli2020persistent, wang2021hermes}, the power method could also be used for computing persistent Betti numbers. However, the asymptotic cost of building the persistent combinatorial Laplacian is the same as that of diagonalizing it~\cite{memoli2020persistent}, so there would likely be no benefit from performing the power method with this object. 

Inspired by our quantum algorithm for the computation of persistent Betti numbers, we present a classical power method for this same problem. At a high level, our algorithm constructs a polynomial approximation to the kernel projector. Consider a univariate polynomial $p(x)=a_3x^3+a_2x^2+a_1x+a_0$. We can rewrite the polynomial in the following way: $x(x(a_3x+a_2)+a_1)+a_0=(xy+a_0)\circ(xy+a_1)\circ(a_3x+a_2)$, where the composition happens in variable $y$. Similarly in general we have $a_d x^d+\ldots+a_2x^2+a_1x+a_0=(xy+a_0)\circ(xy+a_1)\circ\cdots\circ(a_d x+a_{d-1})$. This is known as Horner's method for evaluating polynomials. Now if $X$ is a matrix, and we want to compute $p(X)v$ we can do so by only using $d$ subsequent matrix vector products by evaluating $(X+a_0 I)\circ(X+a_1 I)\circ\cdots\circ(a_d X+a_{d-1} I)v$ from right to left. In our case, we have singular value transformations, so the above expression is modified to $(X+a_0 I)\circ(X^\dagger+a_1 I)\circ\cdots\circ(a_d X+a_{d-1} I)v$ with alternating $X$ vs. $X^\dagger$. In order to realize expressions like $p(q(X))$, it is sufficient to replace $X$ by $q(X)$ above, and implement each $q(X)$ analogously.
Using the above construction, we can implement a kernel projector for a sparse matrix using a sequence of matrix-vector products. For an $\epsilon$-accurate approximation of the projector $\Pi_{\mathrm{Ker}(\partial_k^i) \cap \mathrm{Im}(\partial_{k+1}^j)}$ we require an $\mathcal{O} \left( \frac{\log^2\left(\frac{1}{\epsilon}\right) }{\mathrm{Min}\left(\Lambda_{\partial_k^i}, \Lambda_{\partial_{k+1}^j}\right) \Lambda_{\Pi \Pi}^{0.5}}  \right)$ degree polynomial in the matrices $\partial_k^i, \partial_{k+1}^j$. A reasonable concern would be that the required polynomial could have exponentially large coefficients (in terms of the degree) in front of its monomials, and therefore that naive evaluation could yield large errors. It has been shown that this method of applying matrix polynomials is numerically stable, and so this is not the case~\cite{aurentz2019StableMatrixPolynomial}.

The matrix $\partial_k^i$ has column sparsity of $(k+1)$ and a row sparsity of $(N-k)$. Hence, the costs of each matrix-vector product are: $\partial_k^i : \mathcal{O}\left(S_{k-1}^i (N-k) \right)$, $\left(\partial_k^i\right)^\dag : \mathcal{O}\left(S_{k}^i (k+1) \right)$, $\partial_{k+1}^j : \mathcal{O}\left(S_{k}^j (N-k-1) \right)$, $\left(\partial_{k+1}^j\right)^\dag : \mathcal{O}\left(S_{k+1}^j (k+2) \right)$. Note that we are storing a compressed vector with length $S := |S_{k-1}^i| + |S_{k}^j| + |S_{k+1}^j|$, which only stores the relevant simplices present in the complex at the given length scales. We choose an unnormalized initial random vector for the simplices in $S_k^i$ (e.g. all entries $v_\alpha \in \pm 1$ if $\alpha \in S_k^i$, $v_\alpha = 0$ otherwise) and apply the projector $\Pi_{\mathrm{Ker}(\partial_k^i) \cap \mathrm{Im}(\partial_{k+1}^j)}$ through the repeated matrix-vector multiplication described above. 
By picking a series of random initial states, applying the approximate projector, and performing Gram-Schmidt orthogonalization of the resulting vectors, we can compute the rank of $\Pi_{\mathrm{Ker}(\partial_k^i) \cap \mathrm{Im}(\partial_{k+1}^j)}$, which is equal to the persistent Betti number $\beta_k^{i,j}$. This process uses $\beta_k^{i,j}$ random initial vectors, leading to a cost of 
\begin{equation}
    \mathcal{O} \left( \beta_k^{i,j} N S \frac{\log^2\left(\frac{1}{\epsilon}\right) }{\mathrm{Min}\left(\Lambda_{\partial_k^i}, \Lambda_{\partial_{k+1}^j}\right) \Lambda_{\Pi \Pi}^{0.5}}  \right)
\end{equation}
plus an additive cost of $\mathcal{O}\left(S \left(\beta_k^{i,j}\right)^2  + \left(\beta_k^{i,j}\right)^3 \right)$ to perform Gram-Schmidt orthogonalization and computation of the eigenvalues in the subspace, and an additive cost of $\mathcal{O}\left(N \sum_{\ell=0}^{k+1} |S_\ell^j|  \right)$ to find the $(k-1),k,(k+1)$-simplices in the complex at scales $i,j$\footnote{This additional cost is also present in the classical algorithms introduced above, but is typically neglected as it is subleading for small $k$.}. We observe that for small $\beta_k^{i,j}$, and $S \sim \binom{N}{k+1}$, the complexity in $\beta_k^{i,j}, N, k$ is only quadratically worse than our quantum algorithm, while the complexity in the gap parameters is the same. The classical algorithm also returns the vectors representing the persistent holes. Moreover, this algorithm has a better asymptotic scaling in $S$ than existing classical approaches~\cite{morozov2005persistence,memoli2020persistent} (although the error and gap dependence are worse, and the constant factors may be high).   
For the task of estimating normalized persistent Betti numbers, we can consider a similar approach to that of Ref.~\cite{apers2022simple}, which is similar to performing a path integral. We would start from a single simplex (a computational basis state) and approximate the projector through the sequence of matrix vector multiplications. Averaging over a number of initial states returns an estimate of the normalized persistent Betti number. As this approach performs sparse-matrix-sparse-vector multiplication, it is possible to realize each multiplication with complexity independent of the number of simplices $S$. However, as the sparsity grows with each multiplication, we are limited in the number of steps we can take, which limits the accuracy of the method. We leave careful analysis of the complexity of this approach to future work.

\section{Discussion}\label{Sec:Discussion}

\subsection{Quantum complexity}\label{Subsec:QuantumComplexity}
In Appendix~\ref{AppSec:OverallComplexity} we determine the overall complexity of our quantum algorithm for computing persistent Betti numbers, using the costs of the building blocks introduced in the previous sections. Using the direct mapping, we show in Appendix~\ref{AppSubSec:OverallComplexityDirect} that the algorithm has complexity given by the expression in Table~\ref{tab:QuantumComplexities}, to estimate the persistent Betti number to additive error $\Delta$ with success probability $ \geq 1 - \eta$. The quantum circuit acts on $\mathcal{O}(N)$ qubits. Taking $|S_k^i| \sim \binom{N}{k+1}$, hiding logarithmic factors with the big-$\tilde{\mathcal{O}}$ notation, and accounting for parallel implementation of gates recovers Theorem~\ref{theorem:overallcomplexity}. We can compare this scaling to the corresponding complexity of the quantum algorithm for computing persistent Betti numbers in Ref.~\cite{hayakawa2021quantum}, which we also show in Table~\ref{tab:QuantumComplexities}. That work used QSVT to construct the persistent combinatorial Laplacian (see Appendix~\ref{AppSec:MathsBackgroundPedagogy} for background), and to project onto its kernel, which encodes the persistent Betti numbers~\cite{memoli2020persistent}. We see that our approach provides a large polynomial improvement in $|S_k^i|, N, k$ over the prior state-of-the-art.

\begin{table*}[t]
    \centering
    \begin{tabular}{|c|c|} \hline
        Approach & Gate count \\ \hline
         This work (Direct mapping) & $ \tilde{\mathcal{O}}\left( \frac{N^2 \log(N)\sqrt{|S_k^i| \beta_k^{i,j}}}{\Delta} \log\left( \frac{1}{\eta} \right) \times \left( \sqrt{\frac{\binom{N}{k+1}}{|S_k^i|}}  + \frac{\sqrt{N}}{\Lambda_{\Pi\Pi}^{0.5} \mathrm{Min}\left(\Lambda_{\partial_k^i}, \Lambda_{\partial_{k+1}^j} \right)} \right) \right)$ \\ \hline
         Ref.~\cite{hayakawa2021quantum} (Direct mapping) &         $\tilde{\mathcal{O}}\left( \frac{N^2 |S_k^i| \beta_k^{i,j}}{\Delta^2}\times \left( \sqrt{\frac{\binom{N}{k+1}}{|S_k^i|}}  + \frac{N^6 k^4}{\Lambda_1^2 \Lambda_2} \right) \right)$ \\ \hline
    \end{tabular}
    \caption{A comparison of the asymptotic complexity of our algorithm using the direct encoding, and the approach of Ref.~\cite{hayakawa2021quantum}. We have converted our gate depth estimates into gate counts by pessimistically multiplying by a factor of $N$. The overall complexity of the approach of Ref.~\cite{hayakawa2021quantum} is not explicitly stated in that work, and so we have used our best estimates based on the subroutines given. The gap parameters $\Lambda_1, \Lambda_2$ refer to the gaps of a submatrix of $\partial_{k+1}^j \left(\partial_{k+1}^j\right)^\dag$, and of the persistent combinatorial Laplacian, respectively. }
    \label{tab:QuantumComplexities}
\end{table*}

Both of these quantum algorithms have dependencies on two different gap parameters; in our case, $\Lambda_{\Pi \Pi},  \mathrm{Min}\left(\Lambda_{\partial_k^i}, \Lambda_{\partial_{k+1}^j} \right)$ and in the case of Ref.~\cite{hayakawa2021quantum} $\Lambda_1, \Lambda_2$ (where $\Lambda_1$ is the gap of a submatrix of $\partial_{k+1}^j \left(\partial_{k+1}^j\right)^\dag$, and $\Lambda_2$ is the gap of the persistent combinatorial Laplacian). The gaps $\mathrm{Min}\left(\Lambda_{\partial_k^i}, \Lambda_{\partial_{k+1}^j} \right)$ and $\Lambda_2$ arise due to the need to project into the relevant kernel/image spaces, and their non-persistent analogues are also present when computing regular Betti numbers. However, the gaps $\Lambda_{\Pi\Pi}$ and $\Lambda_1$ have a more subtle origin, that appears to be specific to the problem of computing persistent Betti numbers. In Ref.~\cite{hayakawa2021quantum}, this gap arises from performing a change of basis necessary to build the persistent combinatorial Laplacian (see Appendix~\ref{AppSub:BettiNuances} for further discussion). In our algorithm, this gap arises from transforming the projector $\Pi_{\mathrm{Ker(\partial_k^i)}} \cdot \Pi_{\mathrm{Im(\partial_{k+1}^j)}}$   to $\Pi_{\mathrm{Ker} \cap \mathrm{Im}}$ by applying QSVT. This is necessary because $\Pi_{\mathrm{Ker(\partial_k^i)}}, \Pi_{\mathrm{Im(\partial_{k+1}^j)}}$ do not commute in general. In both cases, this additional complexity arises due to the new simplices that are present at stage $j$, that are not present at stage $i$. This makes it necessary to perform an effective change of basis, to ensure only the relevant simplices in $j$ that can fill holes in $i$ are considered. It remains an open question as to how these gap parameters scale as a function of $N, k$. A number of conjectures about the scaling of the gap of the combinatorial Laplacian were given in Ref.~\cite{friedman1998computing}, but we are unaware of any proofs or further evidence towards these conjectures, or extensions to persistent operators. The scaling of the gap parameters will determine the viability of quantum algorithms for topological data analysis, and so it is critical to establish how these gaps scale for problems of interest.\\

As shown in Theorem~\ref{theorem:overallcomplexity} and Appendix~\ref{AppSubSec:OverallComplexityCompact}, the corresponding expression for the compact mapping (with a slow-load lookup-table implementation of the membership oracle) is similar to that of our direct mapped algorithm. However, the quantum circuit acts on only $\mathcal{O}\big{(} k \log(N) \big{)}$ qubits, providing an exponential reduction in the number of qubits required compared to the prior state-of-the-art~\cite{hayakawa2021quantum}, in addition to a polynomial improvement in gate complexity. This reduction in spatial complexity is significant; in practical regimes of interest, we can expect $k \approx 3$, $N \approx 10^6$, in which case (neglecting constant factors and subleading terms) the number of logical qubits is reduced from $10^6$ to around $80$.

\subsection{Prospects for quantum advantage}\label{Subsec:QuantumSpeedups}
We can consider the regimes in which quantum algorithms may provide a speedup over the classical approaches discussed in Sec.~\ref{Sec:Classical}, including our new quantum-inspired approach. For $k$ constant, both the quantum and classical algorithms run in polynomial time, preventing quantum algorithms from achieving an exponential speedup. For $k$ scaling with $N$, and a clique dense complex, all known classical algorithms run in superpolynomial time. This is likely to be a fundamental result; for clique complexes determining if $\beta_k^i$ is non-zero is QMA$_1$-hard~\cite{crichigno2022clique,king2023promise}, and computing $\beta_k^i$ is $\#$P-hard~\cite{schmidhuber2022complexity}.

When our task is to estimate $\beta_k^{i,j}$ to constant additive error $\Delta$, all known quantum algorithms similarly scale superpolynomially (for $k$ scaling with $N$). The limitation of quantum algorithms stems from their competing subroutines. The algorithms need to efficiently find simplices in the complex, so $|S_k^i|$ should only be polynomially smaller than $\binom{N}{k+1}$. On the other hand, the algorithms compute $\beta_k^{i,j}/|S_k^i|$ to additive error $\delta$ with an overhead of $poly(\delta^{-1})$. Converting the estimate of $\beta_k^{i,j}/|S_k^i|$ to an estimate of $\beta_k^{i,j}$ reintroduces the superpolynomial factor $|S_k^i|$, and eliminates the previously claimed~\cite{lloyd2016quantum, ubaru2021quantum,hayakawa2021quantum,ameneyro2022quantum} exponential speedup when one needs to compute these quantities to constant additive error. A similar normalization issue arises in quantum algorithms for the approximating the Jones polynomial~\cite{aharonov2006polynomial, shor2007estimating} and other \#P problems~\cite{bordewich2009approximate}.

We can assess the potential polynomial speedup of the quantum algorithm over classical algorithms. Let us consider $k=3$, which is sufficiently low-dimensional to be practically relevant, but is large enough so that it is already challenging to compute classically. We will assume that the gap parameters $\Lambda$ are $\Omega(1/polylog(N))$, and note that $1/\Lambda \sim \mathcal{O}(poly(N))$ will reduce the polynomial speedup claimed below. Assuming a dense complex with $|S_{k+1}^j| \sim \binom{N}{k+1}$ this results in a worst-case scaling of approximately $N^{10}$ classically~\cite{milosavljevic2011zigzag}. Our quantum algorithm scales as $  N^{3.5} $, with the additional dependencies on the precision and gap parameters noted in Table~\ref{tab:PriorAlgorithmComp}. Subject to these dependencies, this is an almost cubic speedup in $N$. For larger values of $k$, subject to the dependencies on the gap parameters, the quantum speedup approaches quintic: $\binom{N}{k+1}^{0.5}$ vs $\binom{N}{k+1}^{\omega}$. Nevertheless, in practice the classical algorithm often scales as $\mathcal{O}\left( |S_{k+1}^j| \right)$~\cite{milosavljevic2011zigzag}, or this scaling can be achieved using heuristic sparsification methods~\cite{mischaikow2013morse}. Moreover, our quantum-inspired power method also scales linearly in the number of simplices, with an identical gap dependence as the quantum algorithm. Hence, even if existing classical heuristics do not work, the quantum algorithm is limited to at most a  quadratic speedup over classical approaches. 

Based on the discussion above, the only prospect for large polynomial (or even superpolynomial) quantum speedups appears to be the case where $\beta_k^{i,j}$ is sufficiently large (compared to $|S_k^i|$) that estimating $\beta_k^{i,j}/|S_k^i|$ to relative error (equivalent to estimating $\beta_k^{i,j}$ to relative error) is a meaningful task. In Ref.~\cite{cade2021complexity} it was shown that estimating the normalized quasi-Betti numbers (which accounts for miscounting low-lying but non-zero singular values) of general cohomology groups is DQC1-hard. The hardness of estimating normalized (persistent) Betti numbers of a clique complex, subject to a gap assumption, which is the problem solved by existing quantum algorithms, (or alternatively, the hardness of estimating the normalized quasi-Betti numbers of a clique complex if we do not demand a gap) has not been established (see Ref.~\cite{cade2021complexity}~Sec.~1.1). Even if DQC1-hardness can be proven for this problem, there is currently no evidence that real-world datasets possess sufficiently large (persistent) Betti numbers that such a calculation would be useful. For example, the largest known bounds on $\beta_k^i$ of a Vietoris-Rips complex are $\mathcal{O}(N^k)$~\cite{goff2011extremal} (though this bound is obtained inductively, and so may be loose). That work also provides a construction with $\beta_k = \Omega(N^{(k+1)/2})$. As discussed in Ref.~\cite{schmidhuber2022complexity, berry2022quantifying} this leads to an approximately quartic speedup for the quantum algorithm, compared to classical algorithms scaling as $\mathcal{O}\left( \binom{N}{k+1} \right)$. For complexes constructed from points sampled randomly from an underlying probability distribution on $\mathbb{R}^d$, the Betti numbers are much smaller; $\beta_k^i \in \mathcal{O}(N)$, on average~\cite{Kahle2011RandomComplexes,bobrowski2014RandomComplexesSurvey}, which precludes large polynomial speedups, even for relative error estimates. We note that a family of graphs (not constructed from a Vietoris-Rips complex) were recently presented in Ref.~\cite{berry2022quantifying} which have a sufficiently large Betti numbers and gap that quantum algorithms can achieve superpolynomial advantage for relative error Betti number estimation. Nevertheless, as these graphs are artificially constructed, such that their Betti numbers are known \textit{a priori}, it is unlikely that similar speedups will be present in real-world datasets. 

As discussed earlier, randomized classical algorithms were recently introduced~\cite{berry2022quantifying,apers2022simple} that can also compute normalized (non-persistent) Betti numbers in polynomial time. These algorithms have worse gap/precision dependence than the quantum algorithm, but are efficient in the number of simplices. Hence, this further restricts the potential regime of advantage of quantum algorithms for topological data analysis.

\section{Conclusion}\label{Sec:Conclusion}
We have developed a streamlined quantum algorithm for computing persistent Betti numbers, and analyzed its complexity. Our algorithm uses the QSVT framework to tackle the problem in a quantum native way, which enables improved complexities compared to existing quantum algorithms. We introduce a new compact mapping from simplices to qubits, that provides an exponential reduction in spatial resources for $ k = \mathcal{O} \left( \mathrm{polylog}(N) \right)$. By compiling all steps of our algorithm down to primitive gates, we can provide a fair comparison to classical algorithms. We see that when the caveats of all known quantum algorithms are taken into account, they can achieve, at best, a gap-dependent polynomial speedup over classical algorithms for the salient problem of estimating low-dimensional persistent Betti numbers to constant additive error -- in contrast to previous claims. The polynomial speedup of the quantum algorithm is contingent upon the as yet unknown dependence of the gap parameters on $N$. Moreover, we introduced a quantum-inspired classical algorithm that has better asymptotic dependence on the number of simplices than existing classical algorithms, and which limits the asymptotic speedup of our quantum algorithm to quadratic.

The compact mapping introduced in this work provides an exponential space saving over existing quantum and classical methods for the practically relevant regime of $k = \mathcal{O} \left( \mathrm{polylog}(N) \right)$. This exponential reduction in logical qubits over existing quantum algorithms is significant, particularly for the foreseeable future when logical qubit count will be a valuable resource. In fact, our algorithm achieves an exponential space saving over all classical algorithms for computing persistent Betti numbers\footnote{We note that recent work~\cite{apers2022simple} provides a classical algorithm for estimating normalized Betti numbers with $\mathcal{O}(N)$ classical space complexity, the same as the classical memory required for our quantum algorithm. However, such an algorithm has not yet been formulated for persistent Betti numbers.}.  

Our algorithm provides an efficient and universal lens through which to investigate topological data analysis on quantum computers. Future work could determine its practicality by considering concrete resource estimates for problems of interest. Quantum advantage would be made more likely by identifying scenarios where high-dimensional persistent Betti numbers are of interest, or where computing the normalized persistent Betti numbers would be meaningful. \\ 

\section{Acknowledgements}
We thank Alex~Dalzell and Eric~Peterson for productive discussions on various aspects of this work, and Fernando Brand\~ao for discussions and support throughout the project. A.G. acknowledges funding from the AWS Center for Quantum Computing. M.B. is supported by the EPSRC (Grant number EP/W032643/1).


\appendix
\onecolumngrid

\section{Pedagogical introduction to topological data analysis}\label{AppSec:MathsBackgroundPedagogy}

In this section, we provide a more pedagogical introduction to the mathematical background and notation used in topological data analysis, to build on the introduction given in the main text. This area is known as simplicial persistent homology. We will first discuss the case of computing regular Betti numbers, before moving on to persistent Betti numbers. We refer readers to Ref.~\cite{carlsson_vejdemo-johansson_2021} for a more detailed and rigorous treatment of this area. For pedagogical introductions to algebraic topology and computing persistent homology, we recommend the following review~\cite{otter2017roadmap}, notes~\cite{web:TopazNotes}, popular science article~\cite{feng2021Connecting}, and video lectures~\cite{web:RieckLecture, web:FengLecture}. We will only focus on the computation of persistent homology for point-cloud data, and do not consider extensions to graphs or pixel-based data~\cite{otter2017roadmap}.

\subsection{Computing Betti numbers}\label{AppSubSec:Betti}

\begin{center}
\begin{tabular}{c|c}
    Symbol & Meaning \\ \hline
     $\sigma_k, \tau_k, s_k$ & Individual $k$-simplices \\
     $S^i$ & Simplicial complex at scale $i$ defined by length scale $\mu_i$ \\
     $S_k^i$ & The set of $k$-simplices in $S^i$ \\
     $C_k(S^i)$ & The $k$-th chain group of $S_k^i$ (here, over the coefficients $\{0, \pm1 \}$) \\
     $\partial_k^i$ & The $k$-th boundary operator restricted to $k$-simplices in the complex at scale $i$ \\
     $\mathrm{H}_k^i$ & The $k$-th Homology group (i.e. group of $k$-holes) at scale $i$ \\
     $\beta_k^i$ & The $k$-th Betti number (i.e. number of $k$-dimensional holes) at scale $i$ \\
     $\Delta_k^i$ & The $k$-th combinatorial Laplacian at scale $i$ \\
\end{tabular}
\end{center}

We will assume that our data consists of $N$ points distributed in $\mathbb{R}^d$, sampled from an underlying topological manifold -- for example, $N=5$ in $\mathbb{R}^{2}$:\\
\begin{center}
\begin{tikzpicture}
[align=center,node distance=3cm]
\filldraw [black]  (0,0) circle (3pt);
\node [left] at (0,0) {A};

\filldraw [black]  (3,0) circle (3pt); 
\node [right] at (3,0) {B};

\filldraw [black]  (3,-2) circle (3pt); 
\node [right] at (3,-2) {C};

\filldraw [black]  (0,-2) circle (3pt);
\node [left] at (0,-2) {D};

\filldraw [black]  (1.5,-3) circle (3pt); 
\node [below] at (1.5,-3) {E};
\end{tikzpicture}
\end{center}

These datapoints will also be referred to as vertices, or 0-simplices. The datapoints have an associated ordering; for example, we choose to order the datapoints alphabetically by their labels. We define a length scale $\mu$, and connect vertices by an edge if they are within $\mu$ of each other. For example, at two values of $\mu$, we obtain:\\
\begin{center}
\begin{tikzpicture}
[align=center,node distance=3cm]

\node [above] at (1.5, 0.5) {$\mu_1$};

\filldraw [black]  (0,0) circle (3pt);
\node [left] at (0,0) {A};

\filldraw [black]  (3,0) circle (3pt); 
\node [right] at (3,0) {B};

\filldraw [black]  (3,-2) circle (3pt); 
\node [right] at (3,-2) {C};

\filldraw [black]  (0,-2) circle (3pt);
\node [left] at (0,-2) {D};

\filldraw [black]  (1.5,-3) circle (3pt); 
\node [below] at (1.5,-3) {E};

\draw (0,0) --(0,-2);
\draw (3,0) --(3,-2);
\draw (0,-2) --(1.5,-3);
\draw (3,-2) --(1.5,-3);

\node [above] at (11.5, 0.5) {$\mu_2$};

\filldraw [black]  (10,0) circle (3pt);
\node [left] at (10,0) {A};

\filldraw [black]  (13,0) circle (3pt); 
\node [right] at (13,0) {B};

\filldraw [black]  (13,-2) circle (3pt); 
\node [right] at (13,-2) {C};

\filldraw [black]  (10,-2) circle (3pt);
\node [left] at (10,-2) {D};

\filldraw [black]  (11.5,-3) circle (3pt); 
\node [below] at (11.5,-3) {E};

\draw (10,0) --(10,-2);
\draw (13,0) --(13,-2);
\draw (10,-2) --(11.5,-3);
\draw (13,-2) --(11.5,-3);
\draw (10,0) -- (13,0);
\draw (10,-2) -- (13,-2);

\end{tikzpicture}
\end{center}

This creates a graph for each value of $\mu$. We use the so-called \textit{clique complex} of the graph, where simplices correspond to the cliques in the graph: a $(k+1)$-clique defines a $k$-simplex. This way every $(k+1)$ mutually connected vertices induce a $k$-simplex, which we can think of as the convex hull of the corresponding data points. For example:\\

\begin{center}
\begin{tabular}{|c|c|c|}
    \hline
     Geometric realisation & Simplex & Examples \\ \hline
    \begin{tikzpicture}
     \filldraw [black]  (0,0) circle (2pt);
     \end{tikzpicture}
     & 0-Simplex  & $A, B$  \\ \hline
    \begin{tikzpicture}
     \filldraw [black]  (0,0) circle (2pt);
     \filldraw [black]  (0.5,0) circle (2pt);
     \draw (0,0) -- (0.5,0);
     \end{tikzpicture}
     & 1-Simplex  & $BC, CD$  \\ \hline
    \begin{tikzpicture}
     \filldraw [black]  (0,0) circle (2pt);
     \filldraw [black]  (0.5,0) circle (2pt);
     \filldraw [black]  (0.25,0.433) circle (2pt);
     \draw [black, fill=gray] (0,0) -- (0.5,0) -- (0.25,0.433) -- (0,0);
     \end{tikzpicture}
     & 2-Simplex  & $CDE$  \\ \hline
\end{tabular}
\end{center}

A $k$-simplex has cardinality $k+1$ (it contains $k+1$ vertices), and is of dimension $k$. For two simplices $\tau, \sigma$, if $\tau \subset \sigma$, then we say $\tau$ is a face of $\sigma$. For example, the 1-simplex $CD$ is a face of the 2-simplex $CDE$. Simplices inherit the ordering of the vertices. We ascribe an orientation to the simplices based on the ordering of the vertices. An odd permutation of the vertices changes the orientation of the simplex; for example, $DA = -AD$. A set of simplices (and all of their constituent faces) is referred to as a simplicial complex $S$. The dimension of a simplicial complex is the highest dimension of a simplex in the complex. We define the number of simplices in the complex as $|S|$, and the number of $k$-simplices in the complex as $|S_k|$. As the graphs above only contain vertices and edges (0 and 1-simplices), they are 1-dimensional simplicial complexes, known as the 1-skeleton of the complex. We can then add to the complex all $k$-simplices represented by $(k+1)$-cliques in the graph:\\

\begin{center}
\begin{tikzpicture}
[align=center,node distance=3cm]

\node [above] at (1.5, 0.5) {$S^{\mu_1} = \{A, B, C, D, E, AD, BC, DE, CE \}$ \\
$|S^{\mu_1}_0| = 5$, $|S^{\mu_1}_1| = 4$, $|S^{\mu_1}_2| = 0$};

\filldraw [black]  (0,0) circle (3pt);
\node [left] at (0,0) {A};

\filldraw [black]  (3,0) circle (3pt); 
\node [right] at (3,0) {B};

\filldraw [black]  (3,-2) circle (3pt); 
\node [right] at (3,-2) {C};

\filldraw [black]  (0,-2) circle (3pt);
\node [left] at (0,-2) {D};

\filldraw [black]  (1.5,-3) circle (3pt); 
\node [below] at (1.5,-3) {E};

\draw (0,0) --(0,-2);
\draw (3,0) --(3,-2);
\draw (0,-2) --(1.5,-3);
\draw (3,-2) --(1.5,-3);

\node [above] at (11.5, 0.5) {$S^{\mu_2} = \{A, B, C, D, E, AD, BC, DE, CE, AB, CD, CDE \}$\\
$|S^{\mu_2}_0| = 5$, $|S^{\mu_2}_1| = 6$, $|S^{\mu_2}_2| = 1$};

\filldraw [black]  (10,0) circle (3pt);
\node [left] at (10,0) {A};

\filldraw [black]  (13,0) circle (3pt); 
\node [right] at (13,0) {B};

\filldraw [black]  (13,-2) circle (3pt); 
\node [right] at (13,-2) {C};

\filldraw [black]  (10,-2) circle (3pt);
\node [left] at (10,-2) {D};

\filldraw [black]  (11.5,-3) circle (3pt); 
\node [below] at (11.5,-3) {E};

\draw (10,0) --(10,-2);
\draw (13,0) --(13,-2);
\draw (10,0) -- (13,0);

\draw [black, fill=gray] (10,-2) -- (11.5,-3) -- (13,-2);

\end{tikzpicture}
\end{center}

This complex is known as a Vietoris-Rips complex, and provides an approximation to the underlying topological space whose homology we wish to compute. We refer the reader to Ref.~\cite{otter2017roadmap} for discussion of other types of simplicial complexes that can be used to approximate the topological space. The maximum number of $k$-simplices is $\binom{N}{k+1}$. We represent the $k$-simplices by orthonormal basis vectors.

Under addition with field coefficients, these vectors form an abelian group. We define the $k$-th chain group of the complex $C_k(S)$, over field coefficients $\{-1, 0, 1\}$ (i.e. the group $\mathbb{Z}_3$ defined over addition modulo 3). 

The group elements are linear combinations of $k$-simplices with coefficients $\pm 1$, referred to as $k$-chains. For example:\\
\begin{align}
    AB + CD &\in C_1(S^{\mu_1}) \\
    AB + BC + CD - AD &\in C_1(S^{\mu_2}) \nonumber
\end{align}
The simplicial complex that we have constructed acts as a scaffold that approximates the underlying topological space of interest. Promoting the simplices in the complex to vectors in a vector space enables the use of linear algebraic tools to perform the computation of simplicial homology. We are trying to identify holes in the complex; a hole is a region of empty space, demarcated by its boundary. As studying these boundaries is crucial for finding holes, we define a boundary operator.

The boundary operator $\partial$ is a linear map which sends $k$-chains to their oriented boundaries. The action of the boundary operator on a $k$-simplex in the complex at the given scale $i$ is
\begin{align}
    \partial [v_0, ..., v_k] = \sum_{l=0}^k (-1)^l [v_0, ... \hat{v}_l, ... v_k]
\end{align}
where $v_0 ... v_k$ are the ordered vertices in the $k$-simplex, and $\hat{v}_l$ means the vertex is excluded from the simplex. We denote by $\partial_k^i\colon \langle S_k^i \rangle \rightarrow \langle S_{k-1}^i\rangle$ the boundary operator on the complex at length scale $i$ restricted to the subspaces $\langle S_k^i \rangle$ and $\langle S_{k-1}^i\rangle$ spanned by the basis vectors corresponding to $k$ and $k-1$ simplices present at scale $i$.

The boundary operator has the following actions on simplices in our example complex
\begin{align}
    \partial_1 [AB] &= B - A \\
    \partial_2 [CDE] &= DE - CE + CD \nonumber
\end{align}

A closed cycle of $k$-simplices in the complex can surround either empty space, or can surround a chain of $(k+1)$-simplices in the complex. The former case corresponds to a hole, while the latter case corresponds to a hole that has been filled-in by $(k+1)$-simplices. An example of the former case in our example complex is the $1$-cycle $AB + BC + CD - AD$ (where the minus sign appears due to $DA = -AD$). An example of the latter case in our example complex is the $1$-cycle $DE - CE + CD = \partial_2 [CDE]$. Linear combinations of these two types of $k$-cycles are also valid $k$-cycles. All $k$-cycles are themselves boundaryless (when applying the boundary operator to a closed cycle of $k$-simplices, each $(k-1)$-face appears twice, each time with opposite signs). As a result, $k$-cycles are objects in the kernel of $\partial$. 

To identify the holes in the complex, we consider all possible $k$-cycles, and remove those that are the $k$-boundaries of a $(k+1)$-chain in the complex. These $k$-boundaries are the $k$-chains in the image of $\partial$. As both the $k$-cycles and $k$-boundaries form subgroups of $C_k(S^i)$, we can remove the filled-in holes using quotient groups (this also eliminates the issue of linear combinations of $k$-holes and $k$-boundaries being valid $k$-cycles - such objects are said to be homologous to the hole in question - and either the hole, or a homologous $k$-cycle can be chosen to represent the hole). The $k$-th homology group is defined as the group of $k$-holes in the complex
\begin{equation}
    \mathrm{H}_k^i = \mathrm{Ker}(\partial_k^i)~/~ \mathrm{Im}(\partial_{k+1}^i)
\end{equation}
and the $k$-th Betti number is the rank of the $k$-th homology group 
\begin{align}
\beta_k^i &= \mathrm{rank}\left( \mathrm{H}_k^i \right) \\
&= \mathrm{dim}\left(\mathrm{Ker}(\partial_k^i)\right) - \mathrm{dim}\left(\mathrm{Im}(\partial_{k+1}^i)\right). \nonumber
\end{align} 
$\beta_0^i$ gives the number of connected components in $S^i$, $\beta_1^i$ gives the number of holes, $\beta_2^i$ gives the number of voids, and so on. For our example complex above, $\beta_0(S^{\mu_2}) = 1$, and $\beta_1(S^{\mu_2}) = 1$, so the complex has one connected component, and a single 1-dimensional hole. \\

As discussed in Ref.~\cite{lim2020hodge}, it can be undesirable to have a number of equivalent homologous choices for the representative of the hole. The `harmonic representative' is a chain that is orthogonal to the chains in $\mathrm{Im}(\partial_{k+1}^i)$. The group of harmonic representatives is given by $\mathrm{Ker}(\Delta_k^i)$, where
\begin{equation}
    \Delta_k^i = \partial_{k+1}^i (\partial_{k+1}^i)^\dag + (\partial_k^i)^\dag \partial_k^i
\end{equation}
is known as the $k$-th combinatorial Laplacian of the simplicial complex, which is a higher-order analogue of the graph Laplacian. It can be shown that~\cite{lim2020hodge}
\begin{equation}
    \beta_k^i = \mathrm{dim}\left(\mathrm{Ker}(\Delta_k^i)\right).
\end{equation}
This formulation of the problem is widely used in a number of quantum algorithms~\cite{lloyd2016quantum,gunn2019review,gyurik2020towards, ubaru2021quantum}. Note that the operator $(\partial_k^i)^\dag$ maps $C_{k-1}(S^i) \rightarrow C_k(S^i)$. The combinatorial Laplacian is Hermitian, $(k+1)$-sparse (see Ref.~\cite{goldberg2002combinatorial}, Theorem 3.3.4),
and its entries are bounded by $N$~\cite{gyurik2020towards}. As discussed above, the harmonic representatives are not necessarily in the `obvious' form. For our example complex, we find that $\mathrm{Ker}\left(\Delta_1(S^{\mu_2})\right) = -3(AB +BC +CD -AD) + \partial_2(CDE)$, rather than the obvious choice $(AB +BC +CD -AD)$.
Some of the quantum algorithms~\cite{lloyd2016quantum,gunn2019review} also consider a `$k$-th Dirac operator' that when squared, embeds the $k$-th combinatorial Laplacian in one of its sub-blocks.\\

The above procedures can be used for computing the Betti numbers for a fixed length scale $\mu_i$ that determines the simplicial complex at scale $i$. Unfortunately, it is not obvious \textit{a priori} how the length scale should be chosen. Moreover, computing the Betti number for different values of $\mu$ does not provide information about the persistence of topological features beyond $k=0$. As discussed in Ref.~\cite{neumann2019limitations}, the Betti numbers do not uniquely identify topological features. For example, if there was one $1$-hole created at scale $\mu_1$, which is destroyed at scale $\mu_2$, at the same time as a different $1$-hole is created, then we would find $\beta_1(S^{\mu_1}) = \beta_1(S^{\mu_2}) = 1$, from which we would incorrectly infer that the topological feature created at $\mu_1$ persists at length scale $\mu_2$. In the following subsection we discuss how the methods introduced above can be extended to compute persistent Betti numbers.

\subsection{Computing persistent Betti numbers}\label{AppSubSec:PersistentBetti}

\begin{center}
\begin{tabular}{c|c}
    Symbol & Meaning \\ \hline
     $\mathrm{H}^{i,j}_k$ & The $(j-i)$-th persistent $k$-th Homology group \\
     $\beta^{i,j}_k$ & The $(j-i)$-th persistent $k$-th Betti number (i.e. \# $k$-holes present at $i$ still present at $j$) \\
     $C_{k+1}^{i,j}(S^j)$ & The subgroup of $(k+1)$-chains in $C_{k+1}(S^j)$ mapped to $k$-chains in $C_k(S^i)$ by $\partial_{k+1}^j$. \\
     $\partial_{k+1}^{i,j}$ & The boundary operator $\partial_{k+1}$ restricted to elements of $C_{k+1}^{i,j}(S^j)$. \\
     $\Delta^{i,j}_k$ & The $(j-i)$-th persistent $k$-th combinatorial Laplacian
\end{tabular}
\end{center}

In order to compute the persistent homology of the complex $S$, we consider a nested sequence of increasingly dense simplicial complexes -- known as a filtered simplicial complex
\begin{equation}
    S^0 \subset S^1 \subset ... \subset S^L = S
\end{equation}
This filtration is determined by the scale $\mu_i$; as $\mu_i$ is increased additional simplices enter the complex. If two simplices enter the complex at the same time, then we consider the lower dimensional simplex to enter first. If the two simplices are of the same dimension, an order can be assigned arbitrarily. For example, we have the following filtration:\\

\begin{center}
\begin{tikzpicture}
[align=center,node distance=3cm]

\node at (4.5, -1) {$\subset$};
\node at (10.5, -1) {$\subset$};

\node [above] at (1.5, 0.5) {$S^{0} = \{A, B, C, D, E \}$ };

\filldraw [black]  (0,0) circle (3pt);
\node [left] at (0,0) {A};

\filldraw [black]  (3,0) circle (3pt); 
\node [right] at (3,0) {B};

\filldraw [black]  (3,-2) circle (3pt); 
\node [right] at (3,-2) {C};

\filldraw [black]  (0,-2) circle (3pt);
\node [left] at (0,-2) {D};

\filldraw [black]  (1.5,-3) circle (3pt); 
\node [below] at (1.5,-3) {E};

\node [above] at (7.5, 0.5) {$S^{1} = \{A, B, C, D, E,$\\
$AD, BC, CE, DE\}$ };

\filldraw [black]  (6,0) circle (3pt);
\node [left] at (6,0) {A};

\filldraw [black]  (9,0) circle (3pt); 
\node [right] at (9,0) {B};

\filldraw [black]  (9,-2) circle (3pt); 
\node [right] at (9,-2) {C};

\filldraw [black]  (6,-2) circle (3pt);
\node [left] at (6,-2) {D};

\filldraw [black]  (7.5,-3) circle (3pt); 
\node [below] at (7.5,-3) {E};

\draw (6, 0) -- (6, -2);
\draw (9, 0) -- (9, -2);
\draw (6, -2) -- (7.5, -3);
\draw (9, -2) -- (7.5, -3);

\node [above] at (13.5, 0.5) {$S^{2} = \{A, B, C, D, E,$ \\
$AD, BC, DE, CE, AB, CD, CDE \}$};

\filldraw [black]  (12,0) circle (3pt);
\node [left] at (12,0) {A};

\filldraw [black]  (15,0) circle (3pt); 
\node [right] at (15,0) {B};

\filldraw [black]  (15,-2) circle (3pt); 
\node [right] at (15,-2) {C};

\filldraw [black]  (12,-2) circle (3pt);
\node [left] at (12,-2) {D};

\filldraw [black]  (13.5,-3) circle (3pt); 
\node [below] at (13.5,-3) {E};

\draw (12,0) --(12,-2);
\draw (15,0) --(15,-2);
\draw (12,0) -- (15,0);

\draw [black, fill=gray] (12,-2) -- (13.5,-3) -- (15,-2);

\end{tikzpicture}
\end{center}

As above, we denote the $k$-th boundary operator restricted to $k$-simplices in $S^i$ as $\partial^i_k$. The $k$-holes in complex $S^i$ may become `filled in' by new $(k+1)$-simplices present in the later complex $S^j$. There may also be new $k$-holes in $S^j$ that were not in $S^i$. We can count the holes in $S^i$ that are still present in $S^j$ by:
\begin{enumerate}
    \item Counting all the $k$-holes and $k$-boundaries in $S^i$ (i.e. the $k$-cycles $\mathrm{Ker}(\partial_k^i)$)
    \item Removing those $k$-cycles in $S^i$ that are $k$-boundaries in $S^j$ (this is given by $\mathrm{Ker}(\partial_k^i) \cap \mathrm{Im}(\partial_{k+1}^j)$~\cite{Edelsbrunner2002})
\end{enumerate}
We do not have to worry about double-counting the $k$-boundaries in $S^i$ and $S^j$, as all $k$-boundaries in $S^i$ are also $k$-boundaries in $S^j$, because all simplices in $S^i$ are also in $S^j$. Formally, we can define the $(j-i)$-th persistent $k$-th homology group as~\cite{Edelsbrunner2002}
\begin{equation}
    \mathrm{H}_k^{i,j} := \mathrm{Ker}(\partial_k^i)~/~\left(\mathrm{Ker}(\partial_k^i) \cap \mathrm{Im}(\partial_{k+1}^j) \right)
\end{equation}
The $(j-i)$-th persistent $k$-th Betti number is the rank of this group,
\begin{align}\label{Appeq:PersistentBetti}
    \beta_k^{i,j} &= \mathrm{rank}\left(\mathrm{H}_k^{i,j} \right) \\
    &= \mathrm{dim}\left( \mathrm{Ker}(\partial_k^i) \right) - \mathrm{dim}\left( \mathrm{Ker}(\partial_k^i) \cap \mathrm{Im}(\partial_{k+1}^j) \right) \nonumber.
\end{align}
As discussed in the main text, our quantum algorithm uses this expression for computing $\beta_k^{i,j}$, by estimating the size of these subgroups. \\

Recently, the task of computing persistent Betti numbers has been reformulated in terms of a persistent combinatorial Laplacian~\cite{wang2020persistent, memoli2020persistent}. Although our quantum algorithm does not proceed through this approach, the approach of Ref.~\cite{hayakawa2021quantum} does, and it is instructive to compare this approach with our own. One first defines a subgroup $C^{i,j}_{k+1}(S^j)$ containing $(k+1)$-chains in the complex at scale $S^j$ whose images under the boundary operator $\partial^j_{k+1}$ are $k$-chains at scale $S^i$. Formally,
\begin{equation}
    C^{i,j}_{k+1}(S^j) = \{c \in C_{k+1}(S^j) : \partial^j_{k+1}(c) \in C_k(S^i) \}.
\end{equation}
Intuitively, the chains in $C^{i,j}_{k+1}(S^j)$ are mapped to $k$-cycles in $C_k(S^i)$, which can be divided into the $k$-boundaries already present in $C_k(S^i)$, and a subset of the $k$-holes in $C_k(S^i)$. In other words, the additional $(k+1)$-chains in $C^{i,j}_{k+1}(S^j)$ (beyond those already present in $C_{k+1}(S^i)$) have the action of filling-in $k$-holes in $C_k(S^i)$. Defining $\partial^{i,j}_{k+1}$ as $\partial$ restricted to the chains in $C^{i,j}_{k+1}(S^j)$ it follows that~\cite{memoli2020persistent}
\begin{equation}
    \mathrm{Im}(\partial^{i,j}_{k+1}) \cong \mathrm{Ker}(\partial_k^i) \cap \mathrm{Im}(\partial_{k+1}^j),
\end{equation}
In Refs.~\cite{wang2020persistent, memoli2020persistent} the $(j-i)$-th persistent $k$-th combinatorial Laplacian is defined as
\begin{equation}\label{Eq:PersistentLaplacian}
    \Delta^{i,j}_{k} = \partial^{i,j}_{k+1} \circ (\partial^{i,j}_{k+1})^\dag + (\partial^{i}_k)^\dag \partial^{i}_k.
\end{equation}
It is shown in Ref.~\cite{memoli2020persistent} that $\beta_k^{i,j} = \mathrm{dim}\left(\mathrm{Ker}(\Delta^{i,j}_k)\right)$. We have used $\circ$ in Eq.~(\ref{Eq:PersistentLaplacian}) to highlight that the expression of $\partial^{i,j}_{k+1}$ in terms of matrices is more complex than expressing $\partial_k^i$, as discussed in Ref.~\cite{memoli2020persistent}. We elaborate more on this point in Appendix~\ref{AppSub:BettiNuances}. The reason for this complexity is that the restricted boundary operator $\partial^{i,j}_{k+1}$ must be expressed in a basis of the chains in $C^{i,j}_{k+1}(S^j)$, rather than in the basis of $S_{k+1}^j$. This change of basis adds complexity to both classical~\cite{memoli2020persistent} and quantum algorithms~\cite{hayakawa2021quantum}. A similar phenomenon underpins the additional round of QSVT required in our algorithm to convert $\Pi_{\mathrm{Ker}} \Pi_{\mathrm{Im}}$ to $\Pi_{\mathrm{Ker} \cap \mathrm{Im}}$, discussed in Sec.~\ref{Subsec:SubspaceProj}.


\section{Proofs of Theorems, Corollaries and Lemmas}\label{App:TheoremProofs}

\subsection{Theorem~\ref{Theorem:SubspaceDimensionEst}: Formal version and proof}\label{AppSub:DimensionEstProof}
We will prove a formal version of Theorem~\ref{Theorem:SubspaceDimensionEst}, which we state here:

\begin{theorem}\label{Theorem:SubspaceDimensionEstFormal}
	(Normalized projector rank estimation) Given an $(\alpha,a,\delta/4)$-block-encoding $V_{\Pi}$ of $\Pi$ and an approximate state preparation unitary $V_{\psi}\colon \ket{0}^{\otimes b}\mapsto \ket{\tilde{\psi}}$ such that $\nrm{\ket{\tilde{\psi}}-\ket{\psi}}\leq \delta/4$ and $\ket{\psi}$ is a purification of a state proportional to a projector $\widetilde{\Pi}\succeq \Pi$, we can estimate the quantity $\sqrt{\frac{\mathrm{rank}(\Pi)}{\mathrm{rank}(\widetilde{\Pi})}}$ to additive error $\delta$ with success probability $\geq 1- \eta$ using $\bigO{\log(\eta^{-1})}$ incoherent repetitions of a quantum circuit, which makes $\bigO{\frac{\alpha}{\delta}}$ calls to controlled- $V_{\psi}$ and $V_{\Pi}$ (and their controlled inverses), and uses $\bigO{(a+b)\frac{\alpha}{\delta}}$ additional two-qubit gates.\footnote{In case $\alpha \gg 1$ and the cost of implementing $V_{\Pi}$ is less than $\alpha$ times more costly than implementing $V_{\psi}$ then it could be worth reducing the number of uses of $V_{\psi}$ to $\bigO{\frac{1}{\delta}}$ by pre-amplifying the block-encoding $V_{\Pi}$ by uniform eigen(singular)-value amplification~\cite{low2017HamSimUnifAmp,gilyen2019quantum}.}
\end{theorem}

\begin{proof}
	Let us call the purifying system $B$ such that $\Tr_B(\ketbra{\tilde{\psi}}{\tilde{\psi}})=\frac{\widetilde{\Pi}}{\mathrm{rank}(\widetilde{\Pi})}$. 
	We consider the following process: prepare the state $\ket{\tilde{\psi}}$, apply $V_{\Pi}\otimes I_B$ and measure the first $a$ qubits. The probability that we get the all-$0$ outcome is $\nrm{(\bra{0}^{\otimes a}\otimes I)(V_{\Pi}\otimes I_B)(\ket{0}^{\otimes a}\otimes \ket{\tilde{\psi}})}^2$. 
	We employ amplitude estimation with precision $\delta/(2\alpha)$ to estimate the square root of the above probability 
	with success probability at least $2/3$. We repeat this process $\mathcal{O}\left( \log(\eta^{-1})\right)$ times and take the median of the estimates, boosting the success probability of obtaining such a precise estimate to at least $1-\eta$. 
	This estimate, which we denote by $\hat{x}$, suffices for our purpose since 	
	\begin{align*}
		\sqrt{\frac{\mathrm{rank}(\Pi)}{\mathrm{rank}(\widetilde{\Pi})}}
		=\!\sqrt{\frac{\tr{\Pi}}{\mathrm{rank}(\widetilde{\Pi})}}\!
		=\sqrt{\tr{\!\Pi\frac{\widetilde{\Pi}}{\mathrm{rank}(\widetilde{\Pi})}\!}}\!
		=\sqrt{\bra{\psi}(\Pi\otimes\!I_B\!)\ket{\psi}}\! = \alpha \nrm{(\bra{0}^{\otimes a}\otimes I)(V_{\Pi}\otimes I_B)(\ket{0}^{\otimes a}\otimes \ket{\tilde{\psi}})},
	\end{align*}
	where the final equality only holds when $V_\Pi$ is an $(\alpha, a, 0)$ block-encoding of $\Pi$ and $\nrm{\ket{\tilde{\psi}}-\ket{\psi}} = 0$. Taking into account the errors in $V_\Pi$ and $\ket{\tilde{\psi}}$ yields
	\begin{align*}
		&\sqrt{\frac{\mathrm{rank}(\Pi)}{\mathrm{rank}(\widetilde{\Pi})}}-\nrm{(\Pi\otimes I_B)\ket{\tilde{\psi}}}\\		
		&=\nrm{(\Pi\otimes I_B)\ket{\psi}}-\nrm{(\Pi\otimes I_B)\ket{\tilde{\psi}}}\\
		&\leq \nrm{\Pi\otimes I_B}\nrm{\ket{\psi}-\ket{\tilde{\psi}}}\leq \frac{\delta}{4},
	\end{align*}	
	and
	\begin{align*}
		\nrm{(\Pi\otimes I_B)\ket{\tilde{\psi}}}-\kern-12.2mm&\kern12.2mm\alpha\nrm{(\bra{0}^{\otimes a}\otimes I)(V_{\Pi}\otimes I_B)(\ket{0}^{\otimes a}\otimes \ket{\tilde{\psi}})}\\
		\leq&\nrm{\Pi\otimes I_B-\alpha(\bra{0}^{\otimes a}\otimes I)(V_{\Pi}\otimes I_B)(\ket{0}^{\otimes a}\otimes I)}\\
		=&\nrm{\Pi-\alpha(\bra{0}^{\otimes a}\otimes I)V_{\Pi}(\ket{0}^{\otimes a}\otimes I)}\leq \frac{\delta}{4}.
	\end{align*}
	Hence, including these errors and the uncertainty in $\hat{x}$ gives:
	\begin{align*}
	    & \left| \sqrt{\frac{\mathrm{rank}(\Pi)}{\mathrm{rank}(\widetilde{\Pi})}} - \alpha \hat{x} \right| \\
	    &\leq \left| \sqrt{\frac{\mathrm{rank}(\Pi)}{\mathrm{rank}(\widetilde{\Pi})}}-\nrm{(\Pi\otimes I_B)\ket{\tilde{\psi}}} \right| +  \bigg{|} \nrm{(\Pi\otimes I_B)\ket{\tilde{\psi}}}- \alpha \hat{x} \bigg{|} \\
	    &\leq \frac{\delta}{4} + \bigg{|} \nrm{(\Pi\otimes I_B)\ket{\tilde{\psi}}} - \alpha\nrm{(\bra{0}^{\otimes a}\otimes I)(V_{\Pi}\otimes I_B)(\ket{0}^{\otimes a}\otimes \ket{\tilde{\psi}})} \bigg{|} + \alpha \bigg{|} \nrm{(\bra{0}^{\otimes a}\otimes I)(V_{\Pi}\otimes I_B)(\ket{0}^{\otimes a}\otimes \ket{\tilde{\psi}})} -  \hat{x} \bigg{|} \\
	    &\leq \frac{\delta}{4} + \frac{\delta}{4} + \alpha\frac{\delta}{2\alpha} \\
	    &= \delta
	\end{align*}
	Each elementary iteration step in amplitude estimation requires using controlled- $V_{\Pi}$, $V_{\psi}$ and their (controlled) inverses, together with reflection operators around the initial state $\ket{0}^{\otimes (a+b)}$ and the qubits $\ket{0}^{\otimes a}$ deciding the measurement outcome. This gives the claimed gate complexity bounds on the procedure, since the cost of the final Fourier transform is dominated by that of implementing the $\bigO{\alpha/\delta}$ iteration steps. 
\end{proof}

If amplitude estimation via textbook quantum phase estimation is used~\cite{brassard2002AmpAndEst}, we require $n = \mathcal{O}(\log(1/\delta))$ ancilla qubits, which can be a significant overhead if (exponentially) high precision is desired. Instead, one can use single ancilla variants of amplitude estimation, which can have better performance in practice~\cite{Rall2023amplitudeestimation}.

\subsection{Corollary~\ref{Coro:PersistentBetti}: Applying Theorem~\ref{Theorem:SubspaceDimensionEst} to estimate persistent Betti numbers} \label{AppSub:CoroPersistentProof}

As discussed in the main text, we can estimate persistent Betti numbers using the following expression
\begin{align}
    \beta_k^{i,j} = \mathrm{dim}\left( \mathrm{Ker}(\partial_k^i) \right) - \mathrm{dim}\left( \mathrm{Ker}(\partial_k^i) \cap \mathrm{Im}(\partial_{k+1}^j) \right).
\end{align}
We can estimate this quantity, normalized by the number of $k$-simplices in the complex, using the quantum algorithm of Theorem~\ref{Theorem:SubspaceDimensionEst}. Hence, by rescaling the estimation error, we obtain the following corollary of (the formal version of) Theorem~\ref{Theorem:SubspaceDimensionEst}. As $\beta_k^{i,j}$ takes integer values, we can also consider the case $\beta_k^{i,j}=0$. In this instance, we obtain a decision problem, where the goal is to distinguish between $\beta_k^{i,j}=0$ and $\beta_k^{i,j} \neq 0$ with high certainty. In this case, we can solve the problem outlined below for $\beta_k^{i,j}=1$, setting the precision sufficiently small (e.g., $\Delta=1/3$) to distinguish $\beta_k^{i,j}=0$ (running amplitude estimation and obtaining an estimate $ < 2/3$) from $\beta_k^{i,j}=1$.

\begin{corollary}\label{Coro:PersistentBetti}
To estimate $\beta_k^{i,j}$ to additive error $\Delta$ with success probability $\geq 1- \eta$, we solve two instances of normalized projector rank estimation. The first encodes the value of $|S_k^i|/\binom{N}{k+1}$, and the second encodes the value of $\beta_k^{i,j}/|S_k^i|$. Each instance uses $\mathcal{O}\left( \log(\eta^{-1})\right)$ repetitions of a quantum circuit that makes $\mathcal{O}\left( \frac{\alpha_x}{\delta_x} \right)$ calls to the state preparation unitaries $V_{\psi_x}$ (acting on $b_x$ qubits) and the $(\alpha_x, a_x, \delta_x/4)$ block-encodings $V_{\Pi_x}$ of projectors $\Pi_x$ shown in Table~\ref{tab:CorollaryInstanceData}. Each circuit also uses $\bigO{(a_x + b_x)\frac{\alpha_x}{\delta_x}}$ additional two-qubit gates. 
\end{corollary}

\begin{table*}[ht]
    \centering
    \begin{tabular}{c|c|c}
        Instance $x$  & 1 & 2 \\ \hline
        Estimates & $X:= \sqrt{|S_k^i|/\binom{N}{k+1}}$ & $Y:=\sqrt{\beta_k^{i,j}/|S_k^i|}$ \\ \hline
        $\delta_x$ & $\frac{\Delta}{4\beta_k^{i,j}} \sqrt{\frac{|S_k^i|}{\binom{N}{k+1}}}$ & $\frac{\Delta}{4\sqrt{ |S_k^i| \beta_k^{i,j} }}$ \\ \hline
        $\Pi_x$ & \makecell{$\Pi_{S_k^i}$: Projects onto \\ $k$-simplices in complex at scale $i$} & \makecell{$\Pi_{\mathrm{Ker}} - \Pi_{\mathrm{Ker} \cap \mathrm{Im} }$: Projects onto \\ $\mathrm{Ker}(\partial_k^i) - \mathrm{Ker}(\partial_k^i) \cap \mathrm{Im}(\partial_{k+1}^j)$} \\ \hline
        $\ket{\psi_x}$ & \makecell{$\ket{\psi_{S_k}}$: Uniform superposition of \\ all possible $k$-simplices} & \makecell{$\ket{\psi_{S_k^i}}$: Purification of maximally mixed state \\ over all $k$-simplices in complex at scale $i$ }
    \end{tabular}
    \caption{Details of the error parameters and block encoded operators for each of the instances of normalized projector rank estimation used to estimate the persistent Betti number to additive error $\Delta$ using the algorithm of Theorem~\ref{Theorem:SubspaceDimensionEst}.}
    \label{tab:CorollaryInstanceData}
\end{table*}

\begin{proof}
To prove this corollary, we define the following two random variables, which each instance of normalized projector rank estimation will estimate:
\begin{align}
    X &:= \sqrt{\frac{\mathrm{Rank}\left(\Pi_{S_k^i} \right)}{\binom{N}{k+1}}} \\
    Y &:= \sqrt{\frac{\mathrm{Rank}\left(\Pi_{\mathrm{Ker}} - \Pi_{\mathrm{Ker} \cap \mathrm{Im}}\right)}{\mathrm{Rank}\left(\Pi_{S_k^i} \right)}}
\end{align}
We have that $\mathrm{Rank}\left(\Pi_{\mathrm{Ker}} - \Pi_{\mathrm{Ker} \cap \mathrm{Im}}\right) = \beta_k^{i,j}$, $\mathrm{Rank}\left(\Pi_{S_k^i} \right) = |S_k^i|$. Hence
\begin{equation}
    \beta_k^{i,j} = \binom{N}{k+1} X^2 Y^2.
\end{equation}
We define the estimation errors as $\delta_\beta, \delta_x, \delta_y$, respectively. The error in $\beta_k^{i,j}$ is given by
\begin{align}
    \delta_\beta \leq 2 \binom{N}{k+1}(\delta_x XY^2 + \delta_y X^2 Y).
\end{align}
To achieve an error of $\delta_\beta = \Delta$, we set each term above equal to $\frac{\Delta}{2}$. Hence
\begin{align}
    \delta_x &= \frac{\Delta}{4 \beta_k^{i,j}} \sqrt{\frac{|S_k^i|}{\binom{N}{k+1}}} \\
    \delta_y &= \frac{\Delta}{4 \sqrt{\beta_k^{i,j} |S_k^i|}}.
\end{align}

It remains to show that the specified state preparation unitaries and projector block-encodings are sufficient to encode the random variables $X,Y$. We focus first on $X$, which encodes the number of $k$-simplices in the complex at scale $i$, $|S_k^i|$. The state used is the uniform superposition over all possible $k$-simplices
\begin{align}
    \ket{\psi_{S_k}} &= \frac{1}{\sqrt{\binom{N}{k+1}}} \sum_{s_k} \ket{s_k} \\ \nonumber
    &= \frac{1}{\sqrt{\binom{N}{k+1}}} \sum_{s_k \in S_k^i} \ket{s_k} + \frac{1}{\sqrt{\binom{N}{k+1}}} \sum_{s'_k \notin S_k^i} \ket{s'_k}.
\end{align}
A purification is not required in this case, because we can prepare the uniform superposition in the basis in which the projector $\Pi_{S_{k}^i}$ is diagonal. The projector $\Pi_{S_{k}^i}$ projects onto $k$-simplices in the complex at scale $i$
\begin{equation}
    \Pi_{S_{k}^i} = \sum_{s_k \in S_k^i} \ket{s_k}\bra{s_k}.
\end{equation}
We can verify that
\begin{align*}
    \left( \frac{1}{\sqrt{\binom{N}{k+1}}} \sum_{s_k \in S_k^i} \bra{s_k} + \frac{1}{\sqrt{\binom{N}{k+1}}} \sum_{s'_k \notin S_k^i} \bra{s'_k} \right) & \left( \sum_{\sigma_k \in S_k^i} \ket{\sigma_k}\bra{\sigma_k} \right) \left( \frac{1}{\sqrt{\binom{N}{k+1}}} \sum_{\tau_k \in S_k^i} \ket{\tau_k} + \frac{1}{\sqrt{\binom{N}{k+1}}} \sum_{\tau'_k \notin S_k^i} \ket{\tau'_k} \right) \\
    &= \frac{1}{\binom{N}{k+1}} \sum_{s_k, \sigma_k, \tau_k \in S_k^i} \delta_{s_k, \sigma_k} \delta_{\tau_k, \sigma_k} \\
    &= \frac{1}{\binom{N}{k+1}} \sum_{\sigma_k\in S_k^i} 1 \\
    &= \frac{|S_k^i|}{\binom{N}{k+1}}
\end{align*}
as required.\\

For the random variable $Y$ which encodes $\beta_k^{i,j}/|S_k^i|$, we do require the purified initial state, as this is more simple than trying to prepare a superposition over the basis in which $\Pi_{\mathrm{Ker}} - \Pi_{\mathrm{Ker} \cap \mathrm{Im}}$ is diagonal. The initial state is a purified maximally mixed state over $k$-simplices in the complex at scale $i$
\begin{equation}
    \ket{\psi_{S_k^i}} := \frac{1}{\sqrt{|S_k^i|}} \sum_{s_k \in S_k^i} \ket{s_k} \ket{s_k}.
\end{equation}
The orthonormal basis $\{\ket{s_k \in S_k^i} \}$ is unitarily equivalent to the orthonormal basis $\{\ket{\pi_i} \}$ in which $\Pi_{\mathrm{Ker}} - \Pi_{\mathrm{Ker} \cap \mathrm{Im}}$ is diagonal
\begin{equation}
    \Pi_{\mathrm{Ker}} - \Pi_{\mathrm{Ker} \cap \mathrm{Im}} := \sum_{\alpha=1}^{\beta_k^{i,j}} \ket{\pi_\alpha}\bra{\pi_\alpha}.
\end{equation}
The basis $\{\ket{\pi_i} \}$ has dimension $|S_k^i| \geq \beta_k^{i,j}$, and includes all states in the image of $\Pi_{\mathrm{Ker}} - \Pi_{\mathrm{Ker} \cap \mathrm{Im}}$, and additional orthonormal basis states that are in its kernel. We thus have that
\begin{equation}
    \ket{\psi_{S_k^i}} = \frac{1}{\sqrt{|S_k^i|}} \sum_{x=1}^{|S_k^i|} \sum_{\alpha, \beta =1}^{|S_k^i|}  U_{\alpha x} \ket{\pi_\alpha} U_{\beta x} \ket{\pi_\beta}.
\end{equation}
It can be verified that tracing out the second register yields $\widetilde{\Pi} = \frac{1}{|S_k^i|} \sum_{\alpha=1}^{|S_k^i|} \ket{\pi_\alpha}\bra{\pi_\alpha}$. As a result
\begin{align*}
    & \bra{\psi_{S_k^i}} \left(\Pi_{\mathrm{Ker}} - \Pi_{\mathrm{Ker} \cap \mathrm{Im}}\right) \otimes I_B \ket{\psi_{S_k^i}} \\
    &= \tr{\left(\frac{1}{|S_k^i|} \sum_{\alpha=1}^{|S_k^i|} \ket{\pi_\alpha}\bra{\pi_\alpha} \right) \left( \sum_{\gamma=1}^{\beta_k^{i,j}} \ket{\pi_\gamma}\bra{\pi_\gamma} \right) } \\
    &= \frac{1}{|S_k^i|}  \sum_{\alpha=1}^{|S_k^i|}  \sum_{\gamma=1}^{\beta_k^{i,j}} \delta_{\alpha, \gamma} \\
    &= \frac{\beta_k^{i,j}}{|S_k^i|}
\end{align*}
as required. 

The number of calls required to the specified block-encodings and state preparation unitaries follow from Theorem~\ref{Theorem:SubspaceDimensionEst} and the above error bounds.  
\end{proof}


\section{Resource costs}\label{AppSec:ResourceEstimates}

In this section we determine the resource costs of the individual building blocks of our algorithm.

\subsection{Membership oracle construction}\label{AppSub:MembershipOracles}
As discussed in the main text, the algorithm makes regular calls to a membership oracle $O_{m_k}^i$ which determines if a simplex is present in the complex at scale $i$ (or not) based on the positions of the vertices $\{\vec{r}_\alpha\}$, and the length scale $\mu_i$ considered. The membership oracle acts as
\begin{equation}
O_{m_k}^i \ket{s_k} \ket{a} = \ket{s_k} \ket{a \oplus m(s_k)}
\end{equation}
where we have defined the membership function
\begin{equation}
m(s_k) = \begin{cases}
			1 & \text{if $s_k \in S_k^{i}$}  \\
            0 & \text{if $s_k \notin S_k^{i}$} 
\end{cases}.
\end{equation}
We discuss below how the membership oracle is implemented for both the compact and direct mappings.\\

\subsubsection{Compact mapping}\label{AppSubSub:MembershipCompact}
The membership oracle in the compact encoding checks that the vertices present in the simplex have the correct ordering, and that the distances between these vertices are less than the length scale $\mu_i$. To implement the latter step, we will load the positions of the vertices from a quantum lookup table. The quantum lookup table requires only read-only capabilities, and can be implemented using slow loads (and few ancilla qubits) or fast loads (with many ancilla qubits). These extremes have been referred to as `QROM'~\cite{babbush2018encoding} and `QRAM'~\cite{giovannetti2007QuantumRAM} elsewhere in the literature, but for clarity we do not use this terminology here, referring to the extremes as `slow' and `fast' instead. The quantum lookup table implements the following transformation
\begin{equation}
    \frac{1}{\sqrt{D}} \sum_{t=0}^{D-1} \ket{t} \ket{0}^{\otimes b} \rightarrow \frac{1}{\sqrt{D}} \sum_{t=0}^{D-1} \ket{t} \ket{X_t}
\end{equation}
where $X_t$ are $b$-bit data values to be loaded. Loading $D$ pieces of data with the slow lookup table requires $\mathcal{O}(D\log(b))$ time (with only $\mathcal{O}(D)$ non-Clifford gates~\cite{babbush2018encoding}), and $\mathcal{O}(\log(D))$ ancilla qubits. To load $D$ pieces of $b$-bit data with the fast lookup table requires $\mathcal{O}(\log(D))$ time, and $\mathcal{O}(Db)$ ancilla qubits. It is also possible to consider intermediate space-time tradeoffs~\cite{hann2021qram}. 

The quantum lookup table is used to load the position of the vertices in the compactly encoded simplex. As a result, for a set of points in $\mathbb{R}^d$ with each coordinate stored with $b_d$ bits, each piece of loaded data requires $db_d$ bits. We choose $b_d$ large enough to keep over/underflow errors introduced when calculating the squared distance between points sufficiently small. The distance between the datapoints is then calculated coherently, and used to mark any edges in the simplex that are above the length scale $\mu_i$. If such an edge exists, the simplex does not belong in the complex. The distance checking part of the membership oracle is implemented as follows:
\begin{enumerate}
    \item Load the coordinates of all vertices using the quantum lookup table.
    \item Coherently calculate the distances between all pairs of coordinates, and verify that they are less than the length-scale $\mu_i$.
    \item Check that all distances are below threshold using multi-controlled Toffoli gates, and use to flag if the simplex is in the complex.
    \item Uncompute the distances and unload the coordinates.
\end{enumerate}

We first discuss how to implement steps 2 and 3. We have the following state
\begin{equation}
    \ket{V_0} ... \ket{V_{k}} \ket{\vec{r}_0} ... \ket{\vec{r}_k}
\end{equation}
on $(k+1)\log(N+1) + (k+1)d b_d$ qubits. There are $k(k-1)/2 \in \bigO{k^2}$ distances between pairs of vertices to be checked, which we parallelize into $k$ rounds of $k$ checks. We use $\bigO{k b_d}$ ancilla qubits to store the pairwise distances between $k$ pairs of vertices, and if each distance is below the length-scale $\mu_i$. We use a $k$-controlled Toffoli to flag if any of the checked distances is larger than the length-scale. We then uncompute the pairwise distances, and repeat for a different set of pairings. We require $k$ rounds. Using the primitives shown in Table~\ref{tab:DistancePrimitives}, we can calculate the squared distance using $\mathcal{O}(\log(d)\log(b_d) + b_d)$, using $\mathcal{O}(db_d^2)$ ancilla qubits, and compare it to the length scale using $\mathcal{O}(\log(b_d))$ depth and $\mathcal{O}(b_d)$ ancilla qubits. To do $k$ such checks in parallel, we thus require $\bigO{\log(d)\log(b_d) + b_d}$ depth, and $\bigO{k d b_d^2}$ ancilla qubits. The overall complexity for this step is thus $\bigO{k(\log(d)\log(b_d) + b_d)}$ depth, and it uses $\bigO{k d b_d^2}$ ancilla qubits (including the $k$-controlled Toffoli). Once all $k$ rounds of pairwise checks have been performed, we will have $k$-bits that verify if all rounds of checks have passed. A $k$-controlled Toffoli gate on these bits flips the flag of the membership oracle.

\begin{table}[!ht]
    \centering
    \begin{tabular}{c|c|c}
        Primitive & Depth & Ancilla qubits  \\ \hline
        \makecell{Addition/Subtraction \\(two $a$-bit numbers)~\cite{draper2004logarithmic}} & $\log(a)$ & $a$ \\ \hline
        \makecell{Squaring \\ 
        ($a$-bit number)~\cite{Chakrabarti2021thresholdquantum}} & $a$ & $a^2$ \\ \hline
        \makecell{Comparison \\(two $a$-bit numbers)} & $\log(a)$ & $a$ \\ \hline
    \end{tabular}
    \caption{Costs of primitives used for coherently calculating the squared distances between vertices.}
    \label{tab:DistancePrimitives}
\end{table}

To load the state 
\begin{equation}
    \ket{V_0} ... \ket{V_{k}} \ket{\vec{r}_0} ... \ket{\vec{r}_k}
\end{equation}
on $(k+1)\log(N+1) + (k+1)d b_d$ qubits, we use the quantum lookup table. Using a fast load structure, the lookup table load can be performed with $\mathcal{O}(N k d b_d)$ gates, in $\mathcal{O}(\log(N))$ depth, and $\mathcal{O}(N k d b_d)$ ancilla qubits. If instead we opted to use a slow load structure, there are two options available. We could either use $\bigO{k\log(N)}$ ancilla qubits to perform the lookup table load of each coordinate in parallel, with a cost of $\mathcal{O}(N k d b_d)$ gates~\cite{babbush2018encoding} and $\mathcal{O}(N \log(d b_d))$ depth~\cite{low2018trading}. A more qubit efficient approach is to loop through each vertex value ($i=1$ to $N$), and on each iteration:
\begin{itemize}
    \item XOR the value of $i$ with each of the $(k+1)$ vertex registers.
    \item Using a sequence of $d b_d$ many $\log(N+1)$-controlled Toffoli gates (anti-controlled on the vertex register), write the coordinate $\vec{r}_i$ into the corresponding coordinate register. Do this in parallel for each vertex/coordinate.
    \item XOR the value of $i$ with each of the $(k+1)$ vertex registers.
\end{itemize}
This approach uses no ancilla qubits (dirty qubits can be used to implement the multicontrolled Toffoli gates) and has a gate complexity of $\bigO{Nk\log(N) b_d d}$. This complexity can be further reduced to $\bigO{N k \log\log(N+1) b_d d}$ by observing that on average between $i$ and $i+1$ only two bits change, and so
the XOR only needs to be done with these differing bits. This can also speed up the Toffoli implementation (which acts on $\log(N+1)$ qubits) by computing the OR using a tree and then using a CNOT instead of a Toffoli. The advantage is that in the OR tree updating an input bit requires only $\log\log(N+1)$ gates as opposed to the multicontrol Toffoli which requires $\bigO{\log(N)}$ gates to implement. We do not explicitly use this approach in this work, as it is necessary to reformulate other subroutines to use fewer ancilla qubits, in order to make the space saving worthwhile. \\

To verify that the vertices have the correct ordering, we can introduce $\mathcal{O}(k)$ ancilla qubits. We use these to perform comparisons between $V_i$ and $V_{i+1}$ for $i$ even. We perform all $\lceil (k+1)/2 \rceil$ comparisons in parallel. We then do the same for $V_i$ and $V_{i+1}$ for $i$ odd. This requires a circuit depth of $\mathcal{O}(\log\log(N))$, and $\mathcal{O}(k\log(N))$ ancilla qubits. We can use a single $\mathcal{O}(k)$-controlled Toffoli gate to flag if the vertices have the correct ordering. This requires a circuit depth of $\mathcal{O}(k)$. As the order checking can be performed after the distance checking, we can use the same ancilla qubits for both protocols. The overall complexity for the membership oracle, factoring in both steps, is shown in Table~\ref{tab:App:MembershipOracleCosts}.

\begin{table}[!ht]
    \centering
    \begin{tabular}{c|c|c}
         Data loading & Depth & Ancilla qubits  \\ \hline
         Slow & $\bigO{N \log(db_d) + k(\log(d)\log(b_d) + b_d)} $ & $\bigO{k\log(N) + k d b_d^2}$ \\
         Fast & $\bigO{\log(N) + k(\log(d)\log(b_d) + b_d)} $ & $\bigO{N k d b_d + k d b_d^2}$ 
    \end{tabular}
    \caption{Asymptotic complexities of implementing the compact mapped membership oracle using either slow or fast quantum lookup tables to load the positions of the vertices.}
    \label{tab:App:MembershipOracleCosts}
\end{table}

\subsubsection{Direct mapping}\label{AppSubSub:MembershipDirect}
Our implementation of the membership oracle for the direct mapping follows the approach in Ref.~\cite{metwalli2020cliquefinding} for finding cliques in graphs. We classically store a list of edges in the complex, and then iterate through this list, checking if each edge is present in the given $k$-simplex. This can be checked using a Toffoli gate controlled on the vertices of the edge, targeting an ancilla qubit. We then increment a counter, controlled on the value of the ancilla qubit, and then uncompute the ancilla qubit value. After all edges have been sequentially checked, we verify that the counter register contains the correct number of edges for a $k$-simplex ($0.5(k+1)k$), and use this to flag if the simplex is in the complex or not, and uncompute the counter register. In the worst case there are $0.5N(N-1)$ edges in the complex. By introducing $\mathcal{O}(N)$ ancilla qubits, we can check $N/2$ edges in parallel, and add their values (in a tree-like structure) to increment the counter. This reduces the gate depth of the circuit to $\mathcal{O}(N \log(N))$. For edge-sparse complexes, the complexity may be further reduced. A slightly more optimized implementation of the direct mapped membership oracle is presented in Ref.~\cite{berry2022quantifying}.

\subsection{Boundary operator implementation}\label{AppSub:BoundaryOp}
In this section we determine the resource costs to block-encode $\partial_k^i$ and $\partial_{k+1}^j$. A summary of our results is presented in Table~\ref{tab:App:BoundaryOperatorCosts}. Our compact mapped approach proceeds by swapping the vertex to be deleted into the final position of the register, and then setting it to $\ket{\bar{0}}$ to represent the absence of a vertex. Our direct mapped approach relies on a previously observed correspondence between the boundary operator and second quantized fermionic operators~\cite{cade2021complexity, akhalwaya2022efficient, kerenidis2022quantum, ubaru2021quantum}, and efficient circuits for implementing such operators~\cite{akhalwaya2022efficient, kerenidis2022quantum, wan2021exponentially}. Both approaches make calls to the membership oracle, to ensure that they are restricted to simplices of the desired dimension that are present at the specified length scale. 

We restate the definition of a block-encoding here for convenience:
\begin{definition}\label{def:blockDefRepeat}
	We say that a unitary matrix $U$ is an $(\alpha,(a,b),\eps)$-block-encoding of the matrix $A$ if
	\begin{equation}\label{Appeq:blockDef}
		\nrm{A-\alpha (\bra{0}^{\otimes a}\otimes I)U(\ket{0}^{\otimes b}\otimes I)}\leq \eps.
	\end{equation}
	We also use this notion in cases where $(\bra{0}^{\otimes a}\otimes I)U(\bra{0}^{\otimes b}\otimes I)$ acts on larger subspaces than $A$ itself, in which in \Cref{Appeq:blockDef} case by $A$ we mean its trivial embedding\footnote{By the trivial embedding we mean extending $A$ with $0$ matrix elements, so that the non-zero singular values and the corresponding singular-vector pairs are unchanged.} into these larger subspaces. Setting $c=\max\{a,b\}$ we might call it an $(\alpha,c,\eps)$-block-encoding or $c$-qubit block-encoding for brevity. 
\end{definition}

\begin{table*}[t]
    \centering
\begin{tabular}{c|c|c}
            & $(\alpha, m, \epsilon)$ & Gate depth  \\ \hline
    Compact & $(\sqrt{(N+1)(k+1)}, \log(N+1) + \log(k+1) + 1, 0)$ & $\mathcal{O}(k\log\log(N+1))$ \& $1\times O_{m_k^i}$ \\ \hline
    Direct & $(\sqrt{N}, 2, 0)$ & $\mathcal{O}(\log(N))$ \& $1\times O_{m_k^i}, O_{m_{k-1}^i}$ 
\end{tabular}
    \caption{Parameters for the block-encoding $V_{\partial_k^i}$ of $\partial_k^i$ in the compact and direct mappings. The parameters for $V_{\partial_{k+1}^j}$ follow from replacing $k \rightarrow k+1$ and $i \rightarrow j$. The implementation of the membership oracles differ from the compact and direct encodings, as discussed in Sec.~\ref{Subsec:MembershipOracle}.}
    \label{tab:App:BoundaryOperatorCosts}
\end{table*}

\subsubsection{Compact mapping}\label{AppSubSub:BoundaryCompact}

To simplify the analysis, we choose $N$ to be expressible as $2^m - 1$ for integer $m$. We reserve the state $\ket{\bar{0}}$ to express an absence of a vertex. For example, if we had vertices $A, B, C$, we would map these as $A = \ket{01}, B=\ket{10}, C=\ket{11}$.

We focus first on the $k$-th boundary operator, as the $(k+1)$-th will follow as a simple modification of the former. Under the compact mapping, we represent the $k$-th boundary operator at scale $i$ as
\begin{equation}
    \partial_{k}^i = \sum_{\substack{V_0 < ... < V_{k} \\ s_k \in S_k^i}} \sum_{j=0}^{k} (-1)^j \bigg{(} \ket{V_0} ... \ket{V_{j-1}} \ket{V_{j+1}} ... \ket{V_{k}} \bigg{)} \bigg{(} \bra{V_0} ... \bra{V_j} ... \bra{V_{k}} \bigg{)}
\end{equation}
where the ket (output) registers act on $k\log(N+1)$ qubits and the bra (input) act on $(k+1)\log(N+1)$ qubits. This operator projects from a $k$-simplex in the complex at scale $i$ onto the resulting $(k-1)$-chain.

We will show how to implement an $(\alpha, (\log(N+1) + \log(k+1) + 1 , \log(k+1) + 1), 0)$ block-encoding $V_{\partial_{k}^i}$ such that
\begin{equation}
    \left( I_{k\log(N+1)} \otimes \bra{0}^{\otimes \log(N+1)} \otimes \bra{0}  \otimes \bra{0}^{\otimes \log(k+1)}  \right) V_{\partial_{k}^i} \left( I_{(k+1)\log(N+1)} \otimes \ket{0} \otimes \ket{0}^{\otimes \log(k+1)}  \right) = \frac{\partial_{k}^i}{\alpha}
\end{equation}
where $\alpha$ is a normalization factor determined below. Note that this ordering of the qubits means that we don't have the usual interpretation of the matrix being block-encoded in the upper-left block of $V$. This is not an issue, as long as we keep track of which ancilla qubits are being used for the block-encoding. We first note that we can express
\begin{align*}
    & I_{(k+1)\log(N+1)} \otimes \ket{0} \\
    &= \sum_{V_0, ... V_k} \ket{V_0} ... \ket{V_k} \bra{V_0} ... \bra{V_k} \otimes \ket{0} \\
    &= \left( \sum_{\substack{V_0 < ... < V_k \\ s_k \in S_k^i}} \ket{V_0} ... \ket{V_k} \bra{V_0} ... \bra{V_k} + \sum \ket{V_0} ... \ket{V_k} \bra{V_0} ... \bra{V_k} \right)  \otimes \ket{0}
\end{align*}
where the latter sum accounts for $k$-simplices with incorrectly ordered vertices, or that are not present in the complex at scale $i$. Hence, applying the membership oracle $O_{m_k^i}$ and an $X$ gate to this state, we obtain
\begin{align*}
    & X O_{m_k^i} I_{(k+1)\log(N+1)} \otimes \ket{0} \\
    &= \left( \sum_{\substack{V_0 < ... < V_k \\ s_k \in S_k^i}} \ket{V_0} ... \ket{V_k} \bra{V_0} ... \bra{V_k} \right) \otimes \ket{0} + A \otimes \ket{1}
\end{align*}
where we have used the operator $A$ to denote the states that do not represent simplices in the complex at scale $i$. The subsequent operations will not act on the ancilla qubit, and so we can use this operation to restrict the block-encoding to simplices present in the complex. As a result, to show the implementation of $V_{\partial_{k}^i}$, we consider the following operations applied to the (correctly ordered, present in the complex) state $\ket{V_0} ... \ket{V_j} ... \ket{V_{k}} \otimes \ket{0}^{\otimes \log(k+1)}$, and step through the circuit. \\

We first prepare a uniform superposition over the ancilla register, which can be done with $\log(k+1)$ Hadamard gates
\begin{align}
    \frac{1}{\sqrt{k+1}} \sum_{j=0}^{k} \ket{V_0} ...  \ket{V_{k}} \otimes \ket{j}.
\end{align}
Using a work ancilla qubit, we perform a $j$-conditioned permutation that shuffles the $j$-th state into the final register
\begin{align}
    \frac{1}{\sqrt{k+1}} \sum_{j=0}^{k} \ket{V_0} ... \ket{V_{j-1}} \ket{V_{j+1}} ... \ket{V_{k}} \ket{V_j} \otimes \ket{j}.
\end{align}
This circuit requires $k$ controlled-SWAP gates, and $2k$-comparisons of the $\log(k+1)$-bit number $j$. We apply a $Z$ gate to the final bit of $\ket{j}$
\begin{align}
    \frac{1}{\sqrt{k+1}} \sum_{j=0}^{k} (-1)^j \ket{V_0} ... \ket{V_{j-1}} \ket{V_{j+1}} ... \ket{V_{k}} \ket{V_j} \otimes \ket{j}.
\end{align}
We then compute the value of $j$ into another work register, by coherently checking the value of the final register against the values of each neighbouring pairs of registers (e.g. we check $\ket{V_j}$ against $\ket{V_{l+1}}, \ket{V_{l}}$, and so on). Because the registers store simplices in ascending order, we seek the only pair $\ket{V_l}, \ket{V_{l+1}}$ such that $V_l < V_j < V_{l+1}$. We require $\mathcal{O}(k)$ comparisons of $\log(N+1)$ qubit registers, which takes $\mathcal{O}(k\log\log(N+1))$ depth, and $\log(N+1)$ ancilla qubits. This enables us to uniquely identify $j$, which lets us uncompute the initial superposition over $\ket{j}$. We then uncompute the 2nd register $\ket{j}$ using inverse of the inequality checking outlined above. This yields the state
\begin{align}
    \frac{1}{\sqrt{k+1}} \sum_{j=0}^{k} (-1)^j \ket{V_0} ... \ket{V_{j-1}} \ket{V_{j+1}} ... \ket{V_{k}} \ket{V_j} \otimes \ket{0}^{\otimes \log(k+1)}.
\end{align}
We then apply $\log(N+1)$ Hadamard gates to the $\ket{V_j}$ register, resulting in the state
\begin{align}
    \frac{1}{\sqrt{(N+1)(k+1)}} \sum_{j=0}^{k} (-1)^j \ket{V_0} ... \ket{V_{j-1}} \ket{V_{j+1}} ... \ket{V_{k}} \left( \ket{\bar{0}} + \sum_{x=1}^{N} \pm \ket{x} \right) \otimes \ket{0}^{\otimes \log(k+1)}.
\end{align}
We define the above sequence of operations as the circuit $V_{\partial_{k}^i}$. We can then consider the matrix element prescribed by the block-encoding definition
\begin{align*}
    &\left( I_{k\log(N+1)} \otimes \bra{0}^{\otimes \log(N+1)} \otimes \bra{0}  \otimes \bra{0}^{\otimes \log(k+1)}  \right) V_{\partial_{k}^i} \left( I_{(k+1)\log(N+1)} \otimes \ket{0} \otimes \ket{0}^{\otimes \log(k+1)}  \right) \\
    = & \frac{1}{\sqrt{(N+1)(k+1)}} \sum_{\substack{V_0 < .. V_k \\ s_k \in S_k^i}} \sum_{j=0}^{k} (-1)^j \left( I_{k\log(N+1)} \otimes \bra{0}^{\otimes \log(N+1)} \otimes \bra{0}  \otimes \bra{0}^{\otimes \log(k+1)}  \right)\\
    &\times \left( \ket{V_0} ... \ket{V_{j-1}} \ket{V_{j+1}} ... \ket{V_{k}} \left( \ket{\bar{0}} + \sum_{x=1}^{N} \pm \ket{x} \right) \bra{V_0} ... \bra{V_k} \otimes \ket{0} \otimes \ket{0}^{\otimes \log(k+1)} + A \otimes \ket{1} \otimes \ket{\phi} \right) \\
    =&  \frac{1}{\sqrt{(N+1)(k+1)}} \sum_{\substack{V_0 < ... < V_{k} \\ s_k \in S_k^i}} \sum_{j=0}^{k} (-1)^j  \ket{V_0} ... \ket{V_{j-1}} \ket{V_{j+1}} ... \ket{V_{k}} \bra{V_0} ... \bra{V_j} ... \bra{V_{k}}
\end{align*}
which shows that $V_{\partial_{k}^i}$ is a $\left(\sqrt{(N+1)(k+1)}, (\log(N+1) + \log(k+1) + 1 , \log(k+1) + 1\right), 0)$ block-encoding of $\partial_{k}^i$. Implementing this block-encoding requires a gate-depth of $\mathcal{O}(k\log\log(N+1))$ plus one call to the membership oracle $O_{m_k^i}$. \\

To implement a block-encoding of $\partial_{k+1}^j$
\begin{equation}
    \left( I_{(k+1)\log(N+1)} \otimes \bra{0}^{\otimes \log(N+1)} \otimes \bra{0}  \otimes \bra{0}^{\otimes \log(k+2)}  \right) V_{\partial_{k+1}^j} \left( I_{(k+2)\log(N+1)} \otimes \ket{0} \otimes \ket{0}^{\otimes \log(k+2)}  \right) = \frac{\partial_{k+1}^j}{\alpha'}
\end{equation}
we reuse the analysis above to show that $V_{\partial_{k+1}^j}$ is a $\left( \sqrt{(N+1)(k+2)}, (\log(N+1) + \log(k+2) + 1 , \log(k+2) + 1), 0 \right)$ block-encoding of $\partial_{k+1}^j$, which uses a gate-depth of $\mathcal{O}(k\log\log(N+1))$ and one call to the membership oracle $O_{m_{k+1}^j}$.

\subsubsection{Direct mapping}\label{AppSubSub:BoundaryDirect}
To implement the block-encoding of the boundary operators in the direct mapping, we exploit a previously observed link to second quantized fermionic creation and annihilation operators~\cite{cade2021complexity,akhalwaya2022efficient, ubaru2021quantum, kerenidis2022quantum}. Specifically, the unrestricted boundary operator $\partial$ (which acts on all possible simplices, of all dimensions) can be expressed as
\begin{equation}
    \partial + \partial^\dag = \sum_{i=0}^{N-1} a_i + a_i^\dag
\end{equation}
where $a_i^\dag, a_i$ are second quantized fermionic creation and annihilation operators, respectively. We can take two routes to implementing a block-encoding of this operator. The first, as outlined in Ref.~\cite{kerenidis2022quantum}, uses a unitary circuit with no additional ancilla qubits to implement a scaled version of this operator. This circuit has a depth of $\mathcal{O}(\log(N))$ (composed of $2(N-1)$ single-qubit rotations), and introduces a scaling factor of $\alpha = \sqrt{N}$. Alternatively, one could use techniques for implementing block-encodings of fermionic operators introduced in the literature on quantum chemistry. For example, the approach of Ref.~\cite{wan2021exponentially} would use $\log(N)$ ancilla qubits, and a T gate depth of $24 \lceil \mathrm{log}(N) \rceil$ (with a T count of $24(N-1)$) to implement a block-encoding of $\partial + \partial^\dag$ with a scaling factor of $\alpha=N$. As a result of the reduced ancilla count and improved scaling factor, we expect the first approach to have better performance. \\

We will implement a $(\sqrt{N}, 2, 0)$ block-encoding of $\partial_{k}^i$ 
\begin{equation}
    \left( I_{N} \otimes \bra{00} \right) V_{\partial_{k}^i} \left( I_{N} \otimes \ket{00} \right) = \frac{\partial_{k}^i}{\sqrt{N}}
\end{equation}
with 
\begin{equation}
    \partial_{k}^i = \sum_{s_k \in S_k^i} \sum_{j=0}^k (-1)^j \ket{s_{k}(\hat{v}_j)} \bra{s_k}
\end{equation}
where $s_{k}(\hat{v}_j)$ denotes the Hamming-weight $k+1$ simplex $s_k$ with the $j$th 1 value set to 0. We first observe that
\begin{align*}
    & X_{a_1} O_{m_k^i} \left( I_N \otimes \ket{0}_{a_1} \right) \\
    &= X_{a_1} O_{m_k^i} \left( \sum_{s_k \in S_k^i} \ket{s_k}\bra{s_k} + \sum_{x \notin S_k^i} \ket{x}\bra{x} \right) \otimes \ket{0}_{a_1}  \\
    &= \left( \sum_{s_k \in S_k^i} \ket{s_k}\bra{s_k} \right) \otimes \ket{0}_{a_1} +  \left( \sum_{x \notin S_k^i} \ket{x}\bra{x} \right) \otimes \ket{1}_{a_1}.
\end{align*}
Similarly
\begin{align*}
    & \left( I_N \otimes \bra{0}_{a_2} \right) O_{m_{k-1}^i} X_{a_2} \\
    &= \left( \sum_{s_{k-1} \in S_{k-1}^i} \ket{s_{k-1}}\bra{s_{k-1}} + \sum_{x \notin S_{k-1}^i} \ket{x}\bra{x} \right) \otimes \bra{0}_{a_2} O_{m_{k-1}^i} X_{a_2} \\
    &= \left( \sum_{s_{k-1} \in S_{k-1}^i} \ket{s_{k-1}}\bra{s_{k-1}} \right) \otimes \bra{0}_{a_2} +  \left( \sum_{x \notin S_{k-1}^i} \ket{x}\bra{x} \right) \otimes \bra{1}_{a_2}.
\end{align*}
Sandwiching the circuit from Ref.~\cite{kerenidis2022quantum} between these operations provides the resulting block-encoding of $\partial_k^i$. This approach has a gate depth of $\log(N)$ plus one call each to the membership oracles $O_{m_k^i}$ and $O_{m_{k-1}^i}$. We implement an equivalent block-encoding of $\partial_{k+1}^j$. In Table.~\ref{tab:App:DirectBoundary} we compare the costs of the approach outlined above to the method used to block-encode $\partial_k^i$ in Ref.~\cite{hayakawa2021quantum}.

\begin{table}[!ht]
    \centering
    \begin{tabular}{c|c|c}
            & Here~\cite{kerenidis2022quantum} & Ref~\cite{hayakawa2021quantum}  \\ \hline
         Ancilla qubits & 2 & $\sim \left( \lceil \mathrm{log}(N) \rceil + 2\lceil \mathrm{log}(k) \rceil + 5 \right)$ \\
         $\alpha$ & $\sqrt{N}$ & $Nk$ \\
         Non-Clifford gates & $2(N-1)$ single-qubit rotations & $\mathcal{O}(N^2)$ Toffoli gates  \\
        Gate depth & $\mathcal{O}(\log(N))$ & $\Omega(N)$ \\
        Calls to membership oracle & 2 & 1 \\
    \end{tabular}
    \caption{A comparison of the approach outlined here for implementing a block-encoding of $\partial_k^i$ in the direct mapping against the approach of Ref.~\cite{hayakawa2021quantum}.}
    \label{tab:App:DirectBoundary}
\end{table}

\subsection{Block encodings of subspace projectors}\label{AppSub:SubspaceProjectors}
In this section we discuss how to implement the following block encodings:
\begin{itemize}
    \item $V_{\Pi_{\mathrm{Ker}}}$: A block encoding of $\Pi_{\mathrm{Ker}}$, the projector onto the kernel of $\partial_k^i$.
    \item $V_{\Pi_{\mathrm{Im}}}$: A block encoding of $\Pi_{\mathrm{Im}}$, the projector onto the image of $\partial_{k+1}^j$.
    \item $V_{\Pi_{\mathrm{Ker}(\partial_k^i) \cap \mathrm{Im}(\partial_{k+1}^j)}}$: A block encoding of $\Pi_{\mathrm{Ker} \cap \mathrm{Im} } $, the projector onto $\mathrm{Ker}(\partial_k^i) \cap \mathrm{Im}(\partial_{k+1}^j)$.
    \item $V_{\Pi_\beta}$: A block encoding of $\Pi_{\mathrm{Ker}} - \Pi_{\mathrm{Ker} \cap \mathrm{Im} } $, the projector onto $\mathrm{Ker}(\partial_k^i) - \mathrm{Ker}(\partial_k^i) \cap \mathrm{Im}(\partial_{k+1}^j)$.
\end{itemize}

These will be constructed from QSVT circuits applied to the block encodings of $\partial_k^i$ and $\partial_{k+1}^j$. We use the same approach for both the compact and direct mappings.

\subsubsection{Projector onto kernel}\label{AppSubSub:KernelProjection}
The QSVT circuit for implementing $V_{\Pi_{\mathrm{Ker}}}$ uses singular value threshold projectors~\cite{gilyen2019quantum}, i.e., it (approximately) transforms the block-encoding of $\partial_{k}^i$ into a block-encoding of the projector $\Pi_{\mathrm{Ker}}$ by mapping $0$ singular values to $1$ and non-zero singular values (that are at least $\Lambda_{\partial_k^i}$) to approximately $0$. A technical detail is that we need to use a polynomial of even degree for singular value transformation~\cite{gilyen2019quantum} ensuring that the left and right singular vectors coincide, yielding an (approximate) orthonormal projector. Suitable definite-parity polynomials can be found in Refs.~\cite{gilyen2019quantum, lin2019OptimalQEigenstateFiltering, lin2020near, martyn2021QSVT}. We need to filter out singular values above a threshold $\Lambda_{\partial_k^i}$. Using the results of \cite[Theorem 31]{gilyen2019quantum} given an $(\alpha, m, 0)$ block encoding of an operator $A$, we can implement $V_{\Pi_{\mathrm{Ker}(A)}}$, which is a $(1, m+1, \epsilon_k)$ block encoding of $\Pi_{\mathrm{Ker}(A)}$. We require $\mathcal{O}\left( \frac{\alpha}{\Lambda} \log(\epsilon_k^{-1}) \right)$ calls to an $(\alpha, m, 0)$ block encoding of $A$, and its inverse (where $\Lambda$ is the gap between the zero and lowest non-zero singular values).

As a result, in the compact mapping we can implement a $(1, \log(N+1) + \log(k+1)+2, \epsilon_k)$ block encoding of $\Pi_{\mathrm{Ker}}$ using $\mathcal{O}\left( \frac{\sqrt{(N+1)(k+1)}}{\Lambda_{\partial_k^i}} \log(\epsilon_k^{-1}) \right)$ calls to $V_{\partial_k^i}$, and its inverse. We also require $\mathcal{O}\left( \frac{\sqrt{(N+1)(k+1)}}{\Lambda_{\partial_k^i}} \log(\epsilon_k^{-1}) \right)$ calls to a $(\log(N+1) + \log(k+1)+1)$-controlled-Toffoli gate (and the same number of calls to a $(\log(k+1)+1)$-controlled-Toffoli gate), which can be implemented in $\mathcal{O}(\log(N+1) + \log(k+1)+1)$ depth.

For the direct mapping, we can implement a $(1, 3, \epsilon_k)$ block encoding of $\Pi_{\mathrm{Ker}}$ using $\mathcal{O}\left( \frac{\sqrt{N}}{\Lambda_{\partial_k^i}} \log(\epsilon_k^{-1}) \right)$ calls to $V_{\partial_k^i}$, and its inverse, as well as a similar number of calls to a Toffoli gate. 

Note that the block-encoding of $\partial_k^i$ makes calls to the membership oracle, which dominates the cost of the QSVT circuit, compared to the multicontrol Toffoli gates. As a result, we ignore the costs of the multicontrol Toffoli gates going forwards.

\subsubsection{Projector onto image}\label{AppSubSub:ImageProjection}
To implement a projector onto the image of a matrix $A$, consider its singular value decomposition $A = \sum_i \sigma_i \ket{L_i} \bra{R_i}$. A projector onto the image of $A$ is given by $\Pi_{\mathrm{Im}(A)} = \sum_{j; \sigma_j > 0} \ket{L_j} \bra{L_j}$. We can implement a block encoding of this projector by using QSVT to apply an even polynomial approximating a threshold function to the singular values of the encoded operator $A^\dag$ (which can be implemented using $V_A^\dag$, when $V_A$ is a block encoding of $A$). This ensures that all singular values above the threshold $\Lambda/\alpha$ are set to 1, and all values below are set to zero. The use of an even polynomial ensures that we only keep the left singular vectors of $A$~\cite{gilyen2019quantum}. We can implement $V_{\Pi_{\mathrm{Im}(A)}}$, which is a $(1, m+1, \epsilon_i)$ block encoding of $\Pi_{\mathrm{Im}(A)}$ using $\mathcal{O}\left( \frac{\alpha}{\Lambda} \log(\epsilon_i^{-1}) \right)$ calls to an $(\alpha, m, 0)$ block encoding of $A$, and its inverse, and the same number of calls to $m$-controlled-Toffoli gates.

As a result, in the compact mapping we can implement a $(1, \log(N+1) + \log(k+2)+2, \epsilon_i)$ block encoding of $\Pi_{\mathrm{Im}}$ using $\mathcal{O}\left( \frac{\sqrt{(N+1)(k+2)}}{\Lambda_{\partial_{k+1}^j}} \log(\epsilon_i^{-1}) \right)$ calls to $V_{\partial_{k+1}^j}$, and its inverse, and the same number of calls to $(\log(N+1) + \log(k+2)+1)$- and $( \log(k+2)+1)$-controlled Toffoli gates. 

For the direct mapping, we can implement a $(1, 3, \epsilon_i)$ block encoding of $\Pi_{\mathrm{Im}}$ using $\mathcal{O}\left( \frac{\sqrt{N}}{\Lambda_{\partial_{k+1}^j}} \log(\epsilon_i^{-1}) \right)$ calls to $V_{\partial_{k+1}^j}$, and its inverse, and the same number of calls to additional Toffoli gates.

As described above, we ignore the costs of the multicontrol Toffoli gates going forwards.

\subsubsection{Projector onto kernel and image }\label{AppSubSub:KernelImageProjection}
Given the above methods for implementing $V_{\Pi_{\mathrm{Ker}}}$ and $V_{\Pi_{\mathrm{Im}}}$, we can implement $V_{\Pi_{\mathrm{Ker}(\partial_k^i) \cap \mathrm{Im}(\partial_{k+1}^j)}}$. We cannot simply take a product of the two projectors, as in general they do not commute with each other (because of new simplices that enter the complex at scale $j$). We can implement $V_{\Pi_{\mathrm{Ker}(\partial_k^i) \cap \mathrm{Im}(\partial_{k+1}^j)}}$ by first using a product of block encodings~\cite{gilyen2019quantum} to obtain a product of the projectors $\Pi_{\mathrm{Ker}} \cdot \Pi_{\mathrm{Im}}$. We can use Ref.~\cite[Lemma 53]{gilyen2019quantum}, which states that: \\
\textit{If $U$ is an $(\alpha, a, \delta)$ block encoding of $A$ and $V$ is a $(\beta, b, \epsilon)$ block encoding of $B$, then $(I_b \otimes U)(I_a \otimes V)$ is an $(\alpha \beta, a+b, \alpha\epsilon + \beta\delta)$ block encoding of $AB$.}

As a result, in the compact mapping we can implement a
\begin{equation}
\left( 1, 2\log(N+1) + \log(k+1) + \log(k+2) + 4,  \epsilon_k + \epsilon_i \right)
\end{equation}
block encoding of $\Pi_{\mathrm{Ker}} \Pi_{\mathrm{Im}}$ using one call each to of $V_{\Pi_{\mathrm{Ker}}}$ and $V_{\Pi_{\mathrm{Im}}}$.

In the direct mapping we can implement a
\begin{equation}
\left( 1, 6,  \epsilon_k + \epsilon_i \right)
\end{equation}
block encoding of $\Pi_{\mathrm{Ker}} \Pi_{\mathrm{Im}}$ using one call each to of $V_{\Pi_{\mathrm{Ker}}}$ and $V_{\Pi_{\mathrm{Im}}}$.

We then apply QSVT to this block encoding. We apply an even function that sends all non-unity singular values to zero. In reality, we can only send singular values below a threshold $(1 - \Lambda_{\Pi\Pi})$ to zero. We require $\mathcal{O}\left( \Lambda_{\Pi\Pi}^{-0.5} \log\left(\epsilon_{p}^{-1} \right) \right)$ (see \cite[Lemma 35]{gilyen2018QSingValTransfArXiv} for an explanation of the quadratic improvement) calls to $V_{\Pi_{\mathrm{Ker}}}$ and $V_{\Pi_{\mathrm{Im}}}$ (and their inverses) to implement $V_{\Pi_{\mathrm{Ker}(\partial_k^i) \cap \mathrm{Im}(\partial_{k+1}^j)}}$ in this way. The error in the above block encoding is scaled according to the robustness of QSVT \cite[Lemma 22]{gilyen2019quantum}, by $\mathcal{O}\left( 4\Lambda_{\Pi\Pi}^{-0.5} \log\left(\epsilon_{p}^{-1} \right) \sqrt{\epsilon_k + \epsilon_i} \right)$. The QSVT circuit requires $\mathcal{O}\left( \Lambda_{\Pi\Pi}^{-0.5} \log\left(\epsilon_{p}^{-1} \right) \right)$ calls to either $(2\log(N+1) + \log(k+1) + \log(k+2) + 4)$-controlled-Toffoli gates (compact) or 6-controlled-Toffoli gates (direct). We neglect the cost of these multicontrol Toffoli gates going forwards, as they are dominated by the cost of $V_{\Pi_{\mathrm{Ker}}}$ and $V_{\Pi_{\mathrm{Im}}}$.  

We can prepare the block encoding $V_{\Pi_\beta}$ of $\Pi_{\mathrm{Ker}(\partial_k^i)} - \Pi_{\mathrm{Ker}(\partial_k^i) \cap \mathrm{Im}(\partial_{k+1}^j)}$ using a linear combination of the block encodings $V_{\Pi_{\mathrm{Ker}}}$ and $V_{\Pi_{\mathrm{Ker}(\partial_k^i) \cap \mathrm{Im}(\partial_{k+1}^j)}}$~\cite[Lemma 52]{gilyen2018QSingValTransfArXiv}. This adds a control to each of the operations $V_{\Pi_{\mathrm{Ker}}}$ and $V_{\Pi_{\mathrm{Ker}(\partial_k^i) \cap \mathrm{Im}(\partial_{k+1}^j)}}$, which can be achieved by adding a control qubit to each of the QSVT rotations in their constituent circuits. The overall cost to implement  $V_{\Pi_\beta}$ of $\Pi_{\mathrm{Ker}(\partial_k^i)} - \Pi_{\mathrm{Ker}(\partial_k^i) \cap \mathrm{Im}(\partial_{k+1}^j)}$ is shown in Table~\ref{Apptab:ProjectorEncodingsKerIm}.

\begin{table*}[t]
    \centering
    \begin{tabular}{|c|c|c|} \hline
          & $(\alpha, m, \epsilon)$ & Costs \\ \hline
         Compact & \makecell{$\left(2, 2\log(N+1) + \log(k+1) + \log(k+2) + 6, \right. $ \\ $ \left. \frac{8}{\Lambda_{\Pi\Pi}^{0.5}} \log\left( \epsilon_{p}^{-1} \right) \sqrt{\epsilon_k + \epsilon_i} + \epsilon_{p} + \epsilon_{k} \right)$} &
         \multirow{2}{*}{\makecell{$\mathcal{O}\left( \frac{1}{\Lambda_{\Pi\Pi}^{0.5}} \log\left(\frac{1}{\epsilon_{p}}\right) \right) \times$ \\ $V_{\Pi_{\mathrm{Ker}}}$, $V_{\Pi_{\mathrm{Ker}}}^\dag$, $V_{\Pi_{\mathrm{Im}}}$, $V_{\Pi_{\mathrm{Im}}}^\dag$}}
         \\ \cline{1-2}
        Direct & $\left(2, 8, \frac{8}{\Lambda_{\Pi\Pi}^{0.5}}\log\left( \epsilon_{p}^{-1} \right) \sqrt{\epsilon_{k} + \epsilon_{i}} + \epsilon_{p} + \epsilon_{k} \right) $ &  \\ \hline
    \end{tabular}
    \caption{A summary of the costs to implement the block encoding $V_{\Pi_\beta}$ of $\Pi_{\mathrm{Ker}(\partial_k^i)} - \Pi_{\mathrm{Ker}(\partial_k^i) \cap \mathrm{Im}(\partial_{k+1}^j)}$.}
    \label{Apptab:ProjectorEncodingsKerIm}
\end{table*}

\subsection{State preparation unitary \texorpdfstring{$V_{\psi}$}{V\_psi}}\label{AppSub:StatePrep}
In this section we construct the state preparation unitary $V_\psi$ that approximately prepares the state $\ket{\psi_{S_k^i}}$, the purification of the maximally mixed state over $k$-simplices in the complex at scale $i$. Our approach uses fixed-point amplitude amplification, analyzed from a QSVT perspective~\cite{gilyen2019quantum}.

\subsubsection{Preparing purification of maximally mixed state over all possible \texorpdfstring{$k$}{k}-simplices}\label{AppSubSub:UniformSuperposition}

In this section we consider unitary circuits $U_{\mathrm{uni}}$ that prepare the purification of the maximally mixed state over all $\binom{N}{k+1}$ possible $k$-simplices in the complex. We will discuss how to implement unitaries that first prepare an equal superposition over all possible $k$-simplices in the complex. We can then introduce a second register of equal size, and apply CNOT gates between the corresponding qubits in each register, to generate the purification.\\

The direct encoding stores $k$-simplices as Hamming weight $(k+1)$ computational basis states of $N$ qubits. An equally weighted superposition of states with fixed Hamming weight is known as a Dicke state. We can prepare Dicke states using the approaches in Refs.~\cite{bartschi2019Dicke,bartschi2022DickeImproved}, the most efficient of which requires $\mathcal{O}(k\log(N/k))$ depth. \\

Preparing an equal superposition of $k$-simplices in the compact encoding proceeds as follows. We will consider $N+1 = 2^m$ for some integer $m$ to simplify the analysis. The uniform superposition over $k$-simplices built from $N$ datapoints is
\begin{equation}
    \frac{1}{\sqrt{\binom{N}{k+1}}} \sum_{s_{k}} \ket{s_{k}} = \frac{1}{\sqrt{\binom{N}{k+1}}} \sum_{V_1 = 1}^{[N-k]} \sum_{V_2 > V_1}^{[N-k+1]} ... \sum_{V_{k+1} > V_{k}}^{N} \ket{V_1} \ket{V_2} ... \ket{V_{k+1}}
\end{equation}
where each register acts on $m = \log(N+1)$ qubits. We can prepare this state by first placing each register in an equal superposition of $N$ states
\begin{equation}
    \left( \frac{1}{\sqrt{N}} \sum_{V_i=1}^{N} \ket{V_i} \right)^{\otimes (k+1)}.
\end{equation}
The dimension of the state above is $N^{k+1}$. Preparing the superpositions on each register can be done (for each register in parallel) using $\mathcal{O}(\log(N))$ gates~\cite{babbush2018encoding} (see also Refs.~\cite{gilyen2014mastersthesis,web:GidneyStackExchange,web:PallisterStackExchange}). We need to isolate the components of this state that are permutations of our desired orderings (i.e. we need to remove branches of the superposition with repeated vertices)
\begin{equation}
    \frac{1}{\sqrt{(k+1)! \binom{N}{k+1}}} \sum_{V_1 \neq V_2 ... \neq V_{k+1}}^{N} \ket{V_1} \ket{V_2} ... \ket{V_{k+1}}
\end{equation}
We check all pairs of registers using $k$ parallel rounds of $(k+1)/2$ equality checks. Each in-place equality check uses $\log(N+1)$ CNOT gates followed by a $\log(N+1)$-controlled Toffoli gate targeting an ancilla qubit. The gate depth can be reduced by introducing $\log(N+1)$ ancilla qubits (for each register pairing), and using a tree structure with Toffoli depth $\log\log(N)$ to replace the $\log(N+1)$-controlled Toffoli gate. We then use a $(k+1)/2$-controlled Toffoli to flag any failed equality check in that round. Once again we can replace the $(k+1)/2$-controlled Toffoli with $(k+1)/2$ ancilla qubits and a tree structure with with Toffoli depth $\bigO{\log(k)}$. After all the rounds we require a final $k$-controlled Toffoli to trigger the main flag qubit (which can again be replaced by ancilla qubits and a tree structure). Overall the repeated vertex detection requires $\bigO{k(\log\log(N) + \log(k))}$ gates and $\bigO{k\log(N)}$ ancilla qubits. The dimension of the desired `permutation state' is $N \times ... \times (N-k) = \frac{N!}{(N-k-1)!}$ so the probability of not seeing a duplicate vertex is $\prod_{j=0}^{k}\frac{N-j}{N}=\prod_{j=1}^{k}(1-\frac{j}{N})=\exp\left(\sum_{j=1}^{k}\ln(1-\frac{j}{N})\right)$. 
We can use exact amplitude amplification, because we can classically exactly compute the above probability of not getting repeated vertices. As long as $j\leq k+1 \leq \frac{N}{2}$ we have $\ln(1-\frac{j}{N})\geq -\frac{2j}{N}$ and the above probability can be lower bounded by $\exp\left(-\sum_{j=1}^{k}\frac{2j}{N}\right)=\exp\left(-k(k+1)/N\right)$, upper bounding the required number of amplitude amplification rounds by $\bigO{\exp\left(2\binom{k+1}{2}/N\right)}$.

Once the permutation state is prepared it can be sorted into the correct order using a reversible quantum sorting network, as described in Sec.~6.3.1 of Ref.~\cite{gilyen2014MScThesis}. The sorting network requires $\mathcal{O}\left( k \log(k) \log\log(N)  \right)$ depth and $\mathcal{O}\left( k (\log(k) + \log(N))  \right)=\mathcal{O}\left( k \log(N)  \right)$ additional ancilla qubits. At the end of the sorting network, the ancilla register is unentangled with the desired state, and so can be treated as part of the purifying register. Therefore, assuming that $k\leq \frac{N}{\log(N)}$ (otherwise we would advise using the direct encoding anyway), the overall complexity of preparing the equal superposition over simplices in the compact mapping is 
\begin{equation}
    \mathcal{O}\left( ( \log(N) + k\log(k) + k\log\log(N)) \exp\left(2\binom{k+1}{2}/N\right) + k \log(k) \log\log(N)  \right)
\end{equation}
and requires
\begin{equation}
    \mathcal{O}\left( k \log(N)  \right)
\end{equation}
additional ancilla qubits.

\subsubsection{Preparing the purified maximally mixed state over \texorpdfstring{$k$}{k}-simplices in the complex via fixed-point amplitude amplification}\label{AppSubSub:FixedPointAmpAmp}

Given the above unitary circuits that prepare the purification of the maximally mixed state over all possible $k$-simplices in the complex, we can use fixed point amplitude amplification (implemented using QSVT) to generate $V_{\psi}$. A thorough exposition of this technique is given in Ref.~\cite[Section 3]{martyn2021QSVT}. 

Our presentation makes use of a generalization of block-encodings, known as projected unitary encodings~\cite{gilyen2019quantum}. The necessary details are captured by the following Definitions and Lemma, which are less formal restatements of the results of Ref.~\cite{gilyen2019quantum}.
\begin{definition}
We say an $n+m$ qubit operator $U$ is an $(\alpha, m, \epsilon)$ projected unitary encoding of the $n$ qubit operator $A$ if $||A - \alpha \Pi_L U \Pi_R|| \leq \epsilon$, where $\Pi_L, \Pi_R$ are orthogonal projectors.
\end{definition}
\begin{definition}
We define a C$_\Pi$NOT gate as 
\begin{equation}
    \mathrm{C}_\Pi\mathrm{NOT} = (I-\Pi)\otimes I +  \Pi \otimes X. 
\end{equation}
For example, if $\Pi = \ket{1}\bra{1}$, C$_\Pi$NOT is a regular CNOT gate, while if $\Pi = \ket{11}\bra{11}$, C$_\Pi$NOT is a Toffoli gate.
\end{definition}

Assume we are given an $(\alpha, m, \epsilon)$ projected unitary encoding $U$ of an operator $A$, which has singular value decomposition $A/\alpha = W \Sigma V^\dag $. Assume we are also given the associated C$_{\Pi_L}$NOT and C$_{\Pi_R}$NOT gates, as well as a degree $d$ real polynomial $P(x)$, $|P(x)| \leq 1$ for $x \in [-1,1]$. QSVT provides the following Lemma
\begin{lemma}
We can implement a projected unitary encoding $U_o$ for a degree $d$ odd polynomial, such that $ \Pi_L U_o \Pi_R = P_o(A/\alpha) := W P_o(\Sigma) V^\dag $ (and equivalently $U_e$ for a degree $d$ even polynomial such that $\Pi_R U_e \Pi_R = P_e(A/\alpha) := V P_e(\Sigma) V^\dag $) using a QSVT circuit. The QSVT circuit makes $d$ calls to $U$, $U^\dag$, $2d$ calls to C$_{\Pi_L}$NOT, C$_{\Pi_R}$NOT, and uses $\mathcal{O}(d)$ other single- and two-qubit gates.
\end{lemma}

We define the following operations: $U_{\mathrm{uni}}$ (as above) prepares the purified maximally mixed state over all possible $k$-simplices, $\Pi_0 = \ket{\bar{0}}\bra{\bar{0}}$, $\Pi_{\psi_{k}^i} = \ket{\psi_{S_k^i}}\bra{\psi_{S_k^i}}$. Recall that 
\begin{equation}
    \ket{\psi_{S_k^i}} := \frac{1}{\sqrt{|S_k^i|}} \sum_{s_k \in S_k^i} \ket{s_k} \ket{s_k}.
\end{equation}
By analogy with Grover's algorithm, we note that the membership oracle $O_{m_k^i}$ acts as a $C_{\Pi_{\psi_k^i}}NOT$ gate, with $\Pi_{\psi_{k}^i} = \ket{\psi_{S_k^i}}\bra{\psi_{S_k^i}}$, when we restrict to the 2D subspace spanned by $\ket{\psi_{S_k^i}}, \ket{\psi_{S_k^i}^\perp}$, which we remain in during the course of this subroutine. We then see that
\begin{equation}
    \norm{\Pi_{\psi_k^i} U_{\mathrm{uni}} \Pi_0 - \sqrt{\frac{|S_k^i|}{\binom{N}{k+1}}} \ket{\psi_{S_k^i}}\bra{\bar{0}}} = 0
\end{equation}
i.e. $U_{\mathrm{uni}}$ provides a $(1, 0, 0)$ projected unitary encoding of $\sqrt{\frac{|S_k^i|}{\binom{N}{k+1}}} \ket{\psi_{S_k^i}}\bra{\bar{0}}$. We can use QSVT to map the singular value $\sqrt{\frac{|S_k^i|}{\binom{N}{k+1}}}$ to 1. A layer of the resulting QSVT circuit in shown in Fig.~\ref{fig:App:StatePrepCircuit}. In this circuit, the membership oracle fulfils the role of the C$_{\Pi_L}$NOT gates.

\begin{figure}[!ht]
	\begin{align*}
\Qcircuit @C=0.6em @R=.4em {
\lstick{} & \multigate{1}{U_{uni}} & \qw & \qw & \qw & \multigate{1}{U_{uni}^\dag} & \ctrlo{1} & \qw & \ctrlo{1} & \multigate{1}{U_{uni}} & \qw & \\ 
\lstick{} & \ghost{U_{uni}} &  \multigate{1}{O_{m_k^i}} & \qw & \multigate{1}{O_{m_k^i}} & \ghost{U_{uni}^\dag} & \ctrlo{1} & \qw & \ctrlo{1} & \ghost{U_{uni}} & \qw & \\ 
\lstick{} & \qw & \ghost{O_{m_k^i}} & \gate{R_z} & \ghost{O_{m_k^i}} & \qw & \targ & \gate{R_z} & \targ & \qw & \qw &  \\
}
\end{align*}
\caption{A layer of the QSVT implementation of fixed-point amplitude amplification used to implement the state preparation unitary $V_{\psi}$.}\label{fig:App:StatePrepCircuit}
\end{figure}

The resulting QSVT circuit $V_{\psi}$ is a
\begin{equation}
\left(1, 1, \epsilon_\psi \right)
\end{equation}
projected unitary encoding of $\ket{\psi_{S_k^i}}\bra{\bar{0}}$, and uses $\mathcal{O}\left(\sqrt{\frac{\binom{N}{k+1}}{|S_k^i|}}\log(\epsilon_\psi^{-1})\right)$ calls to $O_{m_k^i}$, $U_{\mathrm{uni}}$, $U_{\mathrm{uni}}^\dag$, and Toffoli gates controlled on all register qubits of the purification ($(2(k+1)\log(N+1))$ qubits for the compact mapping, $2N$ qubits for the direct mapping).


\section{Determining the overall complexity of the algorithm}\label{AppSec:OverallComplexity}

In this section we determine the overall complexity of our algorithm. We invoke Corollary~\ref{Coro:PersistentBetti}, which we restate here:\\
\textit{
To estimate $\beta_k^{i,j}$ to additive error $\Delta$ with success probability $\geq 1- \eta$, we solve two instances of normalized projector rank estimation. The first encodes the value of $|S_k^i|/\binom{N}{k+1}$, and the second encodes the value of $\beta_k^{i,j}/|S_k^i|$. Each instance uses $\mathcal{O}\left( \log(\eta^{-1})\right)$ repetitions of a quantum circuit that makes $\mathcal{O}\left( \delta_x^{-1} \right)$ calls to the state preparation unitaries $V_{\psi_x}$ (acting on $b_x$ qubits) and the $(\alpha_x, a_x, \delta_x/4)$ block-encodings $V_{\Pi_x}$ of projector $\Pi_x$ shown in Table~\ref{tab:CorollaryInstanceData2}. Each circuit also uses $\bigO{(a_x + b_x)\frac{\alpha_x}{\delta_x}}$ additional two-qubit gates. 
}

\begin{table*}[ht]
    \centering
    \begin{tabular}{c|c|c}
        Instance $x$  & 1 & 2 \\ \hline
        Estimates & $X:= \sqrt{|S_k^i|/\binom{N}{k+1}}$ & $Y:=\sqrt{\beta_k^{i,j}/|S_k^i|}$ \\ \hline
        $\delta_x$ & $\frac{\Delta}{4 \beta_k^{i,j}} \sqrt{\frac{|S_k^i|}{\binom{N}{k+1}}}$ & $\frac{\Delta}{4\sqrt{ |S_k^i| \beta_k^{i,j} }}$ \\ \hline
        $\Pi_x$ & \makecell{$\Pi_{S_k^i}$: Projects onto \\ $k$-simplices in complex at scale $i$} & \makecell{$\Pi_{\mathrm{Ker}} - \Pi_{\mathrm{Ker} \cap \mathrm{Im} }$: Projects onto \\ $\mathrm{Ker}(\partial_k^i) - \mathrm{Ker}(\partial_k^i) \cap \mathrm{Im}(\partial_{k+1}^j)$} \\ \hline
        $\ket{\psi_x}$ & \makecell{$\ket{\psi_s}$: Uniform superposition of \\ all possible $k$-simplices} & \makecell{$\ket{\psi_{S_k^i}}$: Purification of maximally mixed state \\ over all $k$-simplices in complex at scale $i$ }
    \end{tabular}
    \caption{Details of the error parameters and block encoded operators for each of the instances of normalized projector rank estimation used to estimate the persistent Betti number to additive error $\Delta$ using the algorithm of Theorem~\ref{Theorem:SubspaceDimensionEst}.}
    \label{tab:CorollaryInstanceData2}
\end{table*}

We will make use of the following facts and assumptions throughout:
\begin{itemize}
    \item Instance 2 is more costly than instance 1, and hence we use the cost of instance 2 to bound the asymptotic cost of the algorithm.
    \item We assume the cost of calls to membership oracles in the algorithm have a cost equal to $O_{m_k^i}$ (for $k-1, k, k+1$ and $i,j$) (this is true for the direct mapping, and asymptotically true for the compact mapping).
    \item We assume the inverse of a block-encoding unitary $V^\dag$ has the same cost as the unitary $V$.
\end{itemize}

The cost of the quantum circuit is then given by:
\begin{align}
    \mathcal{O}&\left( \frac{\alpha_2}{\delta_2} \times \left(V_{\psi} + V_{\Pi_\beta} \right) \right).
\end{align}
In the following subsections, we will determine these costs in terms of the building blocks introduced in the preceding sections. We will do this for both the direct mapping, and the compact mapping with a slow-load lookup table. We will not treat the compact mapping with a fast-load lookup table explicitly, as it is a small modification of the slow-load results, and the large increase in ancilla qubits comes at the expense of a modest improvement in the gate complexity.

\subsection{Direct mapping}\label{AppSubSec:OverallComplexityDirect}
        
We focus first on the direct mapping. We have the following costs:
\begin{itemize}
    \item $O_{m_k^i} : \bigO{N\log(N)}$ depth and $\bigO{N}$ additional qubits.
    \item $V_{\psi}$ : uses $\bigO{\sqrt{\frac{\binom{N}{k+1}}{|S_k^i|}} \log\left( \frac{1}{\epsilon_\psi} \right)} \times [U_{\mathrm{Uni}} + O_{m_k^i}]$.
    \begin{itemize}
        \item $U_{\mathrm{Uni}}: \bigO{k\log(N)}$ depth and no additional qubits.
    \end{itemize}
    \item $V_{\Pi_\beta}$ : $\left(2, 8, \epsilon_p + \epsilon_k + \frac{8}{\Lambda_{\Pi\Pi}^{0.5}}\log\left( \epsilon_{p}^{-1} \right) \sqrt{\epsilon_{k} + \epsilon_{i}} \right)$ block-encoding, using $\bigO{\frac{1}{\Lambda_{\Pi\Pi}^{0.5}} \frac{\sqrt{N}}{\mathrm{Min}\left(\Lambda_{\partial_k^i}, \Lambda_{\partial_{k+1}^j} \right)} \log\left(\frac{1}{\epsilon_{p}}\right) \log\left( \frac{1}{\epsilon_{i/k}} \right)} \times [V_{\partial_k^i}, V_{\partial_{k+1}^j}]$.
    \begin{itemize}
        \item $V_{\partial_k^i}, V_{\partial_{k+1}^j}: \bigO{\log(N)} + O_{m_k^i}$ depth, and only the additional qubits required for $O_{m_k^i}$.
    \end{itemize}
\end{itemize}
We refer to the error $\epsilon_{i/k}$ because as we shall show below, $\epsilon_{i}$ and $\epsilon_{k}$ can be chosen to be the same.\\

Substituting these costs into the overall algorithmic cost and dropping subleading terms ($U_{\mathrm{Uni}}$, $\log(N)$ depth in $V_\partial$) gives:
\begin{align*}
    \bigO{\frac{N\log(N)}{\delta_2}\times \left(\sqrt{\frac{\binom{N}{k+1}}{|S_k^i|}} \log\left( \frac{1}{\epsilon_\psi} \right) + \frac{1}{\Lambda_{\Pi\Pi}^{0.5}} \frac{\sqrt{N}}{\mathrm{Min}\left(\Lambda_{\partial_k^i}, \Lambda_{\partial_{k+1}^j} \right)} \log\left(\frac{1}{\epsilon_{p}}\right) \log\left( \frac{1}{\epsilon_{i/k}} \right) \right)  }
\end{align*}
We must determine suitable values for the errors $\epsilon_\psi, \epsilon_p, \epsilon_{i/k}$. As stated in Corollary~\ref{Coro:PersistentBetti}, the error in $V_{\Pi_\beta}$ must be equal to $\delta_2/4$. Hence
\begin{equation}
    \epsilon_p + \epsilon_k + \frac{8}{\Lambda_{\Pi\Pi}^{0.5}}\log\left( \epsilon_{p}^{-1} \right) \sqrt{\epsilon_{k} + \epsilon_{i}} \sim \delta_2.
\end{equation}
We can thus choose
\begin{align}
    \epsilon_p &\sim \delta_2 \\
    \epsilon_{i/k} &\sim \frac{\delta_2^2 \Lambda_{\Pi\Pi}}{\log^2\left(\frac{1}{\delta_2} \right)}
\end{align}
Similarly $\epsilon_\psi = \delta_2/4$. Substituting these values into the cost of the algorithm, and hiding constant terms using the big-$\mathcal{O}$ notation yields
\begin{equation}
        \mathcal{O}\left( \frac{N\log(N)}{\delta_2} \log\left( \frac{1}{\delta_2} \right) \times \left( \sqrt{\frac{\binom{N}{k+1}}{|S_k^i|}}  + \frac{\sqrt{N}}{ \Lambda_{\Pi \Pi}^{0.5} \mathrm{Min}\left(\Lambda_{\partial_k^i}, \Lambda_{\partial_{k+1}^j} \right)} \log\left(\frac{1}{\delta_2^2 \Lambda_{\Pi\Pi}} \right) \right) \right).
\end{equation}
Substituting in the value of $\delta_2$ yields
\begin{equation}
        \mathcal{O}\left( \frac{N\log(N) \sqrt{\beta_k^{i,j}}}{\Delta} \log\left( \frac{\sqrt{\beta_k^{i,j} |S_k^i|}}{\Delta} \right) \times \left( \sqrt{\binom{N}{k+1}}  + \frac{\sqrt{N|S_k^i|}}{ \Lambda_{\Pi \Pi}^{0.5} \mathrm{Min}\left(\Lambda_{\partial_k^i}, \Lambda_{\partial_{k+1}^j} \right)} \log\left(\frac{\beta_k^{i,j} |S_k^i|}{\Delta^2 \Lambda_{\Pi\Pi}} \right) \right) \right).
\end{equation}
We consider $|S_k^i| \sim \binom{N}{k+1}$, as this is the regime in which classical algorithms are least efficient. The resulting complexity is:
\begin{equation}
        \mathcal{O}\left( \frac{N^{3/2}\log(N) \sqrt{\binom{N}{k+1} \beta_k^{i,j}}}{\Delta   \Lambda_{\Pi \Pi}^{0.5} \mathrm{Min}\left(\Lambda_{\partial_k^i}, \Lambda_{\partial_{k+1}^j} \right)} \log\left(\frac{ \sqrt{\binom{N}{k+1} \beta_k^{i,j}}}{\Delta} \right) \log\left(\frac{ \binom{N}{k+1} \beta_k^{i,j}}{\Delta^2 \Lambda_{\Pi\Pi}} \right) \right)
\end{equation}
depth. We require
\begin{equation}
    \tilde{\mathcal{O}}\left( \log\left(\frac{1}{\eta} \right)  \right)
\end{equation}
incoherent repetitions of this circuit. Hence, the overall complexity is
\begin{equation}
        \tilde{\mathcal{O}}\left( \frac{N^{3/2} \sqrt{\binom{N}{k+1} \beta_k^{i,j} }}{\Delta   \Lambda_{\Pi\Pi}^{0.5} \mathrm{Min}\left(\Lambda_{\partial_k^i}, \Lambda_{\partial_{k+1}^j} \right)}\right).
\end{equation}
The quantum circuit used acts on $\mathcal{O}(N)$ qubits.

\subsection{Compact mapping}\label{AppSubSec:OverallComplexityCompact}
For the compact mapping using a slow-load lookup table, we have the following costs:
\begin{itemize}
    \item $O_{m_k^i} : \bigO{N + k(\log(d)\log(b_d) + b_d)}$ depth and $\bigO{k\log(N) + kdb_d^2}$ additional qubits.
    \item $V_{\psi}$ : uses $\bigO{\sqrt{\frac{\binom{N}{k+1}}{|S_k^i|}} \log\left( \frac{1}{\epsilon_\psi} \right)} \times [U_{\mathrm{Uni}} + O_{m_k^i}]$.
    \begin{itemize}
        \item $U_{\mathrm{Uni}}: \bigO{\log(N)} + k\log(k)\log\log(N)$ depth (when $k \in \mathcal{O}(\sqrt{N})$) and $\bigO{k\log(N)}$ additional qubits.
    \end{itemize}
    \item $V_{\Pi_\beta}$ : $\left(2, \log\left((N+1)^2(k+1)(k+2) \right) + 6, \epsilon_p + \epsilon_k + \frac{8}{\Lambda_{\Pi\Pi}^{0.5}}\log\left( \epsilon_{p}^{-1} \right) \sqrt{\epsilon_{k} + \epsilon_{i}} \right)$ block-encoding, using $\bigO{\frac{1}{\Lambda_{\Pi\Pi}^{0.5}} \frac{\sqrt{Nk}}{\mathrm{Min}\left(\Lambda_{\partial_k^i}, \Lambda_{\partial_{k+1}^j} \right)} \log\left(\frac{1}{\epsilon_{p}}\right) \log\left( \frac{1}{\epsilon_{i/k}} \right)} \times [V_{\partial_k^i}, V_{\partial_{k+1}^j}]$.
    \begin{itemize}
        \item $V_{\partial_k^i}, V_{\partial_{k+1}^j}: \bigO{k\log\log(N)} + O_{m_k^i}$ depth, and $\bigO{\log(N)}$ additional qubits, plus those required for $O_{m_k^i}$.
    \end{itemize}
\end{itemize}
We refer to the error $\epsilon_{i/k}$ because as we shall show below, $\epsilon_{i}$ and $\epsilon_{k}$ can be chosen to be the same.\\

We work in regime with $\bigO{k\log(k)\log\log(N)} \in \bigO{N}$. Substituting these costs into the overall algorithmic cost and dropping subleading terms gives:
\begin{align*}
    \bigO{\frac{N + k(\log(d)\log(b_d) + b_d)}{\delta_2}\times \left(\sqrt{\frac{\binom{N}{k+1}}{|S_k^i|}} \log\left( \frac{1}{\epsilon_\psi} \right) + \frac{1}{\Lambda_{\Pi\Pi}^{0.5}} \frac{\sqrt{Nk}}{\mathrm{Min}\left(\Lambda_{\partial_k^i}, \Lambda_{\partial_{k+1}^j} \right)} \log\left(\frac{1}{\epsilon_{p}}\right) \log\left( \frac{1}{\epsilon_{i/k}} \right) \right)  }
\end{align*}
Repeating the error analysis performed for the direct mapping, 
\begin{align}
    \epsilon_p &\sim \delta_2 \\
    \epsilon_{i/k} &\sim \frac{\delta_2^2 \Lambda_{\Pi\Pi}}{\log^2\left(\frac{1}{\delta_2} \right)} \\
    \epsilon_\psi &\sim \delta_2
\end{align}
Substituting these values into the cost of the algorithm, and hiding constant terms using the big-$\mathcal{O}$ notation yields
\begin{equation}
        \mathcal{O}\left( \frac{N + k(\log(d)\log(b_d) + b_d)}{\delta_2} \log\left( \frac{1}{\delta_2} \right) \times \left( \sqrt{\frac{\binom{N}{k+1}}{|S_k^i|}}  + \frac{\sqrt{Nk}}{ \Lambda_{\Pi \Pi}^{0.5} \mathrm{Min}\left(\Lambda_{\partial_k^i}, \Lambda_{\partial_{k+1}^j} \right)} \log\left(\frac{1}{\delta_2^2 \Lambda_{\Pi\Pi}} \right) \right) \right).
\end{equation}
Substituting in the value of $\delta_2$ yields
\begin{equation}
        \mathcal{O}\left( \frac{(N + k(\log(d)\log(b_d) + b_d)) \sqrt{\beta_k^{i,j} |S_k^i|}}{\Delta} \log\left( \frac{\sqrt{\beta_k^{i,j} |S_k^i|}}{\Delta} \right) \times \left( \sqrt{\frac{\binom{N}{k+1}}{|S_k^i|}}  + \frac{\sqrt{Nk}}{ \Lambda_{\Pi \Pi}^{0.5} \mathrm{Min}\left(\Lambda_{\partial_k^i}, \Lambda_{\partial_{k+1}^j} \right)} \log\left(\frac{\beta_k^{i,j} |S_k^i|}{\Delta^2 \Lambda_{\Pi\Pi}} \right) \right) \right).
\end{equation}
We consider $|S_k^i| \sim \binom{N}{k+1}$, as this is the regime in which classical algorithms are least efficient. Noting that the Johnson-Lidenstrauss lemma ensures we can choose $d \in \bigO{\log(N)/\epsilon_{JL}^2}$, we assume that $\bigO{k(\log(d)\log(b_d) + b_d)} \in \bigO{N}$. The resulting complexity is:
\begin{equation}
        \mathcal{O}\left( \frac{N^{3/2}\sqrt{k} \sqrt{\binom{N}{k+1} \beta_k^{i,j}}}{\Delta   \Lambda_{\Pi \Pi}^{0.5} \mathrm{Min}\left(\Lambda_{\partial_k^i}, \Lambda_{\partial_{k+1}^j} \right)} \log\left(\frac{ \sqrt{\binom{N}{k+1} \beta_k^{i,j}}}{\Delta} \right) \log\left(\frac{ \binom{N}{k+1} \beta_k^{i,j}}{\Delta^2 \Lambda_{\Pi\Pi}} \right) \right)
\end{equation}
depth. We require
\begin{equation}
    \tilde{\mathcal{O}}\left( \log\left(\frac{1}{\eta} \right)  \right)
\end{equation}
incoherent repetitions of this circuit. Hence, the overall complexity (circuit depth) is
\begin{equation}
        \tilde{\mathcal{O}}\left( \frac{N^{3/2} \sqrt{k \beta_k^{i,j} \binom{N}{k+1} }}{\Delta   \Lambda_{\Pi\Pi}^{0.5} \mathrm{Min}\left(\Lambda_{\partial_k^i}, \Lambda_{\partial_{k+1}^j} \right)}\right).
\end{equation}
The quantum circuit used acts on $\bigO{k\log(N)}$ qubits.


\section{Nuances of persistent Betti numbers}\label{AppSub:BettiNuances}

\subsection{Performing a change of basis}\label{AppSubSub:RestrictedChainGroup}
As discussed in the main text, there are a number of nuances associated with instantiating operators that encode persistent Betti numbers. In this work, we compute $\beta_k^{i,j}$ as
\begin{align}
    \beta_k^{i,j} = \mathrm{dim}\left( \mathrm{Ker}(\partial_k^i) \right) - \mathrm{dim}\left( \mathrm{Ker}(\partial_k^i) \cap \mathrm{Im}(\partial_{k+1}^j) \right).
\end{align}
Nevertheless, there are other ways this problem can be formulated. For example, the task of computing persistent Betti numbers has been framed in terms of a persistent combinatorial Laplacian~\cite{wang2020persistent, memoli2020persistent}. This is achieved by first defining a subgroup $C^{i,j}_{k+1}(S^j)$ containing $(k+1)$-chains in the complex at scale $j$ whose images under the boundary operator $\partial^j_{k+1}$ are $k$-chains at scale $i$. Formally,
\begin{equation}
    C^{i,j}_{k+1}(S^j) = \{c \in C_{k+1}(S^j) : \partial^j_{k+1}(c) \in C_k(S^i) \}.
\end{equation}
Intuitively, the chains in $C^{i,j}_{k+1}(S^j)$ are mapped to $k$-cycles in $C_k(S^i)$, which can be divided into the $k$-boundaries already present in $C_k(S^i)$, and a subset of the $k$-holes in $C_k(S^i)$. In other words, the additional $(k+1)$-chains in $C^{i,j}_{k+1}(S^j)$ (beyond those already present in $C_{k+1}(S^i)$) have the action of filling-in $k$-holes in $C_k(S^i)$. Defining $\partial^{i,j}_{k+1}$ as $\partial_{k+1}$ restricted to the chains (not just the simplices) in $C^{i,j}_{k+1}(S^j)$ it follows that~\cite{memoli2020persistent}
\begin{equation}
    \mathrm{Im}(\partial^{i,j}_{k+1}) \cong \mathrm{Ker}(\partial_k^i) \cap \mathrm{Im}(\partial_{k+1}^j).
\end{equation}
Moving from the simplex basis to the chain basis is the main subject of this section.

As discussed in Appendix~\ref{AppSec:MathsBackgroundPedagogy}, $\beta_k^{i,j} = \mathrm{dim}\left(\mathrm{Ker}(\Delta^{i,j}_k)\right)$. The persistent combinatorial Laplacian matrix $\Delta^{i,j}_k$ can be built from the operator $\partial^{i,j}_{k+1}$ (c.f. Eq.~(\ref{Eq:PersistentLaplacian})). Expressing $\partial^{i,j}_{k+1}$ in terms of matrices is more complex than expressing $\partial_k^i$, as discussed in Ref.~\cite{memoli2020persistent}, and as we shall show below. This is because the $(k+1)$-simplices in $S_{k+1}^j$ do not form an appropriate basis for the group $C^{i,j}_{k+1}(S^j)$. Even if a chain $\sigma_1 + \sigma_2 \in C^{i,j}_{k+1}(S^j)$ (such that $\partial_{k+1}^j (\sigma_1 + \sigma_2) \in C_k(S^i)$), it can be the case that $\partial_{k+1}^j(\sigma_1), \partial_{k+1}^j(\sigma_2) \notin C_k(S^i)$. The required change of basis adds complexity to both the classical and quantum algorithms for persistent Betti numbers. We elaborate more on this point below. \\

In this Appendix we provide a worked example for constructing the restricted boundary operator $\partial_{k+1}^{i,j}$, in order to build intuition for the increased complexity. Below, we show a pair of simplicial complexes at different scales $i$ and $j$. Clearly there are zero $1$-holes that persist from scale $i$ to scale $j$, as the hole is filled by the new $2$-simplices at scale $j$.

\begin{center}
\begin{tikzpicture}
[align=center,node distance=3cm]

\node [above] at (1.5, 0.5) {$S^i_0 = \{A, B, C, D\}$ \\
$S^i_1 = \{AB, BC, CD, AD\}$ \\
$S^i_2 = \{ \}$};

\filldraw [black]  (0,0) circle (3pt);
\node [left] at (0,0) {A};

\filldraw [black]  (3,0) circle (3pt); 
\node [right] at (3,0) {B};

\filldraw [black]  (3,-2) circle (3pt); 
\node [right] at (3,-2) {C};

\filldraw [black]  (0,-2) circle (3pt);
\node [left] at (0,-2) {D};

\draw (0,0) --(3,0);
\draw (0,0) --(0,-2);
\draw (3,0) --(3,-2);
\draw (0,-2) --(3,-2);

\node [above] at (11.5, 0.5) {$S^j_0 = \{A, B, C, D\}$ \\
$S^j_1 = \{AB, BC, CD, AD, AC, BD\}$ \\
$S^j_2 = \{ABC, ACD, ABD, BCD \}$};

\filldraw [black]  (10,0) circle (3pt);
\node [left] at (10,0) {A};

\filldraw [black]  (13,0) circle (3pt); 
\node [right] at (13,0) {B};

\filldraw [black]  (13,-2) circle (3pt); 
\node [right] at (13,-2) {C};

\filldraw [black]  (10,-2) circle (3pt);
\node [left] at (10,-2) {D};

\draw [black, fill=gray] (10,0) -- (13,0) -- (13,-2);
\draw [black, fill=gray] (10,0) -- (10,-2) -- (13,-2);

\draw (10,0) --(10,-2);
\draw (13,0) --(13,-2);
\draw (10,0) -- (13,0);
\draw (10,0) -- (13, -2);
\draw (10,-2) -- (13,0);
\draw (10,-2) -- (13,-2);

\end{tikzpicture}
\end{center}

We are interested in the elements of the restricted chain group
\begin{equation}
    C^{i,j}_{k+1}(S^j) = \{c \in C_{k+1}(S^j) : \partial^j_{k+1}(c) \in C_k(S^i) \}.
\end{equation}
In this example, it is straightforward to find the elements of this group manually. First, observe that
\begin{align}
    \partial_2^j [ABC] &= BC - AC + AB \\
    \partial_2^j [ACD] &= CD - AD + AC \nonumber \\
    \partial_2^j [ABD] &= BD - AD + AB \nonumber \\
    \partial_2^j [BCD] &= CD - BD + BC. \nonumber
\end{align}
The $1$-chains $AC$ and $BD$ are present at scale $j$, but are not present at scale $i$. Consequently, we need to take linear combinations of the $2$-simplices in $j$ that will cancel out these errant $1$-chains to form the elements of $C^{i,j}_{2}(S^j)$. These are given by
\begin{align}
    c_1 = ABC + ACD; &\quad \partial_2^j [c_1] = AB + BC + CD - AD \\
    c_2 = ABD + BCD; &\quad \partial_2^j [c_2] = AB + BC + CD - AD, \nonumber
\end{align}
or any linear combination of these chains. We see that while $c_1, c_2 \in C^{i,j}_{2}(S^j)$, their constituent simplices are not, as mentioned above. This is why we must represent $\partial_{2}^{i,j}$ in the basis of elements of $C^{i,j}_{2}(S^j)$, rather than in the basis of simplices. We can see that the dimension of $\partial_{2}^{i,j}$ is two, and its rank is one.
The persistent Betti number is given by
\begin{align}
    \beta_1^{i,j} &= \mathrm{dim}\left(\mathrm{Ker}(\partial_1^i)\right) - \mathrm{dim}\left(\mathrm{Im}(\partial_{2}^{i,j})\right) \\
    &= 1 - 1 \nonumber \\
    &= 0 \nonumber
\end{align}
as expected.\\

We now consider what the matrix representation of $\partial_{2}^{i,j}$ should be. We closely follow the steps outlined in Ref.~\cite{memoli2020persistent}, in particular the proof of Lemma 3.4 in that work.

The dimension 2 boundary matrix at scale $j$, projected onto simplices present at scale $j$ is given by
\begin{equation*}
\partial_2^j P_2^j = 
\left(\begin{array}{rrrr}
{\color{blue}ABC} & {\color{blue}ACD} & {\color{blue}ABD} & {\color{blue}BCD} \\
1   & 0   & 1   & 0   \\
1   & 0   & 0   & 1   \\
0   & 1   & 0   & 1   \\
0   & -1   & -1   & 0   \\
-1   & 1   & 0   & 0   \\
0   & 0   & 1   & -1   \\
\end{array}\right)\!\begin{array}{l}
\\ {\color{blue}AB} \\ {\color{blue}BC} \\ {\color{blue}CD} \\ {\color{blue}AD} \\ {\color{blue}AC} \\ {\color{blue}BD}
\end{array}
\end{equation*}

As an aside, we note that the quantum algorithm in Ref.~\cite{ameneyro2022quantum} defines the restricted boundary operator as $P_1^i \partial_2^j P_2^j$. We have been unable to verify that this quantum algorithm is able to correctly obtain the persistent Betti numbers. The matrix representation of this operator for our example is
\begin{equation*}
P_1^i \partial_2^j P_2^j =
\left(\begin{array}{rrrr}
{\color{blue}ABC} & {\color{blue}ACD} & {\color{blue}ABD} & {\color{blue}BCD} \\
1   & 0   & 1   & 0   \\
1   & 0   & 0   & 1   \\
0   & 1   & 0   & 1   \\
0   & -1   & -1   & 0   \\
\end{array}\right)\!\begin{array}{l}
\\ {\color{blue}AB} \\ {\color{blue}BC} \\ {\color{blue}CD} \\ {\color{blue}AD}
\end{array}
\end{equation*}
This operator has rank 3, and would therefore give $\beta_1^{i,j} = -2$, which is not possible. Another issue with this operator can be seen by considering its action on $ABC$ (or any of the $2$-simplices). The operator maps $ABC$ to $AB + BC$, and we have that $ABC \in C_2(S^j)$, $AB, BC \in C_1(S^i)$, from which we would conclude that $ABC \in C_2^{i,j}(S^j)$. However, this is not the case, as $\partial_2^j(ABC) = BC - AC + AB$, and $AC \notin C_1(S^i)$. As discussed above, the chain group $C_2^{i,j}(S^j)$ is made up of a linear combination of the $2$-simplices in $S^j$, and therefore we need a change of basis from simplices to chains, which is not present in the above expression for $\partial_2^{i,j}$. We provide additional discussion on issues encountered when using this expression for computing persistent Betti numbers at the end of this section.

To obtain the correct expression for the restricted boundary operator, we follow the procedure of Ref.~\cite{memoli2020persistent}, and consider the projection onto the $1$-chains that are not in $S_1^i$
\begin{equation*}
(I - P_1^i) \partial_2^j P_2^j = 
\left(\begin{array}{rrrr}
{\color{blue}ABC} & {\color{blue}ACD} & {\color{blue}ABD} & {\color{blue}BCD} \\
-1   & 1   & 0   & 0   \\
0   & 0   & 1   & -1   \\
\end{array}\right)\!\begin{array}{l}
\\ {\color{blue}AC} \\ {\color{blue}BD}
\end{array}
\end{equation*}
This operator maps $2$-chains in $j$ to $1$-chains in $j$ that are not in $i$. Therefore objects in its kernel are $2$-chains in $j$ that are in $i$. This coincides with the definition of $C_2^{i,j}(S^j)$, and therefore the elements of this group are in the kernel of $(I - P_1^i) \partial_2^j P_2^j := D_2^j$. We can identify these chains by performing column reduction on the matrix $D_2^j$. We need to find a non-singular matrix $Y$ such that $R_2^j = D_2^j Y$ is column-reduced. Such a matrix $Y$ can be found by doing a singular value decomposition of $D_2^j$. In this example, we use the matrix
\begin{equation*}
Y = 
\left(\begin{array}{rrrr}
{\color{blue}\alpha} & {\color{blue}\beta} & {\color{blue}\gamma} & {\color{blue}\delta} \\
1   & 0   & 0   & 0   \\
1   & 1   & 0   & 0   \\
0   & 0   & 1   & 1   \\
0   & 0   & 0   & 1   \\
\end{array}\right)\!\begin{array}{l}
\\ {\color{blue}ABC} \\ {\color{blue}ACD} \\ {\color{blue}ABD} \\ {\color{blue}BCD}
\end{array}
\end{equation*}
such that
\begin{equation*}
R_2^j =
\left(\begin{array}{rrrr}
{\color{blue}\alpha} & {\color{blue}\beta} & {\color{blue}\gamma} & {\color{blue}\delta} \\
0  & 1   & 0   & 0   \\
0   & 0   & 1   & 0   \\
\end{array}\right)\!\begin{array}{l}
\\ {\color{blue}AC} \\ {\color{blue}BD}
\end{array}
\end{equation*}
We can see from $R_2^j$ and $Y$ that the columns indexed by $\alpha, \delta$ are in the kernel of $D_2^j$, and correspond to the chains $\alpha = c_1 = ABC + ACD$, $\delta = c_2 = ABD + BCD$, as found earlier. As a result, $Y$ has acted as a change of basis, from the simplicial basis, to a basis of chains that has, as a subset, a basis of $C_2^{i,j}(S^j)$. 

The matrix representation for $\partial_2^{i,j}$ is then obtained by changing the column space of $\partial_2^j$ using $Y$, and projecting into the columns corresponding to elements of $C_2^{i,j}(S^j)$, and the rows corresponding to $1$-simplices in $S^i$:
\begin{align*}
\partial_2^{i,j} &= \left( (\partial_2^j P_2^j)\cdot Y \right) [AB:AD][\alpha, \delta]\\
&= \left(
\begin{array}{r}
\\ {\color{blue}AB} \\ {\color{blue}BC} \\ {\color{blue}CD} \\ {\color{blue}AD} \\ {\color{blue}AC} \\ {\color{blue}BD}
\end{array}\!\left(\begin{array}{rrrr}
{\color{blue}ABC} & {\color{blue}ACD} & {\color{blue}ABD} & {\color{blue}BCD} \\
1   & 0   & 1   & 0   \\
1   & 0   & 0   & 1   \\
0   & 1   & 0   & 1   \\
0   & -1   & -1   & 0   \\
-1   & 1   & 0   & 0   \\
0   & 0   & 1   & -1   \\
\end{array}\right)
\cdot
\left(\begin{array}{rrrr}
{\color{blue}\alpha} & {\color{blue}\beta} & {\color{blue}\gamma} & {\color{blue}\delta} \\
1   & 0   & 0   & 0   \\
1   & 1   & 0   & 0   \\
0   & 0   & 1   & 1   \\
0   & 0   & 0   & 1   \\
\end{array}\right)\!\begin{array}{l}
\\ {\color{blue}ABC} \\ {\color{blue}ACD} \\ {\color{blue}ABD} \\ {\color{blue}BCD}
\end{array}
\right)[AB:AD][\alpha, \delta] \\
&= \begin{array}{r}
\\ {\color{blue}AB} \\ {\color{blue}BC} \\ {\color{blue}CD} \\ {\color{blue}AD} \\ {\color{blue}AC} \\ {\color{blue}BD}
\end{array}\!\left(\begin{array}{rrrr}
{\color{blue}\alpha} & {\color{blue}\beta} & {\color{blue}\gamma} & {\color{blue}\delta} \\
1   & 0   & 1   & 1   \\
1   & 0   & 0   & 1   \\
1   & 1   & 0   & 1   \\
-1   & -1   & -1   & -1   \\
0   & 1   & 0   & 0   \\
0   & 0   & 1   & 0   \\
\end{array}\right)[AB:AD][\alpha, \delta] \\
&= \begin{array}{r}
\\ {\color{blue}AB} \\ {\color{blue}BC} \\ {\color{blue}CD} \\ {\color{blue}AD}
\end{array}\!\left(\begin{array}{rr}
{\color{blue}\alpha} & {\color{blue}\delta} \\
1     & 1   \\
1     & 1   \\
1     & 1   \\
-1    & -1   \\
\end{array}\right)
\end{align*}
This matrix maps $2$-chains in $C_2(S^j)$ into $1$-chains in $C_1(S^i)$, and has a rank of one. As a result, it correctly gives $\beta_1^{i,j} = 0$ for this filtration.\\

We see that finding the matrix representation of $\partial_{k+1}^{i,j}$ is more complex than finding the matrices $\partial_k^i, \partial_k^j,$ etc. This results from the need to identify the subspace $C_{k+1}^{i,j}(S^j)$, and change to a basis that spans this subspace. In this example, this was achieved using the column reduction of $D_2^j$, and finding the rotation matrix $Y$. A more algorithmic approach was also introduced in Ref.~\cite{memoli2020persistent}, which builds the matrix representation of the persistent combinatorial Laplacian by expressing the matrix representation of $\partial_{k+1}^{i,j} \circ \left( \partial_{k+1}^{i,j} \right)^\dag$ in terms of the Schur complement of the submatrix $\partial_{k+1}^{j} \left( \partial_{k+1}^{j} \right)^\dag[\{I_k^j\}][\{I_k^j\}]$ in the matrix $\partial_{k+1}^{j} \left( \partial_{k+1}^{j} \right)^\dag$, where the $\{I_k^j\}$ is the set of $k$-simplices in $S^j$ that are not in $S^i$. The implementation of this Schur complement requires taking four submatrices of $\partial_{k+1}^{j} \left( \partial_{k+1}^{j} \right)^\dag$, multiplying two of the submatrices by the Moore-Penrose pseudo-inverse of the third, and subtracting this product from the first submatrix. This approach was adapted for use in the quantum algorithm for finding persistent Betti numbers in Ref.~\cite{hayakawa2021quantum}. Implementing the Schur complement requires quantum subroutines for block encoding the matrix $\partial_{k+1}^{j} \left( \partial_{k+1}^{j} \right)^\dag$, performing projections into the subspaces $S^i, S^j$ (and their complements), and being able to take the Moore-Penrose pseudo-inverse of block encoded matrices.\\

We now provide an additional discussion of issues encountered when using $P_{k}^i \partial_{k+1} P_{k+1}^j$ as a representation of the restricted boundary operator $\partial_{k+1}^{i,j}$, as was done in Ref.~\cite{ameneyro2022quantum}. We consider the following pair of complexes:

\begin{center}
\begin{tikzpicture}
[align=center,node distance=3cm]

\node [above] at (1, 2.5) {$S^i_0 = \{A, B, C, D, X\}$ \\
$S^i_1 = \{AB, BC, CD, AD\}$ \\
$S^i_2 = \{ \}$};

\filldraw [black]  (0,0) circle (3pt);
\node [left] at (0,0) {A};

\filldraw [black]  (2,0) circle (3pt); 
\node [right] at (2,0) {B};

\filldraw [black]  (2,-2) circle (3pt); 
\node [right] at (2,-2) {C};

\filldraw [black]  (0,-2) circle (3pt);
\node [left] at (0,-2) {D};

\filldraw [black]  (1,2.2) circle (3pt);
\node [left] at (1, 2.2) {X};

\draw (0,0) --(2,0);
\draw (0,0) --(0,-2);
\draw (2,0) --(2,-2);
\draw (0,-2) --(2,-2);

\node [above] at (8, 2.5) {$S^j_0 = \{A, B, C, D, X\}$ \\
$S^j_1 = \{AB, BC, CD, AD, AX, BX\}$ \\
$S^j_2 = \{ABX \}$};

\filldraw [black]  (7,0) circle (3pt);
\node [left] at (7,0) {A};

\filldraw [black]  (9,0) circle (3pt); 
\node [right] at (9,0) {B};

\filldraw [black]  (9,-2) circle (3pt); 
\node [right] at (9,-2) {C};

\filldraw [black]  (7,-2) circle (3pt);
\node [left] at (7,-2) {D};

\filldraw [black]  (8,2.2) circle (3pt);
\node [left] at (8, 2.2) {X};

\draw [black, fill=gray] (7,0) -- (8,2.2) -- (9,0);

\draw (7,0) --(9,0);
\draw (7,0) --(7,-2);
\draw (9,0) --(9,-2);
\draw (7,-2) --(9,-2);
\draw (7,0) -- (8,2.2);
\draw (8, 2.2) -- (9, 0);

\end{tikzpicture}
\end{center}

The persistent Betti number $\beta_1^{i,j}$ is 1 for this system. However, to compute the matrix representation of $P_1^i \partial_2 P_2^j$ we note that
\begin{equation*}
\partial_2 P_2^j = 
\left(\begin{array}{r}
{\color{blue}ABX} \\
1    \\
-1    \\
1    \\
\end{array}\right)\!\begin{array}{l}
\\ {\color{blue}AB} \\ {\color{blue}AX} \\ {\color{blue}BX}
\end{array}
\end{equation*}
and
\begin{equation*}
P_1^i \partial_2 P_2^j =
\left(\begin{array}{r}
{\color{blue}ABX} \\
1    \\
0    \\
0    \\
\end{array}\right)\!\begin{array}{l}
\\ {\color{blue}AB} \\ {\color{blue}AX} \\ {\color{blue}BX}
\end{array}
\end{equation*}
This matrix has a rank of 1. Hence, if we try to compute the persistent Betti number as
\begin{align}
    \beta_1^{i,j} &= \mathrm{dim}\left(\mathrm{Ker}(\partial_1^i)\right) - \mathrm{dim}\left(\mathrm{Im}(\partial_{2}^{i,j})\right) \\
    &= 1 - 1 \nonumber \\
    &= 0 \nonumber
\end{align}
which is not the expected value.

\subsection{Persistent Betti numbers from the quantum Zeno effect}\label{AppSubSub:Zeno}
In Ref.~\cite{web:LloydTalk} it is suggested that it is possible to compute the persistent Betti numbers using only the original quantum algorithm presented in Ref.~\cite{lloyd2016quantum} for computing Betti numbers. The suggested approach is to initially perform phase estimation with the combinatorial Laplacian at scale $\mu_1$, repeating until an eigenvalue of 0 is found. Phase estimation is then performed again with the combinatorial Laplacian at scale $\mu_1 + \delta \mu$, and the hole is considered to persist for at least $\delta \mu$ if an eigenvalue of 0 is measured. Eventually, at scale $\mu_2$ the hole will close fully, at which point the measured eigenvector no longer has eigenvalue 0. It is suggested that if the steps of phase estimation are performed at small enough increments of $\delta \mu$, the probability of projecting into the hole eigenstate stays high, making use of the quantum Zeno effect.\\

We have been unable to recover the suggested result, based on the following example. Consider the system

\begin{center}
\begin{tikzpicture}
[align=center,node distance=3cm]

\node [above] at (1, 2.5) {$S^i_0 = \{A, B, C, D, X\}$ \\
$S^i_1 = \{AB, BC, CD, AD\}$ \\
$S^i_2 = \{ \}$};

\filldraw [black]  (0,0) circle (3pt);
\node [left] at (0,0) {A};

\filldraw [black]  (2,0) circle (3pt); 
\node [right] at (2,0) {B};

\filldraw [black]  (2,-2) circle (3pt); 
\node [right] at (2,-2) {C};

\filldraw [black]  (0,-2) circle (3pt);
\node [left] at (0,-2) {D};

\filldraw [black]  (1,2.2) circle (3pt);
\node [left] at (1, 2.2) {X};

\draw (0,0) --(2,0);
\draw (0,0) --(0,-2);
\draw (2,0) --(2,-2);
\draw (0,-2) --(2,-2);

\node [above] at (8, 2.5) {$S^j_0 = \{A, B, C, D, X\}$ \\
$S^j_1 = \{AB, BC, CD, AD, AX, BX\}$ \\
$S^j_2 = \{ABX \}$};

\filldraw [black]  (7,0) circle (3pt);
\node [left] at (7,0) {A};

\filldraw [black]  (9,0) circle (3pt); 
\node [right] at (9,0) {B};

\filldraw [black]  (9,-2) circle (3pt); 
\node [right] at (9,-2) {C};

\filldraw [black]  (7,-2) circle (3pt);
\node [left] at (7,-2) {D};

\filldraw [black]  (8,2.2) circle (3pt);
\node [left] at (8, 2.2) {X};

\draw [black, fill=gray] (7,0) -- (8,2.2) -- (9,0);

\draw (7,0) --(9,0);
\draw (7,0) --(7,-2);
\draw (9,0) --(9,-2);
\draw (7,-2) --(9,-2);
\draw (7,0) -- (8,2.2);
\draw (8, 2.2) -- (9, 0);

\end{tikzpicture}
\end{center}

The relevant distances in this diagram are: $d_{AB} = 2, d_{AX} \approx 2.42, d_{AC} = 2\sqrt{2}$. At scale $i$ with $\mu_i = d_{AB}$, the hole $AB + BC + CD - AD$ is created. As the value of $\mu$ used to construct the combinatorial Laplacian is increased, there is no change until $\mu_j = d_{AX}$. At scale $j$ the hole still persists. However, 
objects in the kernel of the combinatorial Laplacian are the harmonic representative of the homology group, and so can have unexpected forms~\cite{lim2020hodge}. In this case, the (unnormalized) zero eigenvector of $\Delta_1^j$ is given by $3(AB +BC +CD -AD) - \partial_2(ABX)$. The squared overlap between the (normalized) new and original holes is not equal to one. As a result, there is a non-zero probability of projecting into a different eigenstate, that is not in the kernel of $\Delta_1^j$. Imagine a complex that consists of many copies of this system, separated by a large distance. As the quantum algorithm for Betti numbers works by sampling eigenstates at random, we are only able to count the fraction of eigenstates of $\Delta_1^i$ with zero eigenvalue. If we then try to project these zero eigenstates onto the eigenstates of $\Delta_1^j$, we would erroneously conclude that some non-zero fraction of the holes had closed -- leading us to believe that we have two types of holes, short-lived holes that close by stage $j$, and longer-lived holes that are still open at stage $j$. It appears that the Zeno-like procedure cannot conclusively determine persistent Betti numbers. The issue appears to be that while the parameter $\mu$ changes smoothly, the Laplacian changes discontinuously whenever a new simplex enters the filtration. As a result, the eigenstates undergo sudden jumps to include new basis states, required by the harmonic form. If these basis states were not present in the original eigenstate, then to ensure that probability is conserved, the state must also gain overlap with other (non-zero eigenvalued) eigenstates that also include the new basis states. \\

A possible modification to the idea outlined above would be to consider adiabatic state preparation between the initial combinatorial Laplacian $\Delta_k^i$ and the final combinatorial Laplacian $\Delta_k^j$. For example, we could start in the ground state of $\Delta_k^i$ (whose normalized form is $0.5(AB + BC + CD -AD)$) and then time evolve the state under a time-dependent Hamiltonian $H(s) = (1-s)\Delta_k^i + s\Delta_k^j$, interpolating slowly between $s=0$ and $s=1$. We would then perform phase estimation with $\Delta_k^j$. If we measured an eigenvalue of $0$, we would conclude that the hole is still open. However, if we measured a non-zero eigenvalue we would conclude that the hole has closed. Unfortunately, we have been unable to obtain the correct result with this approach. The initial Laplacian has three relevant degenerate ground states; $0.5(AB + BC + CD -AD)$, $AX$, and $BX$. Only the first of these is present in the complex at stage $i$, and it is this state that we start in. As soon as $s \neq 0$, the degeneracy between the states is broken, and all three states have non-zero eigenvalues (until $s=1$, when one of the states recovers an eigenvalue of 0). To see this rigorously, we expand the initial state in the non-zero eigenstates of $Q_0 (\partial H / \partial s) Q_0 = Q_0 (H(1) - H(0)) Q_0$, where $Q_0$ is the projector onto the 3 relevant zero-valued eigenvectors of $H(0)$. This determines the amplitudes of the final states, assuming we move along the path infinitely slowly. We observe that the initial state has squared overlap of 0.945 and 0.055 with two of these eigenstates (it is also possible to check the squared overlap with the low-lying eigenstates of $H(s \ll 1)$, and our results are in agreement). As a result, there will be a non-zero probability of finding the state in a non-zero eigenstate at the end of the adiabatic path. As with the example above, this would lead us to conclude that some fraction of the holes present in a complex have closed.\\


\begin{thebibliography}{84}
\providecommand{\natexlab}[1]{#1}
\providecommand{\url}[1]{\texttt{#1}}
\expandafter\ifx\csname urlstyle\endcsname\relax
  \providecommand{\doi}[1]{doi: #1}\else
  \providecommand{\doi}{doi: \begingroup \urlstyle{rm}\Url}\fi

\bibitem[Carlsson(2020)]{carlsson2020topologicalreview}
Gunnar Carlsson.
\newblock Topological methods for data modelling.
\newblock \emph{Nature Reviews Physics}, 2\penalty0 (12):\penalty0 697--708,
  2020.
\newblock \doi{10.1038/s42254-020-00249-3}.

\bibitem[De~Silva and Ghrist(2007)]{de2007coverage}
Vin De~Silva and Robert Ghrist.
\newblock Coverage in sensor networks via persistent homology.
\newblock \emph{Algebraic \& Geometric Topology}, 7\penalty0 (1):\penalty0
  339--358, 2007.
\newblock \doi{10.2140/agt.2007.7.339}.

\bibitem[Feng and Porter(2021)]{feng2021persistent}
Michelle Feng and Mason~A Porter.
\newblock Persistent homology of geospatial data: A case study with voting.
\newblock \emph{SIAM Review}, 63\penalty0 (1):\penalty0 67--99, 2021.
\newblock \doi{10.1137/19m1241519}.

\bibitem[Pranav et~al.(2017)Pranav, Edelsbrunner, Van~de Weygaert, Vegter,
  Kerber, Jones, and Wintraecken]{pranav2017topology}
Pratyush Pranav, Herbert Edelsbrunner, Rien Van~de Weygaert, Gert Vegter,
  Michael Kerber, Bernard~JT Jones, and Mathijs Wintraecken.
\newblock The topology of the cosmic web in terms of persistent betti numbers.
\newblock \emph{Monthly Notices of the Royal Astronomical Society},
  465\penalty0 (4):\penalty0 4281--4310, 2017.
\newblock \doi{10.1093/mnras/stw2862}.

\bibitem[Rieck et~al.(2020)Rieck, Yates, Bock, Borgwardt, Wolf, Turk-Browne,
  and Krishnaswamy]{rieck2020uncovering}
Bastian Rieck, Tristan Yates, Christian Bock, Karsten Borgwardt, Guy Wolf,
  Nicholas Turk-Browne, and Smita Krishnaswamy.
\newblock Uncovering the topology of time-varying fmri data using cubical
  persistence.
\newblock \emph{Advances in neural information processing systems},
  33:\penalty0 6900--6912, 2020.
\newblock \doi{10.5555/3495724.3496303}.
\newblock URL \url{https://dl.acm.org/doi/abs/10.5555/3495724.3496303}.

\bibitem[Perea and Harer(2015)]{perea2015sliding}
Jose~A Perea and John Harer.
\newblock Sliding windows and persistence: An application of topological
  methods to signal analysis.
\newblock \emph{Foundations of Computational Mathematics}, 15\penalty0
  (3):\penalty0 799--838, 2015.
\newblock \doi{10.1007/s10208-014-9206-z}.

\bibitem[Gidea and Katz(2018)]{gidea2018financialcrash}
Marian Gidea and Yuri Katz.
\newblock Topological data analysis of financial time series: Landscapes of
  crashes.
\newblock \emph{Physica A: Statistical Mechanics and its Applications},
  491:\penalty0 820--834, 2018.
\newblock \doi{10.1016/j.physa.2017.09.028}.

\bibitem[Leykam and Angelakis(2023)]{leykam2022TDAphysicsreview}
Daniel Leykam and Dimitris~G. Angelakis.
\newblock Topological data analysis and machine learning.
\newblock \emph{Advances in Physics: X}, 8\penalty0 (1):\penalty0 2202331,
  2023.
\newblock \doi{10.1080/23746149.2023.2202331}.

\bibitem[Sale et~al.(2022)Sale, Giansiracusa, and
  Lucini]{sale2022latticemodelTDA}
Nicholas Sale, Jeffrey Giansiracusa, and Biagio Lucini.
\newblock Quantitative analysis of phase transitions in two-dimensional $xy$
  models using persistent homology.
\newblock \emph{Phys. Rev. E}, 105:\penalty0 024121, Feb 2022.
\newblock \doi{10.1103/PhysRevE.105.024121}.
\newblock URL \url{https://link.aps.org/doi/10.1103/PhysRevE.105.024121}.

\bibitem[Tirelli and Costa(2021)]{tirelli2021hubbardmodelTDA}
Andrea Tirelli and Natanael~C. Costa.
\newblock Learning quantum phase transitions through topological data analysis.
\newblock \emph{Phys. Rev. B}, 104:\penalty0 235146, Dec 2021.
\newblock \doi{10.1103/PhysRevB.104.235146}.
\newblock URL \url{https://link.aps.org/doi/10.1103/PhysRevB.104.235146}.

\bibitem[Hensel et~al.(2021)Hensel, Moor, and Rieck]{hensel2021survey}
Felix Hensel, Michael Moor, and Bastian Rieck.
\newblock A survey of topological machine learning methods.
\newblock \emph{Frontiers in Artificial Intelligence}, 4:\penalty0 681108,
  2021.
\newblock \doi{10.3389/frai.2021.681108}.

\bibitem[Neumann and Breeijen(2019)]{neumann2019limitations}
Niels Neumann and Sterre~den Breeijen.
\newblock Limitations of clustering using quantum persistent homology.
\newblock \emph{arXiv preprint arXiv:1911.10781}, 2019.

\bibitem[Lloyd et~al.(2016)Lloyd, Garnerone, and Zanardi]{lloyd2016quantum}
Seth Lloyd, Silvano Garnerone, and Paolo Zanardi.
\newblock Quantum algorithms for topological and geometric analysis of data.
\newblock \emph{Nature Communications}, 7\penalty0 (1):\penalty0 1--7, 2016.
\newblock \doi{10.1038/ncomms10138}.

\bibitem[Lloyd(2021)]{web:LloydTalk}
Seth Lloyd.
\newblock Quantum algorithms for topological and geometric analysis of data,
  2021.
\newblock URL \url{https://youtu.be/G4t7Pdn9R6c?t=3853}.
\newblock Timecode: 1:04:13.

\bibitem[Gunn and Kornerup(2019)]{gunn2019review}
Sam Gunn and Niels Kornerup.
\newblock Review of a quantum algorithm for betti numbers.
\newblock \emph{arXiv preprint arXiv:1906.07673}, 2019.

\bibitem[Ubaru et~al.(2021)Ubaru, Akhalwaya, Squillante, Clarkson, and
  Horesh]{ubaru2021quantum}
Shashanka Ubaru, Ismail~Yunus Akhalwaya, Mark~S Squillante, Kenneth~L Clarkson,
  and Lior Horesh.
\newblock Quantum topological data analysis with linear depth and exponential
  speedup.
\newblock \emph{arXiv preprint arXiv:2108.02811}, 2021.

\bibitem[Hayakawa(2022)]{hayakawa2021quantum}
Ryu Hayakawa.
\newblock Quantum algorithm for persistent {B}etti numbers and topological data
  analysis.
\newblock \emph{{Quantum}}, 6:\penalty0 873, December 2022.
\newblock ISSN 2521-327X.
\newblock \doi{10.22331/q-2022-12-07-873}.
\newblock URL \url{https://doi.org/10.22331/q-2022-12-07-873}.

\bibitem[Ameneyro et~al.(2024)Ameneyro, Maroulas, and
  Siopsis]{ameneyro2022quantum}
Bernardo Ameneyro, Vasileios Maroulas, and George Siopsis.
\newblock Quantum persistent homology.
\newblock \emph{Journal of Applied and Computational Topology}, 8:\penalty0
  1929--1963, 2024.
\newblock \doi{10.1007/s41468-023-00160-7}.

\bibitem[Berry et~al.(2024)Berry, Su, Gyurik, King, Basso, Barba, Rajput,
  Wiebe, Dunjko, and Babbush]{berry2022quantifying}
Dominic~W Berry, Yuan Su, Casper Gyurik, Robbie King, Joao Basso, Alexander
  Del~Toro Barba, Abhishek Rajput, Nathan Wiebe, Vedran Dunjko, and Ryan
  Babbush.
\newblock Analyzing prospects for quantum advantage in topological data
  analysis.
\newblock \emph{PRX Quantum}, 5\penalty0 (1):\penalty0 010319, 2024.
\newblock \doi{10.1103/prxquantum.5.010319}.

\bibitem[{Edelsbrunner} et~al.(2002){Edelsbrunner}, {Letscher}, and
  {Zomorodian}]{Edelsbrunner2002}
{Edelsbrunner}, {Letscher}, and {Zomorodian}.
\newblock Topological persistence and simplification.
\newblock \emph{Discrete {\&} Computational Geometry}, 28\penalty0
  (4):\penalty0 511--533, Nov 2002.
\newblock ISSN 1432-0444.
\newblock \doi{10.1007/s00454-002-2885-2}.
\newblock URL \url{https://doi.org/10.1007/s00454-002-2885-2}.

\bibitem[Milosavljevi{\'c} et~al.(2011)Milosavljevi{\'c}, Morozov, and
  Skraba]{milosavljevic2011zigzag}
Nikola Milosavljevi{\'c}, Dmitriy Morozov, and Primoz Skraba.
\newblock Zigzag persistent homology in matrix multiplication time.
\newblock In \emph{Proceedings of the twenty-seventh Annual Symposium on
  Computational Geometry}, pages 216--225, 2011.
\newblock \doi{10.1145/1998196.1998229}.

\bibitem[Milosavljevic et~al.(2010)Milosavljevic, Morozov, and
  Skraba]{milosavljevic2010:inriaReport}
Nikola Milosavljevic, Dmitriy Morozov, and Primoz Skraba.
\newblock {Zigzag Persistent Homology in Matrix Multiplication Time}.
\newblock Research Report RR-7393, {INRIA}, September 2010.
\newblock URL \url{https://hal.inria.fr/inria-00520171}.

\bibitem[Mischaikow and Nanda(2013)]{mischaikow2013morse}
Konstantin Mischaikow and Vidit Nanda.
\newblock Morse theory for filtrations and efficient computation of persistent
  homology.
\newblock \emph{Discrete \& Computational Geometry}, 50\penalty0 (2):\penalty0
  330--353, 2013.
\newblock \doi{10.1007/s00454-013-9529-6}.

\bibitem[Friedman(1998)]{friedman1998computing}
Joel Friedman.
\newblock Computing betti numbers via combinatorial laplacians.
\newblock \emph{Algorithmica}, 21\penalty0 (4):\penalty0 331--346, 1998.
\newblock \doi{10.1007/pl00009218}.

\bibitem[Huang et~al.(2018)Huang, Wang, Rohde, Luo, Zhao, Liu, Li, Liu, Lu, and
  Pan]{huang2018demonstration}
He-Liang Huang, Xi-Lin Wang, Peter~P Rohde, Yi-Han Luo, You-Wei Zhao, Chang
  Liu, Li~Li, Nai-Le Liu, Chao-Yang Lu, and Jian-Wei Pan.
\newblock Demonstration of topological data analysis on a quantum processor.
\newblock \emph{Optica}, 5\penalty0 (2):\penalty0 193--198, 2018.
\newblock \doi{10.1364/optica.5.000193}.

\bibitem[Akhalwaya et~al.(2022{\natexlab{a}})Akhalwaya, Ubaru, Clarkson,
  Squillante, Jejjala, He, Naidoo, Kalantzis, and
  Horesh]{akhalwaya2022TowardsNisqTDA}
Ismail~Yunus Akhalwaya, Shashanka Ubaru, Kenneth~L. Clarkson, Mark~S.
  Squillante, Vishnu Jejjala, Yang-Hui He, Kugendran Naidoo, Vasileios
  Kalantzis, and Lior Horesh.
\newblock Towards quantum advantage on noisy quantum computers.
\newblock 2022{\natexlab{a}}.
\newblock URL \url{https://arxiv.org/abs/2209.09371}.

\bibitem[Gyurik et~al.(2022)Gyurik, Cade, and Dunjko]{gyurik2020towards}
Casper Gyurik, Chris Cade, and Vedran Dunjko.
\newblock Towards quantum advantage via topological data analysis.
\newblock \emph{Quantum}, 6:\penalty0 855, 2022.
\newblock \doi{10.22331/q-2022-11-10-855}.

\bibitem[Cade and Crichigno(2024)]{cade2021complexity}
Chris Cade and P~Marcos Crichigno.
\newblock Complexity of supersymmetric systems and the cohomology problem.
\newblock \emph{Quantum}, 8:\penalty0 1325, 2024.
\newblock \doi{10.22331/q-2024-04-30-1325}.

\bibitem[Aharonov et~al.(2009)Aharonov, Jones, and
  Landau]{aharonov2006polynomial}
Dorit Aharonov, Vaughan Jones, and Zeph Landau.
\newblock A polynomial quantum algorithm for approximating the {J}ones
  polynomial.
\newblock \emph{Algorithmica}, 55\penalty0 (3):\penalty0 395--421, 2009.
\newblock \doi{10.1007/s00453-008-9168-0}.

\bibitem[Shor and Jordan(2008)]{shor2007estimating}
Peter~W Shor and Stephen~P Jordan.
\newblock Estimating {J}ones polynomials is a complete problem for one clean
  qubit.
\newblock \emph{Quantum Information and Computation}, 8\penalty0
  (8--9):\penalty0 681--714, 2008.
\newblock \doi{10.26421/qic8.8-9-1}.

\bibitem[Bordewich et~al.(2005)Bordewich, Freedman, Lov{\'a}sz, and
  Welsh]{bordewich2009approximate}
Magnus Bordewich, Michael Freedman, L{\'a}szl{\'o} Lov{\'a}sz, and D~Welsh.
\newblock Approximate counting and quantum computation.
\newblock \emph{Combinatorics, Probability and Computing}, 14\penalty0
  (5--6):\penalty0 737--754, 2005.
\newblock \doi{10.1017/s0963548305007005}.

\bibitem[Gily{\'e}n et~al.(2019)Gily{\'e}n, Su, Low, and
  Wiebe]{gilyen2019quantum}
Andr{\'a}s Gily{\'e}n, Yuan Su, Guang~Hao Low, and Nathan Wiebe.
\newblock Quantum singular value transformation and beyond: exponential
  improvements for quantum matrix arithmetics.
\newblock In \emph{Proceedings of the 51st Annual ACM SIGACT Symposium on
  Theory of Computing}, pages 193--204, 2019.
\newblock \doi{10.1145/3313276.3316366}.

\bibitem[Crichigno and Kohler(2024)]{crichigno2022clique}
Marcos Crichigno and Tamara Kohler.
\newblock Clique homology is {QMA}$_1$-hard.
\newblock \emph{Nature Communications}, 15:\penalty0 9846, 2024.
\newblock \doi{10.1038/s41467-024-54118-z}.

\bibitem[Schmidhuber and Lloyd(2023)]{schmidhuber2022complexity}
Alexander Schmidhuber and Seth Lloyd.
\newblock Complexity-theoretic limitations on quantum algorithms for
  topological data analysis.
\newblock \emph{PRX Quantum}, 4\penalty0 (4):\penalty0 040349, 2023.
\newblock \doi{10.1103/prxquantum.4.040349}.

\bibitem[King and Kohler(2026)]{king2023promise}
Robbie King and Tamara Kohler.
\newblock Gapped clique homology on weighted graphs is {QMA}$_1$-hard and
  contained in {QMA}.
\newblock \emph{SIAM Journal on Computing}, 55\penalty0 (1):\penalty0
  S198--S231, 2026.
\newblock \doi{10.1137/24m1710243}.

\bibitem[Apers et~al.(2023)Apers, Sen, and Szab{\'o}]{apers2022simple}
Simon Apers, Sayantan Sen, and D{\'a}niel Szab{\'o}.
\newblock A (simple) classical algorithm for estimating betti numbers.
\newblock \emph{Quantum}, 7:\penalty0 1202, 2023.
\newblock \doi{10.22331/q-2023-12-06-1202}.

\bibitem[Otter et~al.(2017)Otter, Porter, Tillmann, Grindrod, and
  Harrington]{otter2017roadmap}
Nina Otter, Mason~A Porter, Ulrike Tillmann, Peter Grindrod, and Heather~A
  Harrington.
\newblock A roadmap for the computation of persistent homology.
\newblock \emph{EPJ Data Science}, 6:\penalty0 1--38, 2017.
\newblock \doi{10.1140/epjds/s13688-017-0109-5}.

\bibitem[Hatcher(2005)]{hatcher2005algebraic}
Allen Hatcher.
\newblock \emph{Algebraic topology}.
\newblock Cambridge University Press, 2005.

\bibitem[Gilyén et~al.(2018)Gilyén, Su, Low, and
  Wiebe]{gilyen2018QSingValTransfArXiv}
András Gilyén, Yuan Su, Guang~Hao Low, and Nathan Wiebe.
\newblock Quantum singular value transformation and beyond: {E}xponential
  improvements for quantum matrix arithmetics [full version], 2018.
\newblock arXiv: \href{https://arxiv.org/abs/1806.01838}{\ttfamily{1806.01838}}.

\bibitem[Metwalli et~al.(2020)Metwalli, Le~Gall, and
  Van~Meter]{metwalli2020cliquefinding}
Sara~Ayman Metwalli, François Le~Gall, and Rodney Van~Meter.
\newblock Finding small and large $k$-clique instances on a quantum computer.
\newblock \emph{IEEE Transactions on Quantum Engineering}, 1:\penalty0 1--11,
  2020.
\newblock \doi{10.1109/TQE.2020.3045692}.

\bibitem[Hann et~al.(2021)Hann, Lee, Girvin, and Jiang]{hann2021qram}
Connor~T. Hann, Gideon Lee, S.M. Girvin, and Liang Jiang.
\newblock Resilience of quantum random access memory to generic noise.
\newblock \emph{PRX Quantum}, 2:\penalty0 020311, Apr 2021.
\newblock \doi{10.1103/PRXQuantum.2.020311}.
\newblock URL \url{https://link.aps.org/doi/10.1103/PRXQuantum.2.020311}.

\bibitem[Akhalwaya et~al.(2022{\natexlab{b}})Akhalwaya, He, Horesh, Jejjala,
  Kirby, Naidoo, and Ubaru]{akhalwaya2022efficient}
Ismail~Yunus Akhalwaya, Yang-Hui He, Lior Horesh, Vishnu Jejjala, William
  Kirby, Kugendran Naidoo, and Shashanka Ubaru.
\newblock Representation of the fermionic boundary operator.
\newblock \emph{Physical Review A}, 106\penalty0 (2):\penalty0 022407,
  2022{\natexlab{b}}.
\newblock \doi{10.1103/physreva.106.022407}.

\bibitem[Kerenidis and Prakash(2022)]{kerenidis2022quantum}
Iordanis Kerenidis and Anupam Prakash.
\newblock Quantum machine learning with subspace states.
\newblock \emph{arXiv preprint arXiv:2202.00054}, 2022.

\bibitem[Wan(2021)]{wan2021exponentially}
Kianna Wan.
\newblock Exponentially faster implementations of select (h) for fermionic
  hamiltonians.
\newblock \emph{Quantum}, 5:\penalty0 380, 2021.
\newblock \doi{10.22331/q-2021-01-12-380}.

\bibitem[Lin and Tong(2020{\natexlab{a}})]{lin2020near}
Lin Lin and Yu~Tong.
\newblock Near-optimal ground state preparation.
\newblock \emph{Quantum}, 4:\penalty0 372, 2020{\natexlab{a}}.
\newblock \doi{10.22331/q-2020-12-14-372}.

\bibitem[B\"artschi and Eidenbenz(2019)]{bartschi2019Dicke}
Andreas B\"artschi and Stephan Eidenbenz.
\newblock Deterministic preparation of dicke states.
\newblock In \emph{Fundamentals of Computation Theory}, pages 126--139.
  Springer International Publishing, 2019.
\newblock \doi{10.1007/978-3-030-25027-0_9}.
\newblock URL \url{https://doi.org/10.1007\%2F978-3-030-25027-0_9}.

\bibitem[B\"artschi and Eidenbenz(2022)]{bartschi2022DickeImproved}
Andreas B\"artschi and Stephan Eidenbenz.
\newblock Short-depth circuits for dicke state preparation.
\newblock In \emph{2022 IEEE International Conference on Quantum Computing and
  Engineering (QCE)}, pages 87--96. IEEE, 2022.
\newblock \doi{10.1109/qce53715.2022.00027}.

\bibitem[Gilyén(2014)]{gilyen2014MScThesis}
András Gilyén.
\newblock Quantum walk based search methods and algorithmic applications.
\newblock Master's thesis, Eötvös Loránd University, 2014.
\newblock URL
  \url{http://web.cs.elte.hu/blobs/diplomamunkak/msc_mat/2014/gilyen_andras_pal.pdf}.

\bibitem[Morozov(2005)]{morozov2005persistence}
Dmitriy Morozov.
\newblock Persistence algorithm takes cubic time in worst case.
\newblock \emph{BioGeometry News, Dept. Comput. Sci., Duke Univ}, 2, 2005.

\bibitem[Chen and Kerber(2011)]{chen2011persistent}
Chao Chen and Michael Kerber.
\newblock Persistent homology computation with a twist.
\newblock In \emph{Proceedings 27th European workshop on computational
  geometry}, volume~11, pages 197--200, 2011.

\bibitem[De~Silva et~al.(2011)De~Silva, Morozov, and
  Vejdemo-Johansson]{de2011dualities}
Vin De~Silva, Dmitriy Morozov, and Mikael Vejdemo-Johansson.
\newblock Dualities in persistent (co) homology.
\newblock \emph{Inverse Problems}, 27\penalty0 (12):\penalty0 124003, 2011.
\newblock \doi{10.1088/0266-5611/27/12/124003}.

\bibitem[Joswig and Pfetsch(2006)]{joswig2006computing}
Michael Joswig and Marc~E Pfetsch.
\newblock Computing optimal morse matchings.
\newblock \emph{SIAM Journal on Discrete Mathematics}, 20\penalty0
  (1):\penalty0 11--25, 2006.
\newblock \doi{10.1137/s0895480104445885}.

\bibitem[Mrozek et~al.(2008)Mrozek, Pilarczyk, and
  {\.Z}elazna]{mrozek2008homology}
Marian Mrozek, Pawe{\l} Pilarczyk, and Natalia {\.Z}elazna.
\newblock Homology algorithm based on acyclic subspace.
\newblock \emph{Computers \& Mathematics with Applications}, 55\penalty0
  (11):\penalty0 2395--2412, 2008.
\newblock \doi{10.1016/j.camwa.2007.08.044}.

\bibitem[Zomorodian(2010)]{zomorodian2010tidy}
Afra Zomorodian.
\newblock The tidy set: a minimal simplicial set for computing homology of
  clique complexes.
\newblock In \emph{Proceedings of the twenty-sixth annual symposium on
  Computational geometry}, pages 257--266, 2010.
\newblock \doi{10.1145/1810959.1811004}.

\bibitem[Barmak and Minian(2012)]{barmak2012strong}
Jonathan~Ariel Barmak and Elias~Gabriel Minian.
\newblock Strong homotopy types, nerves and collapses.
\newblock \emph{Discrete \& Computational Geometry}, 47\penalty0 (2):\penalty0
  301--328, 2012.
\newblock \doi{10.1007/s00454-011-9357-5}.

\bibitem[D{\l}otko and Wagner(2014)]{dlotko2014simplification}
Pawe{\l} D{\l}otko and Hubert Wagner.
\newblock Simplification of complexes for persistent homology computations.
\newblock \emph{Homology, Homotopy and Applications}, 16\penalty0 (1):\penalty0
  49--63, 2014.
\newblock \doi{10.4310/hha.2014.v16.n1.a3}.

\bibitem[Boissonnat et~al.(2023)Boissonnat, Pritam, and
  Pareek]{boissonnat2018strong}
Jean-Daniel Boissonnat, Siddharth Pritam, and Divyansh Pareek.
\newblock Strong collapse and persistent homology.
\newblock \emph{Journal of Topology and Analysis}, 15\penalty0 (1):\penalty0
  185--213, 2023.
\newblock \doi{10.1142/s1793525321500291}.

\bibitem[Wang et~al.(2020)Wang, Nguyen, and Wei]{wang2020persistent}
Rui Wang, Duc~Duy Nguyen, and Guo-Wei Wei.
\newblock Persistent spectral graph.
\newblock \emph{International journal for numerical methods in biomedical
  engineering}, 36\penalty0 (9):\penalty0 e3376, 2020.
\newblock \doi{10.1002/cnm.3376}.

\bibitem[M{\'e}moli et~al.(2022)M{\'e}moli, Wan, and
  Wang]{memoli2020persistent}
Facundo M{\'e}moli, Zhengchao Wan, and Yusu Wang.
\newblock Persistent laplacians: Properties, algorithms and implications.
\newblock \emph{SIAM Journal on Mathematics of Data Science}, 4\penalty0
  (2):\penalty0 858--884, 2022.
\newblock \doi{10.1137/21m1435471}.

\bibitem[Wang et~al.(2021)Wang, Zhao, Ribando-Gros, Chen, Tong, and
  Wei]{wang2021hermes}
Rui Wang, Rundong Zhao, Emily Ribando-Gros, Jiahui Chen, Yiying Tong, and
  Guo-Wei Wei.
\newblock Hermes: Persistent spectral graph software.
\newblock \emph{Foundations of data science (Springfield, Mo.)}, 3\penalty0
  (1):\penalty0 67, 2021.
\newblock \doi{10.3934/fods.2021006}.

\bibitem[Aurentz et~al.(2019)Aurentz, Austin, Benzi, and
  Kalantzis]{aurentz2019StableMatrixPolynomial}
Jared~L. Aurentz, Anthony~P. Austin, Michele Benzi, and Vassilis Kalantzis.
\newblock Stable computation of generalized matrix functions via polynomial
  interpolation.
\newblock \emph{SIAM Journal on Matrix Analysis and Applications}, 40\penalty0
  (1):\penalty0 210--234, 2019.
\newblock \doi{10.1137/18M1191786}.
\newblock URL \url{https://doi.org/10.1137/18M1191786}.

\bibitem[Goff(2011)]{goff2011extremal}
Michael Goff.
\newblock Extremal {B}etti numbers of {V}ietoris-{R}ips complexes.
\newblock \emph{Discrete \& Computational Geometry}, 46\penalty0 (1):\penalty0
  132--155, 2011.
\newblock \doi{10.1007/s00454-010-9274-z}.

\bibitem[Kahle(2011)]{Kahle2011RandomComplexes}
Matthew Kahle.
\newblock Random geometric complexes.
\newblock \emph{Discrete {\&} Computational Geometry}, 45\penalty0
  (3):\penalty0 553--573, jan 2011.
\newblock \doi{10.1007/s00454-010-9319-3}.
\newblock URL \url{https://doi.org/10.1007\%2Fs00454-010-9319-3}.

\bibitem[Bobrowski and Kahle(2018)]{bobrowski2014RandomComplexesSurvey}
Omer Bobrowski and Matthew Kahle.
\newblock Topology of random geometric complexes: a survey.
\newblock \emph{Journal of Applied and Computational Topology}, 1\penalty0
  (3--4):\penalty0 331--364, 2018.
\newblock \doi{10.1007/s41468-017-0010-0}.

\bibitem[Carlsson and Vejdemo-Johansson(2021)]{carlsson_vejdemo-johansson_2021}
Gunnar Carlsson and Mikael Vejdemo-Johansson.
\newblock \emph{Topological Data Analysis with Applications}.
\newblock Cambridge University Press, 2021.
\newblock \doi{10.1017/9781108975704}.

\bibitem[Topaz(2016)]{web:TopazNotes}
Chad Topaz.
\newblock Chad's self-help homology tutorial for the simple(x) minded, 2016.
\newblock URL
  \url{https://drive.google.com/file/d/0B3Www1z6Tm8xV3ozTmN5RE94bDg/view?resourcekey=0-tE7y-zXFtV3OWSGmjUebYA}.
\newblock Last accessed 7 April 2022.

\bibitem[Feng et~al.(2021)Feng, Hickok, Kureh, Porter, and
  Topaz]{feng2021Connecting}
Michelle Feng, Abigail Hickok, Yacoub Kureh, Mason Porter, and Chad Topaz.
\newblock Connecting the dots: Discovering the `shape' of data.
\newblock \emph{Front. Young Minds}, 9\penalty0 (551557), 2021.
\newblock \doi{10.3389/frym.2021.551557}.
\newblock URL
  \url{https://kids.frontiersin.org/articles/10.3389/frym.2021.551557}.

\bibitem[Rieck(September 2020)]{web:RieckLecture}
Bastian Rieck.
\newblock Topological data analysis for machine learning i: Algebraic topology,
  September 2020.
\newblock URL \url{https://www.youtube.com/watch?v=gVq_xXnwV-4}.
\newblock Last accessed 7 April 2022.

\bibitem[Feng(February 2021)]{web:FengLecture}
Michelle Feng.
\newblock Michelle feng: Topological techniques, February 2021.
\newblock URL \url{https://www.youtube.com/watch?v=M3TU4NmHDkM}.
\newblock Last accessed 7 April 2022.

\bibitem[Lim(2020)]{lim2020hodge}
Lek-Heng Lim.
\newblock Hodge laplacians on graphs.
\newblock \emph{SIAM Review}, 62\penalty0 (3):\penalty0 685--715, 2020.
\newblock \doi{10.1137/18m1223101}.

\bibitem[Goldberg(2002)]{goldberg2002combinatorial}
Timothy~E Goldberg.
\newblock Combinatorial laplacians of simplicial complexes.
\newblock \emph{Senior Thesis, Bard College}, 2002.

\bibitem[Low and Chuang(2017)]{low2017HamSimUnifAmp}
Guang~Hao Low and Isaac~L. Chuang.
\newblock Hamiltonian simulation by uniform spectral amplification.
\newblock arXiv: \href{https://arxiv.org/abs/1707.05391}{\ttfamily{1707.05391}}, 2017.

\bibitem[Brassard et~al.(2002)Brassard, H{\o}yer, Mosca, and
  Tapp]{brassard2002AmpAndEst}
Gilles Brassard, Peter H{\o}yer, Michele Mosca, and Alain Tapp.
\newblock Quantum amplitude amplification and estimation.
\newblock In \emph{Quantum Computation and Quantum Information: A Millennium
  Volume}, volume 305 of \emph{Contemporary Mathematics Series}, pages 53--74.
  AMS, 2002.
\newblock \doi{10.1090/conm/305/05215}.
\newblock arXiv: \href{https://arxiv.org/abs/quant-ph/0005055}{\ttfamily{quant-ph/0005055}}.

\bibitem[Rall and Fuller(2023)]{Rall2023amplitudeestimation}
Patrick Rall and Bryce Fuller.
\newblock Amplitude {E}stimation from {Q}uantum {S}ignal {P}rocessing.
\newblock \emph{{Quantum}}, 7:\penalty0 937, March 2023.
\newblock ISSN 2521-327X.
\newblock \doi{10.22331/q-2023-03-02-937}.
\newblock URL \url{https://doi.org/10.22331/q-2023-03-02-937}.

\bibitem[Babbush et~al.(2018)Babbush, Gidney, Berry, Wiebe, McClean, Paler,
  Fowler, and Neven]{babbush2018encoding}
Ryan Babbush, Craig Gidney, Dominic~W Berry, Nathan Wiebe, Jarrod McClean,
  Alexandru Paler, Austin Fowler, and Hartmut Neven.
\newblock Encoding electronic spectra in quantum circuits with linear t
  complexity.
\newblock \emph{Physical Review X}, 8\penalty0 (4):\penalty0 041015, 2018.
\newblock \doi{10.1103/physrevx.8.041015}.

\bibitem[Giovannetti et~al.(2008)Giovannetti, Lloyd, and
  Maccone]{giovannetti2007QuantumRAM}
Vittorio Giovannetti, Seth Lloyd, and Lorenzo Maccone.
\newblock Quantum random access memory.
\newblock \emph{Physical Review Letters}, 100\penalty0 (16):\penalty0 160501, 2008.
\newblock \doi{10.1103/PhysRevLett.100.160501}.
\newblock arXiv: \href{https://arxiv.org/abs/0708.1879}{\ttfamily{0708.1879}}.

\bibitem[Draper et~al.(2006)Draper, Kutin, Rains, and
  Svore]{draper2004logarithmic}
Thomas~G Draper, Samuel~A Kutin, Eric~M Rains, and Krysta~M Svore.
\newblock A logarithmic-depth quantum carry-lookahead adder.
\newblock \emph{Quantum Information and Computation}, 6\penalty0
  (4--5):\penalty0 351--369, 2006.
\newblock \doi{10.26421/qic6.4-5-4}.

\bibitem[Chakrabarti et~al.(2021)Chakrabarti, Krishnakumar, Mazzola,
  Stamatopoulos, Woerner, and Zeng]{Chakrabarti2021thresholdquantum}
Shouvanik Chakrabarti, Rajiv Krishnakumar, Guglielmo Mazzola, Nikitas
  Stamatopoulos, Stefan Woerner, and William~J. Zeng.
\newblock A {T}hreshold for {Q}uantum {A}dvantage in {D}erivative {P}ricing.
\newblock \emph{{Quantum}}, 5:\penalty0 463, June 2021.
\newblock ISSN 2521-327X.
\newblock \doi{10.22331/q-2021-06-01-463}.
\newblock URL \url{https://doi.org/10.22331/q-2021-06-01-463}.

\bibitem[Low et~al.(2024)Low, Kliuchnikov, and Schaeffer]{low2018trading}
Guang~Hao Low, Vadym Kliuchnikov, and Luke Schaeffer.
\newblock Trading {T} gates for dirty qubits in state preparation and unitary
  synthesis.
\newblock \emph{Quantum}, 8:\penalty0 1375, 2024.
\newblock \doi{10.22331/q-2024-06-17-1375}.

\bibitem[Lin and Tong(2020{\natexlab{b}})]{lin2019OptimalQEigenstateFiltering}
Lin Lin and Yu~Tong.
\newblock Optimal polynomial based quantum eigenstate filtering with
  application to solving quantum linear systems.
\newblock \emph{{Quantum}}, 4:\penalty0 361, 2020{\natexlab{b}}.
\newblock \doi{10.22331/q-2020-11-11-361}.
\newblock arXiv: \href{https://arxiv.org/abs/1910.14596}{\ttfamily{1910.14596}}.

\bibitem[Martyn et~al.(2021)Martyn, Rossi, Tan, and Chuang]{martyn2021QSVT}
John~M. Martyn, Zane~M. Rossi, Andrew~K. Tan, and Isaac~L. Chuang.
\newblock Grand unification of quantum algorithms.
\newblock \emph{PRX Quantum}, 2:\penalty0 040203, Dec 2021.
\newblock \doi{10.1103/PRXQuantum.2.040203}.
\newblock URL \url{https://link.aps.org/doi/10.1103/PRXQuantum.2.040203}.

\bibitem[Gily{\'e}n(2014)]{gilyen2014mastersthesis}
Andr{\'a}s~P{\'a}l Gily{\'e}n.
\newblock \emph{Quantum walk based search methods and algorithmic
  applications}.
\newblock PhD thesis, MSc Thesis, E{\"o}tv{\"o}s Lor{\'a}nd University, 2014.

\bibitem[Gidney(2020)]{web:GidneyStackExchange}
C~Gidney.
\newblock Quantum computing stack exchange, 2020.
\newblock URL
  \url{https://quantumcomputing.stackexchange.com/questions/11734/what-is-the-complexity-of-splitting-a-state-into-a-superposition-of-n-computat}.

\bibitem[Pallister(2022)]{web:PallisterStackExchange}
S~Pallister.
\newblock Quantum computing stack exchange, 2022.
\newblock URL
  \url{https://quantumcomputing.stackexchange.com/questions/27864/creating-a-uniform-superposition-of-a-subset-of-basis-states}.

\end{thebibliography}
\end{document}